\documentclass[11pt,letterpaper]{article}
\usepackage[utf8]{inputenc}

\usepackage{amsmath,amssymb,amsfonts} 
\usepackage{epsfig} \usepackage{latexsym,nicefrac,bbm}
\usepackage{xspace}
\usepackage{color,fancybox,graphicx,subfigure,fullpage}
\usepackage[top=1in, bottom=1in, left=1in, right=1in]{geometry}
\usepackage{tabularx} \usepackage{hyperref} 
\usepackage{pdfsync}
\usepackage[boxruled, linesnumbered]{algorithm2e}
\usepackage{multicol}
\usepackage{enumitem}
\usepackage{multirow}
 \usepackage{url}

\renewcommand{\epsilon}{\varepsilon}

\newcommand{\nfrac}{\nicefrac}
\newcommand{\eps}{\varepsilon}

\newtheorem{theorem}{Theorem}[section]

\newtheorem{definition}{Definition}[section]

\newtheorem{lemma}[theorem]{Lemma}

\newtheorem{proposition}[theorem]{Proposition}
\newtheorem{corollary}[theorem]{Corollary}

\newenvironment{proof}{\begin{trivlist} \item {\bf Proof:~~}}
   {\qed\end{trivlist}}

\def\FullBox{\hbox{\vrule width 6pt height 6pt depth 0pt}}

\def\qed{\ifmmode\qquad\FullBox\else{\unskip\nobreak\hfil
\penalty50\hskip1em\null\nobreak\hfil\FullBox
\parfillskip=0pt\finalhyphendemerits=0\endgraf}\fi}

\usepackage[T1]{fontenc}
\usepackage{lmodern}

\SetKwInOut{KwHyperparameters}{Hyperparameters}

\renewcommand{\textstyle}{}

\title{\bf  Private  Covariance Approximation and Eigenvalue-Gap Bounds for Complex Gaussian Perturbations}

 \author{Oren Mangoubi\\ Worcester Polytechnic Institute \and Nisheeth K. Vishnoi \\ Yale University}

\begin{document}
 \date{}
\maketitle

\begin{abstract}
We consider the problem of  approximating a $d \times d$ covariance matrix $M$ with a rank-$k$ matrix  under $(\eps,\delta)$-differential privacy. 
We present and analyze a complex variant of the Gaussian mechanism and show that the Frobenius norm of the difference between the matrix output by this mechanism and the best rank-$k$ approximation to $M$ is bounded by roughly $\tilde{O}(\sqrt{kd})$,  whenever there is an appropriately large gap between the $k$'th and the $k+1$'th eigenvalues of $M$. This improves on previous work that requires that the gap between every pair of top-$k$ eigenvalues of $M$ is at least $\sqrt{d}$ for a similar bound.
Our analysis leverages the fact that the eigenvalues of  complex matrix Brownian motion repel more than in the real case,
and uses Dyson's stochastic differential equations governing the evolution of its eigenvalues to show that  
the eigenvalues of the matrix $M$ perturbed by complex Gaussian noise have large gaps with high probability.
Our results contribute to the analysis of low-rank approximations under average-case perturbations and to an understanding of eigenvalue gaps for random matrices, which may be of independent interest.

\end{abstract}

\newpage

\tableofcontents

 \newpage

\section{Introduction}

Given a  matrix $M \in \mathbb{R}^{d \times d}$,  consider the following basic problem of finding a rank-$k$ matrix $X$ that is closest to $M$ in Frobeinus norm   \cite{bhatia2013matrix,blum2020foundations}:
$\textstyle \min_{X: \mbox{ rank} (X)\leq k} \|M-X\|_F.$
Of interest is the case when $M$ is the covariance matrix of a data matrix: Given a 
matrix $A \in \mathbb{R}^{m \times d}$, consisting of $m$ individuals with $d$-dimensional features, $M=A^\top A$.
Such an $M$ is positive semi-definite (PSD) and has non-negative eigenvalues $\sigma_1 \geq \cdots \geq \sigma_d \geq 0$.
The  solution to the rank-$k$  optimization problem above is well-known \cite{bhatia2013matrix}: It is given by 
$M_k:= V \Gamma_k V^\top$  where $\Gamma_k := \mathrm{diag}(\sigma_1,\ldots, \sigma_k, 0,\ldots, 0)$
and $V$ is a matrix whose columns are the eigenvectors of $M$.

In several modern applications of this low-rank approximation problem, 
the rows of $A$ encode  sensitive features of  individuals and the release of a low-rank approximation to $M$  may reveal these features; see \cite{bennett2007netflix}. 
In such contexts, differential privacy (DP) has been employed to quantify the extent to which an algorithm preserves privacy   \cite{dwork2006calibrating} and, in particular, algorithms for low-rank  covariance matrix approximation under differential privacy  have been widely studied    \cite{blum2005practical, dwork2006calibrating, kapralov2013differentially, blocki2012johnson,  dwork2014analyze,  upadhyay2018price,  sheffet2019old,mangoubi2022private, DBM_Neurips}.
In the low-rank approximation problem, a randomized mechanism $\mathcal{A}$ is said to be $(\eps, \delta)$-differentially private for  privacy parameters $\eps, \delta \geq 0$  if for all ``neighboring'' matrices $M, M' \in  \mathbb{R}^{d \times d}$, and any measurable subset $S$ of the range of $\mathcal{A}$, we have 
\begin{equation}\label{DP}
\textstyle \mathbb{P}(\mathcal{A}(M) \in S) \leq e^\eps \mathbb{P}(\mathcal{A}(M') \in S) + \delta.
\end{equation}
$M$ and $M'$ are said to be neighbors if their corresponding data matrices $A,A' \in \mathbb{R}^{m \times d}$ differ by at most one row ($M'=M-uu^\top+vv^\top$ where, for each row vector, $\|u\|_2,\|v\|_2\leq 1$).
To measure the ``utility'' of the mechanism, the   Frobenius-norm distance $\|\mathcal{A}(M) - M_k\|_F$ is often used; see  \cite{chaudhuri2012near,dwork2014analyze, amin2019differentially, DBM_Neurips} 
for a discussion on the rationale for this norm. 
This leads to the problem of designing an $(\eps,\delta)$-differentially private mechanism that, given a covariance matrix $M$ with eigenvalues $\sigma_1 \geq \cdots \geq \sigma_d \geq 0$, outputs a rank-$k$ matrix $Y$ that minimizes $\| Y - M_k\|_F$.
 \cite{dwork2014analyze} present the (real) Gaussian mechanism which  ensures $(\eps,\delta)$-DP by adding a real symmetric matrix $E$ with i.i.d. Gaussian entries from $\textstyle N(0,\nfrac{\sqrt{\log\nfrac{1}{\delta}}}{\eps})$ to $M$ and then outputting the Frobenius-norm minimizing rank-$k$ approximation to $M+E$. 
 Roughly speaking, they show that the output $Y$ of their  mechanism satisfies  $\| M-Y\|_F - \|M-M_k\|_F = \tilde{O}(k\sqrt{d})$ w.h.p.

\cite{DBM_Neurips} also consider the Gaussian mechanism and prove that, if  the top-$k$ eigenvalue gaps of $M$ satisfy $\sigma_i - \sigma_{i+1}   \geq \tilde{\Omega}(\sqrt{d})$ for every $i \leq k$, then the utility bound (for a stronger metric $\|Y- M_k\|_F$) can be improved by a factor of $\sqrt{k}$ to, roughly,  $\|Y- M_k\|_F \leq \tilde{O}(\sqrt{kd})$ in expectation whenever the $k$'th eigenvalue gap of $M$ is at least  $\sigma_k - \sigma_{k+1} \geq \Omega(\sigma_k)$. 
The assumption of a large $k$'th eigenvalue gap is common in the matrix approximation literature, as a large gap $\sigma_k - \sigma_{k+1}$ in the eigenvalues of the covariance matrix $M$ for some $k < d$ can motivate the problem of finding a rank-$k$ approximation
because it suggests the presence of a useful rank-$k$ ``signal'' $M_k$ which one wishes to extract from the background ``noise'' in the data (see e.g. \cite{dwork2014analyze}).
Such a gap is also necessary for good rank-$k$ approximations to exist under the stronger utility metric $\|Y- M_k\|_F$.
However, in many datasets where one has a large $k$'th eigenvalue gap for some $k<d$, the other gaps in the top-$k$ eigenvalues are oftentimes small or even zero (for instance, this happens whenever two or more of the features in a dataset are highly correlated).
Thus, \cite{DBM_Neurips} left as an open problem to investigate if the assumption that {\em all} the top-$k$ eigenvalues of $M$ have large gaps can be removed.

\noindent
\textbf{Our contributions.} We show that a {\em complex}  Gaussian mechanism (Algorithm \ref{alg_quaternion_Gaussian}) can give a utility bound $\|Y- M_k\|_F \leq \tilde{O}(\sqrt{kd})$ {\em without} assuming that the gaps in all of the top-$k$ eigenvalues of $M$ are at least $\tilde{\Omega}(\sqrt{d})$; see Theorem \ref{thm_utility}.
As in \cite{DBM_Neurips}, we view the addition of Gaussian noise $B(t)$ as a stochastic process $M + B(t)$, whose eigenvalues $\gamma_i(t)$ and eigenvectors evolve according to the  stochastic differential equations (SDE) discovered by \cite{dyson1962brownian}; see  \eqref{eq_DBM_matrix}. 
This leads to a bound on the utility which includes terms of the form $\nfrac{(\lambda_i - \lambda_{i+1})^2}{(\gamma_i(t) -\gamma_{i+1}(t))^{2}}$  and $\nfrac{(\lambda_i - \lambda_{i+1})}{(\gamma_i(t) -\gamma_{i+1}(t))^{2}}$ integrated over time, where, roughly speaking, $\lambda_1,\ldots, \lambda_k$ are the eigenvalues of the rank-$k$ approximation and the $\gamma_i(t)$ are the eigenvalues of the matrix $M+B(t)$.
If one does not assume that the initial gaps are $\tilde{\Omega}(\sqrt{d})$, the gaps $\gamma_i(t)-\gamma_{i+1}(t)$ may become very small at some times $t$, causing the terms in the utility bound to become very large.
To bypass this,  the following novel steps are employed here:
1) Rather than analyzing the utility by considering the output eigenvalues $\lambda_i$ to be fixed numbers, we instead set the top-$k$ output eigenvalues $\lambda_i$ to be {\em dynamically} changing over time and equal to $\gamma_i(t)$, making the gaps in the numerators small at exactly those times when the denominators are small.
2) We then leverage the fact that our mechanism adds {\em complex} Gaussian noise, which implies that $\gamma_i(t)$s evolve by repelling each other with a stronger ``force'' than when only real noise is added, to show that the gaps between the eigenvalues satisfy a high-probability lower bound of $\mathbb{P}(\gamma_i(t) - \gamma_{i+1}(t) \leq  \nfrac{s}{\sqrt{td}}) \leq \tilde{O}(s^3)$;  see Lemma \ref{lemma_GUE_gaps}.
Our bound improves, in the setting where the random matrix is Gaussian, on previous eigenvalue gap bounds of \cite{nguyen2017random} where the bound on the probability decays as $O(s^2)$, which is insufficient for our application.
We prove Lemma \ref{lemma_GUE_gaps}  by first showing that one can reduce the problem of bounding the gaps $\gamma_i(t)-\gamma_{i+1}(t)$ to the  special case when the initial eigenvalues are all zero; see Lemma \ref{lemma_gap_comparison}.

\section{Main results}

For an $S \in \mathbb{C}^{d \times d}$, denote by $S^\ast$ its conjugate transpose. 
$S$ is Hermitian if $S=S^\ast$.
For any Hermitian matrix $S \in \mathbb{C}^{d \times d}$ with spectral decomposition $S = U \Lambda U^\ast$ where $\Lambda := \mathrm{diag}(\lambda_1,\ldots, \lambda_d)$ is a diagonal matrix containing the eigenvalues $\lambda_1 \geq \cdots \geq \lambda_d$ of $S$ and $U  := [u_1, \ldots, u_d]$ a unitary matrix containing the eigenvectors $u_1, \ldots, u_d$ of $S$, denote by $\Lambda_k := \mathrm{diag}(\lambda_1,\ldots, \lambda_k, 0,\ldots,0)$ and by $S_k := U\Lambda_k U^\ast$ the best rank-$k$ approximation of $S$.
Denote by $U_k := [u_1, \ldots, u_k]$ the $d \times k$ matrix of the top-$k$ eigenvectors of $U$.

\subsection{Private covariance approximation and complex random perturbations} \label{sec_our_results_1}

Our first result (Theorem \ref{thm_utility}) analyzes  the complex Gaussian mechanism (Algorithm \ref{alg_quaternion_Gaussian})  and provides an upper bound on the expected Frobenius distance utility of this mechanism for the problem of rank-$k$ covariance approximation.
In the following, the $\tilde{O}$ notation hides factors of $(\log d)^{\log \log d}$.

\begin{algorithm} \label{alg_quaternion_Gaussian}
\caption{Complex Gaussian Mechanism} \label{alg_general_self_concordant}
\KwIn{$\epsilon, \delta >0$, $d, k \in \mathbb{N}$. A real symmetric PSD matrix $M \in \mathbb{R}^{d \times d}$.}

 \KwOut{A real symmetric matrix $Y \in \mathbb{R}^{d\times d}$.}

Sample matrices $W_1, W_2 \in \mathbb{R}^{d \times d}$ with i.i.d. $N(0,1)$ entries.

Set $G := (W_1 +  \mathfrak{i} W_2) + (W_1 +  \mathfrak{i} W_2)^\ast$

Set $\hat{M} := M + \sqrt{T}G$,  where $T := \frac{2\log\frac{1.25}{\delta}}{\eps^2}$.

Compute the diagonalization $\hat{M} = \hat{V} \hat{\Sigma} \hat{V}^\ast$ with eigenvalues $\hat{\sigma}_1 \geq \cdots \geq \hat{\sigma}_d$

Set $\hat{M}_k := \hat{V} \hat{\Sigma}_k \hat{V}^\ast$,  where $\hat{\Sigma}_k := \mathrm{diag}(\hat{\sigma}_1, \ldots,  \hat{\sigma}_k, 0,\ldots, 0)$

Output  $Y$ to be the real part of $\hat{M}_k$

\end{algorithm}

\begin{theorem}[$(\varepsilon, \delta)$-DP rank-$k$ covariance approximation] \label{thm_utility}
Given $\eps,\delta>0$, there is an $(\varepsilon, \delta)$-differentially private algorithm (Algorithm \ref{alg_quaternion_Gaussian}) that, on input $k>0$ and a real symmetric matrix $M \in \mathbb{R}^{d \times d}$ with eigenvalues $\sigma_1 \geq \cdots \geq \sigma_d \geq 0$
satisfying 
$\sigma_k - \sigma_{k+1} \geq  4\frac{\sqrt{2\log\frac{1.25}{\delta}}}{\eps} \sqrt{d}$
 and $\sigma_1 \leq d^{50}$, 
outputs a rank-$k$ matrix $Y \in \mathbb{R}^{d \times d}$ such that  
$$\textstyle \sqrt{\mathbb{E}[\|M_k - Y\|_F^2]}  \leq \tilde{O}\left(\sqrt{kd} \frac{\sigma_k}{\sigma_k - \sigma_{k+1}} \times  \frac{\sqrt{\log\frac{1}{\delta}}}{\eps}\right).$$
$M_k$ 
is the Frobenius-norm minimizing rank-$k$ approximation to $M$.
\end{theorem}
\noindent
The  proof of Theorem \ref{thm_utility} appears   in Section \ref{sec_proof_utility_privacy}.
We note that the requirement in Theorem \ref{thm_utility} that  $\sigma_1 \leq d^{50}$ is likely an artifact of the proof, and can be replaced with $\sigma_1 \leq d^{C}$ for any large universal constant $C>0$.

Theorem \ref{thm_utility} only requires a lower bound $\sigma_k - \sigma_{k+1} 
 \geq \tilde{\Omega}\left(\frac{\sqrt{d\log \frac{1}{\delta}}}{\epsilon}\right)$ on the $k$'th eigenvalue gap. 
  Thus, it improves on the main result of \cite{DBM_Neurips} (their Corollary 2.3) as it no longer relies on their assumption (Assumption 2.1) that all the gaps in the top-$k$ eigenvalues of $M$ are at least $\sigma_i- \sigma_{i+1} \geq \tilde{\Omega}\left(\frac{\sqrt{d \log\frac{1}{\delta}}}{\epsilon}\right)$ for all $i \leq k$. 

Moreover, for matrices $M$ whose  $k$'th eigengap  satisfies $\sigma_k - \sigma_{k+1} = \Omega(\sigma_k)$, Theorem \ref{thm_utility} improves by a factor of $\sqrt{k}$ on the bound in Theorem 7 of \cite{dwork2014analyze} which says  the output $Y$ of their mechanism satisfies  $\|Y - M\|_F - \|M_k - M\|_F = \tilde{O}(k\sqrt{d})$ w.h.p.
This is because an upper bound on $\| Y - M_k\|_F$ implies an upper bound on their utility measure by the triangle inequality.
The reason why \cite{dwork2014analyze} is independent of the gap $\sigma_k - \sigma_{k+1}$ while our bound depends on the  ratio $\frac{\sigma_k}{\sigma_k -\sigma_{k+1}}$  is due to the fact that 
if, e.g., $\sigma_k - \sigma_{k+1} = 0$ an arbitrarily small Gaussian perturbation to $M$ would lead to a perturbation of $\|\hat{V}_k - V_k\|_2 =\Omega(1)$ w.h.p., where $\hat{V}_k$ and $V_k$ are the matrices containing the top-$k$ eigenvectors of $\hat{M}$ and $M$ respectively.
 Roughly speaking, this, in turn, would lead to a perturbation of at least $\|Y - M_k \|_F  \geq \|\sigma_k\hat{V}_k \hat{V}_k^\ast - \sigma_k V_k V_k^\ast\|_2 \geq \Omega(  \sigma_k)$.
The techniques used in the proof of Theorem \ref{thm_utility} can also be used to improve this Frobenius utility to $\tilde{O}(\sqrt{kd})$ {\em without} assuming the eigengap condition; see Section \ref{appendix_gap_free_bounds_in_weaker_metric}.
Finally, the expectation bound in Theorem \ref{thm_utility} immediately implies a high probability bound with polynomial decay in the probability via Chebyshev's inequality; see also the discussion at the end of Section \ref{sec_technical_gaps}.

The privacy guarantee in Theorem \ref{thm_utility} follows directly from prior works on the (real) Gaussian mechanism (see Section \ref{sec_proof_utility_privacy})
The utility bound in Theorem \ref{thm_utility} follows from the following ``average-case'' matrix perturbation bound for complex Gaussian random perturbations.
Theorem \ref{thm_rank_k_covariance_approximation_new} may be of independent interest to other applications where matrix perturbation bounds are used.

\begin{theorem}[Frobenius bound for complex Gaussian perturbations] \label{thm_rank_k_covariance_approximation_new}
Suppose we are given $k>0$, $T>0$, and a  Hermitian matrix  $M \in \mathbb{C}^{d \times d}$ with eigenvalues  $\sigma_1 \geq \cdots \geq \sigma_d \geq 0$.   
Let $\hat{M} := M + \sqrt{T}[(W_1 + \mathfrak{i}W_2) + (W_1 + \mathfrak{i}W_2)^\ast]$ where $W_1, W_2 \in \mathbb{R}^{d \times d}$ have entries which are independent $N(0,1)$ random variables.
 Denote, respectively, by  $\sigma_1 \geq \cdots \geq \sigma_d$ and  $\hat{\sigma}_1\geq \ldots \geq \hat{\sigma}_d \geq 0$ the eigenvalues of $M$ and $\hat{M}$, and by $V$ and  $\hat{V}$ the matrices whose columns are the corresponding eigenvectors of $M$ and $\hat{M}$.
Moreover, let  $M_k := V \Gamma_k V^\ast$ and  $\hat{M}_k := \hat{V} \hat{\Gamma}_k \hat{V}^\ast$ be the rank-$k$ approximations of $M$ and $\hat{M}$, where $\Gamma_k := \mathrm{diag}(\sigma_1,\ldots, \sigma_k,0,\ldots,0)$  and $\hat{\Gamma}_k := \mathrm{diag}(\hat{\sigma}_1,\ldots, \hat{\sigma}_k,0,\ldots,0)$.
Suppose that 
 $\sigma_k - \sigma_{k+1} \geq 4\sqrt{Td} $ and that $\sigma_1 \leq d^{50}$.
Then we have
$$ \textstyle
\sqrt{\mathbb{E}\left[\left\|\hat{M}_k -  M_k\right \|_F^2\right]} \leq \tilde{O}\left(\sqrt{kd} \frac{\sigma_k}{\sigma_k - \sigma_{k+1}} \times \sqrt{T}\right).
$$
\end{theorem}
 The proof of Theorem \ref{thm_rank_k_covariance_approximation_new} is presented in Section \ref{sec_utility}.
  We do not know if our bound in Theorem  \ref{thm_rank_k_covariance_approximation_new} is tight for every $M$ and $k$, however, one can easily check that it is tight when $k=d$.
  In this case, one can plug in $\sigma_{d+1} = 0$ to the r.h.s. of our utility bound which gives a bound of $\sqrt{\mathbb{E}[\|\hat{M}-M\|_F^2]} \leq \tilde{O}(d  \sqrt{T})$.
 Since $\hat{M}-M = (G + G^\ast) \times \sqrt{T}$ where $G$ has  complex Gaussian entries, we have that $\|\hat{M}-M\|_F = \Theta(d  \sqrt{T})$ w.h.p. from standard matrix concentration bounds.
 Thus, in this case, our bound is tight up to factors of $(\log d)^{\log \log d}$ hidden in the $\tilde{O}$ notation.

\subsection{Eigenvalue gaps}\label{sec_technical_results}

To prove Theorem \ref{thm_rank_k_covariance_approximation_new}, we 
 view the addition of complex Gaussian noise to the matrix $M$ as a matrix-valued Brownian motion.
 Towards this end, let $W(t) \in  \mathbb{C}^{d \times d}$ be a matrix where the real part and complex part of each entry is an independent standard Brownian motion with distribution $N(0, tI_d)$ at time $t$, and let $B(t) := W(t) + W(t)^\ast$.
Define the Hermitian-matrix valued stochastic process $\Phi(t)$ as follows:
\begin{equation} \label{eq_DBM_matrix}
\textstyle    \Phi(t):= M + B(t) \qquad  \forall t\geq 0.
\end{equation}
At every time $t>0$, the eigenvalues $\gamma_1(t), \ldots, \gamma_d(t)$ of $\Phi(t)$ are real-valued and  distinct w.p. $1$, and \eqref{eq_DBM_matrix} induces a stochastic process on the eigenvalues and eigenvectors.
The  evolution of the eigenvalues 
can be expressed by the following stochastic differential equations  (SDE) \cite{dyson1962brownian}: 
\begin{equation} \label{eq_DBM_eigenvalues}
\textstyle    \mathrm{d} \gamma_i(t) = \mathrm{d}B_{i i}(t) +  \beta \sum_{j \neq i} \frac{1}{\gamma_i(t) - \gamma_j(t)} \mathrm{d}t \qquad \qquad \forall i \in [d], t > 0,
\end{equation}
where the parameter $\beta=2$ for the complex case ($\beta=1$ for the real matrix Brownian motion) (Figure \ref{fig_DBM}).
It is well known that, with probability $1$, a solution to \eqref{eq_DBM_eigenvalues} exists and is unique when coupled to the underlying Brownian motion $B(t)$.
Moreover, the paths traversed by the eigenvalues are continuous on all $t \in [0,\infty)$ and do not intersect at any time $t>0$ \eqref{eq_DBM_eigenvalues} (see e.g. \cite{anderson2010introduction, inukai2006collision, rogers1993interacting}).

\begin{figure}
    \centering
    \includegraphics[width=0.45\textwidth]{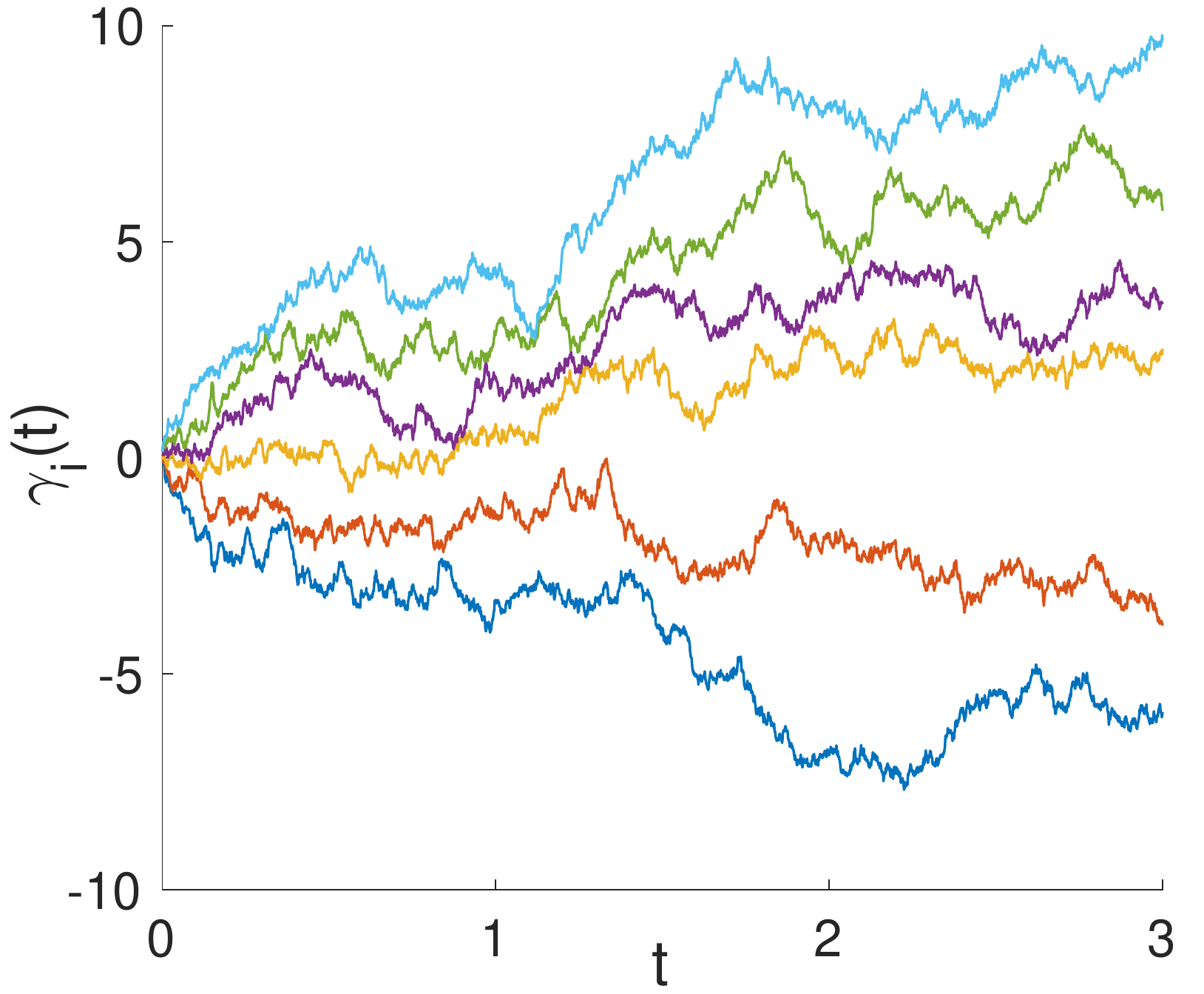}
    \includegraphics[width=0.45\textwidth]{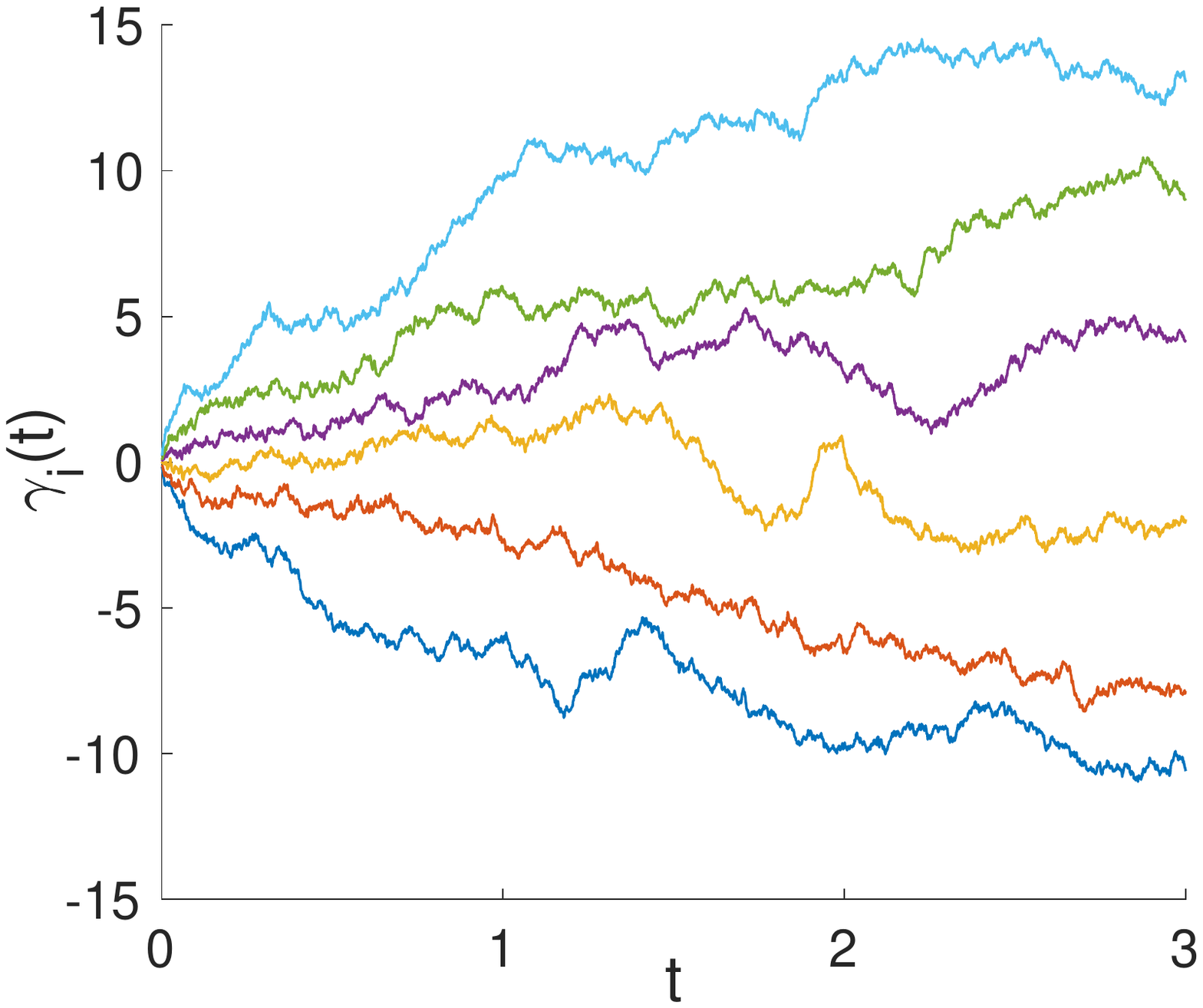}
    \caption{One run of a simulation of the eigenvalues $\gamma_1(t) \geq \cdots \geq \gamma_d(t)$ of Dyson Brownian, in the real case (left) and the complex case (right) with initial condition $\gamma_1(0) = \cdots = \gamma_d = 0$, for $d=6$.
    In the complex case, eigenvalue repulsion is stronger and the gaps between the eigenvalues are not as small as in the real case.
    }
    \label{fig_DBM}
\end{figure}

The corresponding eigenvector process $u_1(t), \ldots, u_d(t)$, referred to as the Dyson vector flow, is also a ``diffusion'' and,  conditional on the eigenvalue process \eqref{eq_DBM_eigenvalues}, is given by the following SDEs:
\begin{equation}\label{eq_DBM_eigenvectors}
\textstyle    \mathrm{d}u_i(t) = \sum_{j \neq i} \frac{\mathrm{d}B_{ij}(t)}{\gamma_i(t) - \gamma_j(t)}u_j(t) - \frac{\beta}{2}\sum_{j \neq i} \frac{\mathrm{d}t}{(\gamma_i(t)- \gamma_j(t))^2}u_i(t) \qquad \qquad \forall i \in [d], t > 0.
\end{equation}
We use \eqref{eq_DBM_eigenvectors} to track the utility over time.
Letting $\Phi(t) = U(t) \Gamma(t) U(t)^\ast$  be a spectral decomposition of the Hermitian matrix $\Phi(t)$ at every time $t$ where $\Gamma(t)$ is a diagonal matrix of eigenvalues at time $t$ and $U(t)$ a unitary matrix of eigenvectors.
We now define the rank-$k$ matrix $\Psi(t)$ to be the Hermitian matrix with any eigenvalues $\lambda_1(t) \geq \cdots  \geq\lambda_d(t)$, where $\lambda_i(t) = \gamma_i(t)$ for $i\leq k$ and  $\lambda_i(t) =0$ for $i>k$,   and with eigenvectors $U(t)$:  
 $  \Psi(t):= U(t) \Lambda(t) U(t)^\ast$ for all $t \in [0,T],$
 where $\Lambda(t) :=\mathrm{diag}(\lambda_1(t), \ldots, \lambda_d(t))$.

Using \eqref{eq_DBM_eigenvectors} to obtain  SDEs for $\Psi(t)$, and integrating these SDEs, we obtain a formula for the utility $\mathbb{E}[\|\hat{M}_k - M_k\|_F^2] = \mathbb{E}[\|\Psi(T) - \Psi(0)\|_F^2]$ (Lemma \ref{Lemma_integral}): $\mathbb{E}[\|\hat{M}_k - M_k\|_F^2] =$
\begin{eqnarray}\label{eq_o1}
\textstyle
 O
\left(\sum_{i=1}^{d} \int_{0}^{T}  \sum_{j \neq i}  \mathbb{E}\left[  \frac{(\lambda_i(t) - \lambda_j(t))^2}{(\gamma_i(t)-\gamma_j(t))^2} \right]  +    T \mathbb{E}\left[\left(\sum_{j\neq i} \frac{\lambda_i(t) - \lambda_j(t)}{(\gamma_i(t)-\gamma_j(t))^2}\right)^2   \right]\mathrm{d}t\right).
\end{eqnarray}
After simplifying \eqref{eq_o1}, we are left with terms which are a time-integral of  $\mathbb{E}\left[\frac{1}{(\gamma_i(t)-\gamma_j(t))^2}\right]$.
To bound these terms, we wish to show that at any time $t$ the gaps $\gamma_i(t)-\gamma_j(t)$ of Dyson Brownian motion are large with high probability.

We first show, in the following lemma, that one can reduce this task to the problem of bounding the gaps of a Dyson Brownian motion initialized at the $0$ vector.
 In the following, we define $\mathcal{W}_d := \{(x_1, \ldots, x_d) \in \mathbb{R}^d:  x_1 \geq \cdots \geq x_d\}$.
\begin{lemma}[Eigenvalue-gap comparison Lemma]\label{lemma_gap_comparison}
\quad Let $\xi(t) = (\xi_1(t), \ldots, \xi_d(t))$ \, \, and \, \, $\gamma(t) =$ \\ $(\gamma_1(t), \ldots, \gamma_d(t))$ be two solutions of \eqref{eq_DBM_eigenvalues} coupled to the same underlying Brownian motion $B(t)$, starting respectively from initial conditions $\xi(0), \gamma(0)$. 
Assume that $\xi_i(0) - \xi_{i+1}(0)  \leq \gamma_i(0) - \gamma_{i+1}(0)$ for all $1\leq i < d$.
Then, with probability $1$, $\xi_i(t) - \xi_{i+1}(t)  \leq \gamma_i(t) - \gamma_{i+1}(t)$ for all  $t>0$ and all $1\leq i < d$.
\end{lemma}
 The proof of Lemma \ref{lemma_gap_comparison} appears in Section \ref{sec_gap_comparison_proof} and an overview appears in Section \ref{sec_technical_overview}.
 \cite{anderson2010introduction} show a different eigenvalue comparison theorem (their Lemma 4.3.6) which says that if $\xi$ and $\gamma$ are two coupled Dyson Brownian motions with initial conditions satisfying $\xi_i(0) \leq \gamma_i(0)$ for all $i \in [d]$, then with probability $1$,  $\xi_i(t) \leq \gamma_i(t)$ at every $t \geq 0$.  However, this does not imply the gaps of $\gamma_i(t)$ are at least as large as the corresponding gaps of $\xi_i(t)$ since we could have that $\gamma_i(t) - \gamma_{i+1}(t) < \xi_i(t)- \xi_{i+1}(t)$ even if $\xi_i(t) \leq \gamma_i(t)$ for all $i$; see also \cite{erdHos2011universality, landon2017convergence,lee2016bulk} for results about the eigenvalues of Dyson Brownian motion and their gaps from non-zero initial conditions.

Lemma \ref{lemma_gap_comparison} implies that it is enough to show the following high probability lower bound on the gaps of the eigenvalues of the GUE random matrix.
\begin{lemma}[Eigenvalue gaps of Gaussian Unitary Ensemble (GUE)]\label{lemma_GUE_gaps}
Let $A := G + G^\ast$ where $G$ is a matrix with i.i.d. complex standard Gaussian entries, and denote by $\eta_1,\ldots, \eta_d$ the eigenvalues of $A$.
Then
$$
  \textstyle  \mathbb{P}\left(\eta_i - \eta_{i+1} \leq s \frac{1}{\mathfrak{b}\sqrt{d}}\right) \leq s^{3} + \frac{1}{d^{1000}}$$ for all $s>0$, and for all $1\leq i < d$,
where $\mathfrak{b} = (\log d)^{ L\log \log d}$ and $L$ is a universal constant.

\end{lemma}
The proof of Lemma \ref{lemma_GUE_gaps} is presented in Section \ref{section_GUE_proof} and an overview appears in Section \ref{sec_technical_overview}.
We note that the term $\frac{1}{d^{1000}}$ in Lemma \ref{lemma_GUE_gaps}  can be replaced by $\frac{1}{d^C}$ for any universal constant $C>0$.
Thus, Lemma \ref{lemma_GUE_gaps} says that for any $s> d^{-C}$ (where $C$ can be taken to be any large universal constant), the probability that any gap $\eta_i- \eta_{i+1}$ of the GUE random matrix is $\leq \tilde{O}\left(\frac{s}{\sqrt{d}}\right)$ is $O(s^3)$.
The $s^3$ dependence is important to our analysis of the Frobenius-distance utility in Theorem  \ref{thm_rank_k_covariance_approximation_new},  where we wish to bound the time-average of the second moment of the inverse gaps $\mathbb{E}\left[\frac{1}{(\gamma_i(t) - \gamma_j(t))^2}\right]$.
Lemma \ref{lemma_GUE_gaps} allows us to bound this term by $O(d)$.
We  use it to bound the utility in \eqref{eq_o1} by $O(kd)$, thus implying the bound in Theorem \ref{thm_rank_k_covariance_approximation_new}.

The distribution of the gaps of the GUE in the limit as $d \rightarrow \infty$ was studied e.g. in \cite{dyson1963random, tao2013asymptotic, arous2013extreme}, and was also studied non-asymptotically in e.g. \cite{nguyen2017random}.
However, to the best of our knowledge, we are not aware of a previous (non-asymptotic in $d$) lower bound on the gaps of GUE random matrices which scales as small as $O(s^3)$.
For instance, \cite{nguyen2017random}, which studies eigenvalue gaps of Wigner random matrices with sub-Gaussian tails--a more general class of random matrices which includes as a special case the GUE random matrices--show a bound of $\mathbb{P}\left(\eta_i- \eta_{i+1} \leq \frac{s}{\sqrt{d}}\right) \leq O(s^2)$ for any  $s> d^{-C}$ where $C>0$ is a universal constant (Corollary 2.2 in \cite{nguyen2017random}, which they can extend to the complex case).
On the other hand, we note that \cite{nguyen2017random} focus on matrix universality results that apply to a larger class of  random matrices than the GUE random matrices,
and that our bound includes additional factors of $(\log d)^{\log \log d}$ hidden in the $\tilde{O}$ notation.

 Finally, we note that while many results in the random matrix literature rely on explicit determinantal formulas that are only available for complex-valued random matrices (see e.g. \cite{ratnarajah2004eigenvalues, johansson2005random, leake2020polynomial}),  Lemma \ref{lemma_gap_comparison} applies to both the real and complex versions of Dyson Brownian motion, and the proofs of Lemma \ref{lemma_GUE_gaps} and Corollary \ref{lemma_gaps_any_start} can be extended to the real case with minor modifications.
The main difference is that, for the real case, we would get a term $s^2$  (which is $= s^{\beta+1}$ for $\beta=1$) on the r.h.s. of Lemma \ref{lemma_GUE_gaps}  in place of the term $s^3$ 
which appears in the complex version of these results.

\section{Preliminaries}

\subsection{Brownian motion and It\^o calculus}

In this section, we give preliminaries on Brownian motion and Stochastic calculus (also referred to as It\^o calculus).
A Brownian motion $W(t)$ is a continuous process that has stationary
independent increments  (see e.g., \cite{morters2010brownian}).
In a  multi-dimensional Brownian motion,  each coordinate is an independent and identical Brownian motion. 
The filtration $\mathcal{F}_t$  generated by  $W(t)$ is defined as $\sigma \left(\cup_{s \leq t} \sigma(W(s))\right)$, where $\sigma(\Omega)$ is the $\sigma$-algebra generated by $\Omega$.
$W(t)$ is a martingale with respect to $\mathcal{F}_t$.

\begin{definition}[\bf It\^o Integral]
Let $W(t)$ be a Brownian motion for $t \geq 0$, let $\mathcal{F}_t$ be the filtration generated by $W(t)$, and let $z(t) :  \mathcal{F}_t \rightarrow \mathbb{R}$ be a stochastic process adapted to $\mathcal{F}_t$.
The It\^o integral is defined as 
$$\int_0^T z(t) \mathrm{d}W(t) := \lim_{\omega \rightarrow 0} \sum_{i=1}^{\frac{T}{\omega}} z(i\omega)\times[W((i+1)\omega) -W(i\omega)].$$   
\end{definition}

\noindent
The following lemma (It\^o's Lemma) generalizes the chain rule of deterministic derivatives to stochastic derivatives and allows one to compute the derivative of a function $f(X(t))$ of a stochastic process $X(t)$.
We state It\^o's Lemma in its integral form:
\begin{lemma}[\bf It\^o's Lemma, integral form with no drift; Theorem 3.7.1 of \cite{lawler2010stochastic}] \label{lemma_ito_lemma_new}
Let $f:  \mathbb{R}^n \rightarrow \mathbb{R}$ be any twice-differentiable function.
Let $W(t) \in \mathbb{R}^n$  be a Brownian motion, and let $X(t)  \in \mathbb{R}^n$ be an It\^o diffusion process with mean zero defined by the following stochastic differential equation:
\begin{equation}
\mathrm{d}X_j(t) = \sum_{i=1}^d R_{i j }(t) \mathrm{d}W_i(t),
\end{equation}
for some It\^o diffusion $R(t) \in \mathbb{R}^{n \times n}$ adapted to the filtration generated by the Brownian motion $W(t)$.
Then for any $T\geq 0$,
\begin{align*}
  f(X(T)) -f(X(0)) &= \int_0^T \sum_{i=1}^n \sum_{\ell=1}^n \left(\frac{\partial}{\partial X_\ell} f(X(t))\right) R_{i \ell}(t) \mathrm{d}W_i(t)\\
  &\qquad +  \frac{1}{2} \int_0^T \sum_{i=1}^n \sum_{j=1}^n \sum_{\ell=1}^n \left(\frac{\partial^2}{\partial X_{j} \partial X_\ell} f(X(t))\right) R_{i j }(t) R_{i \ell}(t) \mathrm{d}t.
\end{align*}
\end{lemma}
We note that the above version of It\^o's Lemma ( Lemma \ref{lemma_ito_lemma_new}) is given for real-valued variables.
We  apply It\^o's Lemma to complex matrix-valued stochastic processes, we will separate the real and imaginary parts of the It\^o integral and apply It\^o's Lemma separately to each part.

\subsection{Dyson Brownian motion}

In this section, we give additional results about the existence, continuity, and related properties of Dyson Brownian motion.
Let $\mathcal{O}(d)$ denote the space of $d \times d$ real orthogonal matrices, and $\mathcal{U}(d)$ the space of $d \times d$ complex unitary matrices. 
The following lemma, which guarantees the existence and uniqueness of solutions to the eigenvalue \eqref{eq_DBM_eigenvalues} and eigenvector SDE's \eqref{eq_DBM_eigenvectors}, is known -- see Theorem 2.3(a) in \cite{bourgade2017eigenvector} and Lemma 4.3.3 in \cite{anderson2010introduction} for solutions of just the eigenvalue process for any $\beta \geq 1$.
While the solutions are random processes, the outcome of these solutions can be shown to be unique when coupled to the underlying Brownian motion processes driving the SDE.
Such a coupling is referred to as a ``strong solution'' to the SDE (see e.g. \cite{lawler2010stochastic}).

\begin{lemma}[Existence and uniqueness of solutions to Dyson Brownian motion] \label{lemma_strong}
Consider any $T \geq T_0 \geq 0$ and $\beta \in \{1,2\}$.
Let $\{\gamma(t)\}_{t \in [0,T_0]} \subseteq \mathcal{W}_d$ be a continuous initial path for \eqref{eq_DBM_eigenvalues} and let $\{u(t)\}_{t \in [0,T_0]} \subseteq \mathcal{U}(d)$ if $\beta=2$ (or  $\{u(t)\}_{t \in [0,T_0]} \subseteq \mathcal{O}(d)$  if $\beta =1$) be a continuous initial path for \eqref{eq_DBM_eigenvectors}.
Then there exists a unique strong solution for the system of SDEs \eqref{eq_DBM_eigenvalues} on all of $[0,T]$.
Moreover, there exists a unique strong solution on all of $[0,T]$ for the system of SDEs comprising \eqref{eq_DBM_eigenvalues} and \eqref{eq_DBM_eigenvectors}.

\end{lemma}
In particular (by the definition of strong solution) the existence of strong solutions implies that the paths of Dyson Brownian motion are almost surely continuous on $[0,\infty)$.
This fact will be useful in proving our gap comparison theorem for coupled solutions of Dyson Brownian motions (Lemma \ref{lemma_gap_comparison}).
The following result shows that the paths of Dyson Brownian motion are continuous with respect to their initial conditions:
\begin{lemma}[Continuity w.r.t. initial condition; Proposition 4.3.5 in \cite{anderson2010introduction}]\label{lemma_continuity} 
Let $\gamma$ be a strong solution to \eqref{eq_DBM_eigenvalues} for any initial condition $\gamma(0)\in \mathcal{W}_d$. 
Then, at any time $t\geq 0$, $\gamma(t)$ is a continuous function of the initial condition $\gamma(0)$.
\end{lemma}
The following lemma is known; see Theorem 1.1 in \cite{inukai2006collision} and also \cite{rogers1993interacting}.
\begin{lemma}[Non-collision of Dyson Brownian motion for $\beta \geq 1$] \label{lemma_DBM_collision}
Let $\gamma$ be a solution to \eqref{eq_DBM_eigenvalues} with any initial condition $\gamma(0) \in \mathcal{W}_d$.
Let $\tau := \inf \{t>0: \gamma_i(t) = \gamma_j(t) \textrm{ for some } i\neq j\}$ be the first positive time any of the particles in  $\gamma(t)$ collide.
Then if $\beta \geq 1$, $\mathbb{P}(\tau <\infty) = 0$.
\end{lemma}

\subsection{Matrix inequalities}

The following  lemmas will help us bound the gaps in the eigenvalues of the Dyson Brownian motion:

\begin{lemma} [Theorem 4.4.5 of \cite{vershynin2018high}, special case \footnote{The theorem is stated for sub-Gaussian entries in terms of a constant $C$; this constant is $C=2$ in the special case where the entries are $N(0,1)$ Gaussian.}] \label{lemma_concentration} 
Let $W \in \mathbb{R}^{d \times d}$ with i.i.d. $N(0,1)$ entries. Then 
\begin{equation*}
    \mathbb{P}(\|W\|_2 > 2\sqrt{d} +s) <  2e^{-s^2}
\end{equation*}
 for any $s>0$.
\end{lemma}
Note that Lemma \ref{lemma_concentration} also applies to complex Gaussian matrices $W_1 + \mathfrak{i}W_2$ where $W_1, W_2$ have i.i.d. real $N(0,1)$ entries, since $\|W_1 + \mathfrak{i}W_2\|_2 \geq \|W_1\|_2$.

\begin{lemma}[\bf Weyl's inequality \cite{bhatia2013matrix}]\label{lemma_weyl}
If $A,B \in \mathbb{C}^{d\times d}$ are two Hermitian matrices, and denoting the $i$'th-largest eigenvalue of any Hermitian matrix $M$ by $\sigma_i(M)$, we have
\begin{equation*}
 \sigma_i(A) + \sigma_d(B) \leq  \sigma_i(A + B) \leq  \sigma_i(A) + \sigma_1(B).
\end{equation*}
\end{lemma}

\begin{lemma}[\bf Spectral norm bound (Theorem 4.3 of \cite{DBM_Neurips})] \label{lemma_spectral_martingale_b}
 For some universal constant $C$, and every $T>0$, we have,
      $$ \mathbb{P}\left(\sup_{t \in [0,T]}\|B(t)\|_2 > T\sqrt{d} + \alpha)\right) \leq e^{-C\frac{\alpha^2}{T^2}}.$$
\end{lemma}

\subsection{Davis-Kahan Sin-Theta theorem}

The following lemma gives a deterministic bound  on the change to the subspace spanned by the top-$k$ eigenvectors of a Hermitian matrix when it is perturbed by the addition of another Hermitian matrix.
Let A and $\hat{A}$ be two Hermitian matrices with eigenvalue decompositions
\begin{equation}\label{eq_eigenvalue_decomposition1}
    A = U \Lambda U^\ast = (U_1, U_2) \left({\begin{array}{cc}
   \Lambda_1 &  \\
    &  \Lambda_2 \\
  \end{array}}  \right) \left({\begin{array}{c}
   U_1^\ast \\
      U_2^\ast \\
  \end{array}}  \right)
\end{equation}

\begin{equation}\label{eq_eigenvalue_decomposition2}
    \hat{A} = \hat{U} \hat{\Lambda} \hat{U}^\ast = (\hat{U}_1, \hat{U}_2) \left({\begin{array}{cc}
   \hat{\Lambda}_1 &  \\
    &   \hat{\Lambda}_2 \\
  \end{array}}  \right) \left({\begin{array}{c}
    \hat{U}_1^\ast \\
       \hat{U}_2^\ast \\
  \end{array}}  \right),
\end{equation}
(although when we apply the Sin-Theta theorem we will only need the special case where $\hat{\Lambda} = \Lambda$).

\begin{lemma}[\bf sin-$\Theta$ Theorem \cite{davis1970rotation}] \label{lemma_SinTheta}
Let $A, \hat{A}$ be two Hermitian matrices with eigenvalue decompositions given in \eqref{eq_eigenvalue_decomposition1} and \eqref{eq_eigenvalue_decomposition2}.
Suppose that there are $\alpha > \beta>0$ and $\Delta>0$ such that the spectrum of $\Lambda_1$ is contained in the interval $[\alpha, \beta]$ and the spectrum of $\hat{\Lambda}_2$ lies entirely outside of the interval $(\alpha -\Delta, \beta + \Delta)$.
Then
\begin{equation*}
    ||| U_1 U_1^\ast - \hat{U}_1 \hat{U}_1^\ast |||  \leq \frac{||| \hat{A}- A||| }{\Delta},
\end{equation*}
where $||| \cdot |||$ denotes the operator norm or Frobenius norm (or, more generally, any  unitarily invariant norm).
\end{lemma}

\section{Overview of the proof of Theorem \ref{thm_rank_k_covariance_approximation_new}}\label{sec_technical_overview}

To prove Theorem \ref{thm_rank_k_covariance_approximation_new}, we bound the Frobenius-distance utility $\|\hat{M}_k - M_k\|_F$, where $\hat{M}:= M + G + G^\ast$ and $G$ is a matrix of i.i.d. standard complex Gaussians.
For simplicity, we assume $T=1$ in this section.

\subsection{Bounding the utility of the Gaussian mechanism with Dyson Brownian motion}
\subsubsection{The approach of \cite{DBM_Neurips}} To obtain their bounds, \cite{DBM_Neurips} also view the  addition of (real) Gaussian noise as a continuous-time matrix Brownian motion $\Phi(t) := M + B(t)$.
The  eigenvalues $\gamma_i(t)$ and eigenvectors $u_i(t)$ of $\Phi(t)$ evolve according to the same SDEs \eqref{eq_DBM_eigenvalues}, \eqref{eq_DBM_eigenvectors} as in the complex case, but with parameter $\beta = 1$.
To bound the utility of the rank-$k$ approximation $\hat{M}_k$, letting $\Phi(t) = U(t) \Gamma(t) U(t)^\top$  be a spectral decomposition of the symmetric matrix $\Phi(t)$ at every time $t\geq 0$, they define a new rank-$k$ matrix-valued stochastic process $\Theta(t):= U(t) \Sigma_k U(t)^\top$ whose eigenvalues are fixed to be the top-$k$ eigenvalues of the input $M$ at every time $t$.
Using the evolution equations \eqref{eq_DBM_eigenvectors} for the eigenvectors $u_i(t)$, they obtain an SDE for the matrix-valued process $\Theta(t)$ and integrate this SDE to get an expression for the expected 
utility:
\begin{align} \label{eq_t7}
 \textstyle \mathbb{E}\left[\left\|\Theta(T) -  \Theta(0)\right \|_F^2 \right] & \textstyle  = \frac{1}{2} \mathbb{E}\left[\left\| \int_0^{T}\sum_{i=1}^{d} \sum_{j \neq i} (\lambda_i - \lambda_j) \frac{\mathrm{d}B_{ij}(t)}{\gamma_i(t)-\gamma_j(t)}(u_i(t) u_j^\top(t) + u_j(t) u_i^\top(t))\right\|_F^2 \right] \nonumber\\
    &  \textstyle +  \mathbb{E}\left[\left\| \int_0^{T}\sum_{i=1}^{d} \sum_{j\neq i} (\lambda_i - \lambda_j) \frac{\mathrm{d}t}{(\gamma_i(t) - \gamma_j(t))^2} u_i(t) u_i^\top(t) \right\|_F^2 \right],
    \end{align}
    where $\lambda_i := \sigma_i$ for $i \leq k$ and $\lambda_i := 0$ for $i>k$.
    The idea is that, roughly speaking, each differential term $\frac{\mathrm{d}B_{ij}(t)}{\gamma_i(t)-\gamma_j(t)}(u_i(t) u_j^\top(t) + u_j(t) u_i^\top(t))$ adds noise to the matrix independently of the other terms at every time $t$ since the stochastic derivatives of the Brownian motions, $\mathrm{d}B_{ij}(t)$, are independent for every $i,j,t$ and independent of the $u_i(s)$ for all current and past times $s \leq t$.
    This allows the contribution of each of these terms to the (squared) Frobenius norm of the first term on the r.h.s. to add up as a sum of squares.
     Integrating \eqref{eq_t7} via Ito's Lemma (restated in our preliminaries as Lemma \ref{lemma_ito_lemma_new}), 
    they obtain an expression for the utility as a sum-of-squares of the ratios of the eigenvalue gaps:
    \begin{equation}\label{eq_t7_2} \textstyle
\mathbb{E}\left[\left\|\Theta(T) -  \Theta(0)\right \|_F^2 \right] =   \sum_{i=1}^{d} \int_0^{T} \mathbb{E}\left[ \sum_{j \neq i}  \frac{(\lambda_i - \lambda_j)^2}{(\gamma_i(t)-\gamma_j(t))^2} \right]
     +    T \mathbb{E}\left[\left(\sum_{j\neq i} \frac{\lambda_i - \lambda_j}{(\gamma_i(t)-\gamma_j(t))^2}\right)^2   \right]\mathrm{d}t.
\end{equation}
To bound the gap terms $\gamma_i(t) - \gamma_j(t)$ in \eqref{eq_t7_2} for all $i,j \leq k$, $i \neq j$, they use Weyl's inequality (restated here as Lemma \ref{lemma_weyl}), a deterministic eigenvalue perturbation bound which says that $\gamma_i(t) - \gamma_j(t) \geq \gamma_i(0) - \gamma_j(0) - \|B(t)\|_2$ for all $t$.
Thus, since $\|B(t)\|_2 = \Theta(\sqrt{d})$ w.h.p. for all $t\in [0,T]$, they must require that all gaps in the top-$k$  eigenvalues of $M$ satisfy  $\gamma_i(0) - \gamma_{i+1}(0) = \sigma_i - \sigma_{i+1} \geq \Omega(\sqrt{d})$ for every $i\leq k$.
Simplifying \eqref{eq_t7_2}, 
they show that $$\textstyle \mathbb{E}\left[\left\|\hat{M}_k - M_k\right \|_F^2 \right] \approx \mathbb{E}\left[\left\|\Theta(T) -  \Theta(0)\right \|_F^2 \right] \leq \tilde{O}\left(\sqrt{k}\sqrt{d} \frac{\sigma_k}{\sigma_k-\sigma_{k+1}}\right)$$ under their assumption that 
$\sigma_i - \sigma_{i+1} \geq \Omega(\sqrt{d})$ for every $i\leq k$.

\subsubsection{Our approach}
To bound the utility of the Gaussian mechanism
without any assumptions on the eigenvalue gaps $\sigma_i - \sigma_{i+1}$ for $i\neq k$, we would like to prove bounds on the gaps $\gamma_i(t) - \gamma_j(t)$ which hold even when initial gaps $\gamma_i(0) - \gamma_j(0) = \sigma_i- \sigma_{i+1}$ may not be $\Omega(\sqrt{d})$.
Unfortunately, since $\|B(t)\|_2 \geq \Omega(\sqrt{d})$ w.h.p. for $t = \Omega(1)$, we cannot rely on deterministic eigenvalue bounds such as Weyl's inequality, as this would not give any bound on $\gamma_i(t) - \gamma_{i+1}(t)$ unless $\sigma_i- \sigma_{i+1} \geq \Omega(\sqrt{d})$. 
To bypass this difficulty we would ideally like to obtain probabilistic lower bounds on the eigenvalue gaps $\gamma_i(t)- \gamma_{i+1}(t)$ which hold for any initial conditions on the top-$k$ eigengaps of $\gamma(0)$.

To see what bounds we might hope to show, note that if $\gamma(0) = 0$ then $\gamma(t)$ has the same joint distribution as the eigenvalues $\eta_1,\ldots, \eta_d$ of the rescaled GOE (GUE) matrix $\sqrt{t} (G+G^\ast)$ where $G$ is a matrix of i.i.d. real (complex) Gaussians.
This joint distribution is given by the following formula \cite{dyson1963random, ginibre1965statistical},
\begin{equation}\label{eq_joint_density}
 \textstyle   f(\eta_1,\ldots, \eta_d) = \frac{1}{R_\beta} \prod_{i<j} | \eta_i - \eta_j|^\beta e^{-\frac{1}{2} \sum_{i=1}^d \eta_i^2},
\end{equation}
where $R_\beta := \int \prod_{i<j} | \eta_i - \eta_j|^\beta e^{-\frac{1}{2} \sum_{i=1}^d \eta_i^2} \mathrm{d} \eta_1 \cdots \mathrm{d} \eta_d $ is a normalization constant.

From the repulsion factor $|\eta_i- \eta_{i+1}|^\beta$ in the joint distribution of the eigenvalues \eqref{eq_joint_density}, (and noting that the average eigenvalue gap of the standard GOE/GUE matrix $(G+G^\ast)$ is $\Theta(\frac{1}{\sqrt{d}})$ w.h.p. since $\|G+G^\ast\|_2 = \Omega\left(\sqrt{d}\right)$), roughly speaking one might expect that the GOE/GUE eigenvalue gaps satisfy $$ \textstyle \mathbb{P}\left(\eta_i- \eta_{i+1} \leq \frac{s}{\sqrt{d}}\right) = O\left(\int_0^s z^\beta \mathrm{d}z\right) = O(s^{\beta +1})$$ for all $s \geq 0$, where $\beta = 1$ in the real case and $\beta = 2$ in the complex case.
Assuming we can obtain such a bound, we would like to apply these bounds to bound the expectations of the terms on the r.h.s. of  \eqref{eq_t7_2}.  
The terms on the r.h.s. of  \eqref{eq_t7_2} with the smallest denominator, and therefore the most challenging to bound, are the terms $\mathbb{E}\left[\frac{(\lambda_i- \lambda_{i+1})^2}{(\gamma_i(t) - \gamma_{i+1}(t))^4}\right]$.
Assuming for the moment that we are able to show that
\begin{equation}\label{eq_conjectured_gap_bounds}
   \textstyle  \mathbb{P}\left(\gamma_i(t)- \gamma_{i+1}(t) \leq s\frac{\sqrt{t}}{\sqrt{d}}\right) \leq s^{\beta +1},
    \end{equation}
   then we would have the following bound for terms with denominators of order $r$:
\begin{equation}\label{eq_t8}
\textstyle \mathbb{E}\left[\frac{1}{(\gamma_i(t) - \gamma_{i+1}(t))^r}\right] 
\textstyle = \int_0^{\infty} \mathbb{P}\left(\gamma_i(t)- \gamma_{i+1}(t) \leq s^{-\frac{1}{r}}\right) \mathrm{d}s
 \leq \textstyle \left(\frac{d}{t}\right)^{\frac{r}{2}}\int_0^\infty s^{-\frac{1}{r}{(\beta+1)}}\mathrm{d}s.
\end{equation}
For the terms of order $r=2$, the r.h.s. of \eqref{eq_t8} is  $\int_0^\infty s^{-\frac{1}{2}{(\beta+1)}}\mathrm{d}s= \infty$ in the real case where $\beta=1$.
To bypass this problem, we observe that when the Gaussian noise is complex the integral on the r.h.s. of \eqref{eq_t8} becomes  $\int_0^\infty s^{-\frac{1}{2}{(\beta+1)}}\mathrm{d}s= O(1)$ since $\beta=2$ in the complex case.
Thus, while in the real case one expects the gaps to be  small enough that their inverse second moment $\mathbb{E}\left[\frac{1}{(\gamma_i(t) - \gamma_j(t))^2}\right]$ is infinite, in the complex case the repulsion between eigenvalues allows the gaps to be large enough that the inverse second moment is finite.
This motivates replacing the real Gaussian perturbation in the Gaussian mechanism with Complex-valued Gaussian noise (Algorithm \ref{alg_general_self_concordant}).

\paragraph{An SDE for a rank-$k$ matrix diffusion with dynamically changing eigenvalues, to track the utility under small initial eigengaps.}
Unfortunately, for the highest-order terms, of order $r=4$, the r.h.s. of \eqref{eq_t8} is  $\int_0^\infty s^{-\frac{1}{4}{(\beta+1)}}\mathrm{d}s= \infty$ even in the complex case where $\beta = 2$.
To get around this problem we replace the fixed eigenvalues  $\lambda_i = \sigma_i$ for $i \leq k$, of the rank-$k$ stochastic process $\Theta(t)$, with eigenvalues $\lambda_i(t)$ which change dynamically over time where at each time $t \geq 0$ we set $\lambda_i(t) = \gamma_i(t)$ for $i \leq k$ and $\lambda_i(t) = 0$ for $i>k$, in the hope that this will lead to cancellations in the highest-order terms.
This gives us a new rank-$k$ stochastic process $\Psi(t):= U(t) \Lambda(t) U(t)^\ast$ with dynamically changing eigenvalues $\Lambda(t) :=\mathrm{diag}(\lambda_1(t), \ldots, \lambda_d(t))$.
Since $\Psi(T) = \hat{M}_k$ and $\Psi(0) = M_k$, our goal is to bound $\|\hat{M}_k - M_k \|_F = \|\Psi(T) - \Psi(0)\|_F$.
Roughly speaking, this would lead to cancellations in the terms on the r.h.s. of \eqref{eq_t7_2} at every time $t \geq 0$:  the second-order terms would be reduced to constant terms $$ \textstyle \frac{(\lambda_i(t) - \lambda_j(t))^2}{(\gamma_i(t)-\gamma_j(t))^2} = \frac{(\gamma_i(t) - \gamma_j(t))^2}{(\gamma_i(t)-\gamma_j(t))^2} = 1$$ for $i\neq j$ $i,j \leq k$,
and fourth-order terms would be reduced to second-order terms, e.g.,  
 $$ \textstyle \frac{(\lambda_i(t)- \lambda_{i+1}(t))^2}{(\gamma_i(t) - \gamma_{i+1}(t))^4} = \frac{(\gamma_i(t)- \gamma_{i+1}(t))^2}{(\gamma_i(t) - \gamma_{i+1}(t))^4} = \frac{1}{(\gamma_i(t)- \gamma_{i+1}(t))^2}$$ for $i< k$.
This would allow us to obtain a finite bound for the expectation on the r.h.s. of \eqref{eq_t7_2}.

Towards this end, we first use the equations for the evolution of the eigenvalues \eqref{eq_DBM_eigenvalues} and eigenvectors \eqref{eq_DBM_eigenvectors} of Dyson Brownian motion to derive an SDE for our new rank-$k$ process $\Psi(t)$ (Lemma \ref{Lemma_projection_differntial} and \eqref{eq_ito_derivative}):
\begin{equation}\label{eq_t9}
 \textstyle \mathrm{d}\Psi(t)  = \sum_{i=1}^d \lambda_i(t) \mathrm{d}(u_i(t) u_i^\ast(t))) + (\mathrm{d}\lambda_i(t)) (u_i(t) u_i^\ast(t)) + \mathrm{d}\lambda_i(t) \mathrm{d}( u_i(t) u_i^\ast(t)),
\end{equation}
where  
$$ \textstyle \mathrm{d}(u_i(t) u_i^\ast(t))
     =  \sum_{j \neq i} \frac{u_i(t) u_j^\ast(t)\mathrm{d}B_{ij}(t)  + u_j(t) u_i^\ast(t)\mathrm{d}B_{ij}^\ast(t)}{\gamma_i(t) - \gamma_j(t)}
    - \frac{(u_i(t) u_i^\ast(t) - u_j(t)u_j^\ast(t))\mathrm{d}t}{(\gamma_i(t)- \gamma_j(t))^2},$$ and where $\mathrm{d}\lambda_i(t)$ is given by \eqref{eq_DBM_eigenvalues}.
    The last term  $\mathrm{d}\lambda_i(t) \mathrm{d}( u_i(t) u_i^\ast(t))$ in \eqref{eq_t9} vanishes as it consists only of higher-order differential terms.
Applying It\^o's lemma to compute the integral $\left\|\int_0^T \mathrm{d} \Psi(t)\right\|_F^2$ for the change in the (squared) Frobenius  distance, we get (Lemma \ref{Lemma_integral} and \eqref{eq_ito_integral_1} in the Proof of Theorem \ref{thm_rank_k_covariance_approximation_new}),
\begin{eqnarray}
&\mathbb{E}[\| \Psi(T) - \Psi(0)\|_F^2] =\left\|\int_0^T \mathrm{d} \Psi(t)\right\|_F^2 
\leq \int_{0}^{T}   \mathbb{E}\left[ \sum_{i=1}^{d}  \sum_{j \neq i}  \frac{(\lambda_i(t) - \lambda_j(t))^2}{(\gamma_i(t) - \gamma_j(t))^2} \mathrm{d}t \right] \nonumber\\ 
&+   T \int_{0}^{T}\mathbb{E}\left[\sum_{i=1}^{d}\left(\sum_{j\neq i} \frac{\lambda_i(t) - \lambda_j(t)}{(\gamma_i(t) - \gamma_j(t))^2}\right)^2   \right]\mathrm{d}t 
+ \int_{0}^T \sum_{i=1}^k\mathbb{E}\left[  \left(\sum_{j \neq i} \frac{1}{\gamma_i(t) - \gamma_j(t)}\right)^2 \right] \mathrm{d}t. \qquad 
     \end{eqnarray}
Plugging in our choice of $\lambda_i(t)$, we get (Equations \eqref{eq_u1} and \eqref{eq_a4} in the proof of Theorem \ref{thm_rank_k_covariance_approximation_new}),
\begin{eqnarray}\label{eq_t9_2}
 &\left\|\int_0^T \mathrm{d} \Psi(t)\right\|_F^2 \leq  \sum_{i=1}^{k}\int_{0}^{T}   \mathbb{E}\left[ \left( k +  \sum_{j >k}  \frac{(\gamma_i(t))^2}{(\gamma_i(t)- \gamma_j(t))^2} \right)\right]
     +    T \mathbb{E}\bigg[\left(\sum_{j \neq i: j\leq k} \frac{1}{\gamma_i(t)- \gamma_j(t)}\right)^2 \nonumber\\ &  + \bigg(\sum_{j> k} \frac{\gamma_i(t)}{(\gamma_i(t)- \gamma_j(t))^2}\bigg)^2 \bigg]
    + \mathbb{E}\left[  \left(\sum_{j \neq i} \frac{1}{\gamma_i(t) - \gamma_j(t)}\right)^2 \right] \mathrm{d}t.
    \end{eqnarray}
\noindent If we can prove the conjectured  gap bounds \eqref{eq_conjectured_gap_bounds}, we will have from \eqref{eq_t8} that $$\textstyle \mathbb{E}\left[\frac{1}{(\gamma_i(t) - \gamma_j(t))^2}\right] \leq \frac{d}{t (i-j)^2}$$ for all $i\neq j$ and, more generally, that  
$$\textstyle \mathbb{E}\left[\frac{1}{(\gamma_i(t) - \gamma_j(t))(\gamma_\ell(t) - \gamma_r(t))}\right] \leq \frac{d}{t \min((i-j)^2, (\ell-r)^2)}$$ for all $i\neq j$, $\ell\neq r$.
Moreover, if we assume a bound only on the $k$'th eigenvalue gap of $M$, $\sigma_k - \sigma_{k+1} \geq \Omega(\sqrt{d})$  (without assuming any bounds on the other eigenvalue gaps of $M$), we have by Weyl's inequality that $\gamma_k(t) - \gamma_{k+1}(t) \geq \sigma_k- \sigma_{k+1} - \|B(t)\|_2 = \Omega(\sigma_i- \sigma_{i+1})$.
Plugging in these conjectured probabilistic bounds, together with the worst-case Weyl inequality bounds for the $k$'th gap ($\gamma_k(t) - \gamma_{k+1}(t) \geq \Omega(\sigma_i- \sigma_{i+1})$, into \eqref{eq_t9_2} gives, (Equation \eqref{eq_u3} in the Proof of Theorem \ref{thm_rank_k_covariance_approximation_new})
$$\textstyle\mathbb{E}\left[\left\|\hat{M}_k -  M_k\right \|_F^2\right]= \left\|\int_0^T \mathrm{d} \Psi(t)\right\|_F^2 \leq \tilde{O}\left(kd \frac{\sigma_k^2}{(\sigma_k - \sigma_{k+1})^2}\right).$$

\subsection{Bounding the eigenvalue gaps of Dyson Brownian motion}\label{sec_technical_gaps}

To complete the proof of Theorem \ref{thm_rank_k_covariance_approximation_new},  we still need to show the conjectured bounds in \eqref{eq_conjectured_gap_bounds} (or at least show a close approximation to these bounds).
We do this by proving Lemmas \ref{lemma_gap_comparison} and \ref{lemma_GUE_gaps}, and present an overview of their proofs in this section.
 We start by recalling a few important ideas and results from random matrix theory.

\subsubsection{Useful ideas from random matrix theory} Starting with \cite{dyson1963random, ginibre1965statistical}, many works have made use of the intuition that the eigenvalues of a random matrix tend to repel each other, and can be interpreted in the context of statistical mechanics as a many-body system of charged particles undergoing a Brownian motion.
These particles repel each other with an ``electrical force'' arising from a potential  that decays logarithmically with the distance between pairs of particles (see e.g. \cite{rodriguez2014calibration}). 
The dynamics of these particles are described by the eigenvalue evolution equations \eqref{eq_DBM_eigenvalues} discovered by Dyson \cite{dyson1963random}, where the diffusion term $\mathrm{d} \gamma_i(t) = \mathrm{d}B_{i i}(t)$ describes the random component of each particle's motion and the terms $\frac{\beta}{\gamma_i(t) - \gamma_j(t)}$ describe  the repulsion between particles; the parameter $\beta$ can be interpreted either as the strength of the electrical force, or equivalently, as the (inverse) temperature of the system.

\cite{dyson1963random} showed that from the evolution equations \eqref{eq_DBM_eigenvalues} one can obtain the joint distribution of the eigenvalues of Dyson Brownian motion at equilibrium \eqref{eq_joint_density}.
If one initializes the matrix Brownian motion with all eigenvalues at $0$, at every time $t$ the matrix Brownian motion is in equilibrium (after scaling by $\frac{1}{\sqrt{t}}$) and equal in distribution to a GOE or GUE matrix scaled by $\sqrt{t}$, and thus \eqref{eq_joint_density} gives an explicit formula for the joint distribution of the eigenvalues of the GOE random matrix (for the real case $\beta = 1$) and GUE random matrix (for the complex case $\beta = 2$).

In the limit as $\beta\rightarrow \infty$ (with appropriate rescaling), the temperature of the system can be thought of as going to zero, and the solution to the evolution equations \eqref{eq_DBM_eigenvalues} converges to a  deterministic solution with particles ``frozen'' at  $\gamma_i(t) = \sqrt{t}\omega_i$ with probability 1 for some $\omega_1,\ldots, \omega_d \in \mathbb{R}$.
It has long been observed \cite{ginibre1965statistical, girko1985circular} that the gaps $\omega_i - \omega_{i+1}$ between these particles is, roughly, $\frac{1}{\sqrt{d}}$ in the ``bulk'' of the spectrum (i.e., the set of eigenvalues with index $cd <i< d- cd$ for any small constant $c$), while the particles have larger gaps near the edge of the spectrum.

More recently, \cite{erdHos2012rigidity} showed (restated here as Lemma \ref{lemma_rigidity}) that with high probability, the eigenvalues $\eta$ of the GOE/GUE random matrices are ``rigid'' in the sense that each eigenvalue $\eta_i$ falls within a small distance $\tilde{O}(\min(i, d-i+1)^{-\frac{1}{3}} d^{-\frac{1}{6}})$ of the zero-temperature eigenvalue $\omega_i$, where $\min(i, d-i+1)^{-\frac{1}{3}} d^{-\frac{1}{6}}$ is the average eigengap size in the region of the spectrum
  near $\omega_i$:
\begin{equation}\label{eq_ridgidity}
\textstyle |\eta_i - \omega_i| \leq O( \min(i, d-i+1)^{-\frac{1}{3}} d^{-\frac{1}{6}} \log(d)^{\log \log d}), \qquad \forall i \in [d], \quad \textrm{ w.h.p.}
\end{equation}

%
\subsubsection{Our results on eigenvalue gaps of Dyson Brownian motion}

\paragraph{Overview of proof of Lemma \ref{lemma_gap_comparison}.}

Recall that in Lemma \ref{lemma_gap_comparison}  we are given two solutions $\gamma(t)$ and $\xi(t)$ with any initial conditions $\gamma(0), \xi(0) \in \mathcal{W}_d$, and a coupling between  $\gamma(t)$ and $\xi(t)$ is defined by setting the underlying Brownian motion which generates each process in \eqref{eq_DBM_eigenvalues} to be equal at every time $t \geq 0$.
To prove the lemma, we must show that whenever the initial gaps of $\gamma(0)$ are $\geq$ the corresponding initial gaps of $\xi(0)$, $\gamma_i(0) - \gamma_{i+1}(0) \geq \xi_i(0) - \xi_{i+1}(0)$,  with probability $1$ the gaps of $\gamma(t)$ remain $\geq$ the gaps of the coupled process $\xi(t)$ at every time $t \geq 0$.

The idea behind the proof of Lemma \ref{lemma_gap_comparison} is to consider the net ``electrostatic pressure'' on each gap $\gamma_i(t) - \gamma_{i+1}(t)$-- that is, the difference between the sum of the forces from the eigenvalues $\gamma_j(t)$ for $j \notin\{i, i+1\}$ pushing on the gap $\gamma_i(t) - \gamma_{i+1}(t)$ from the outside to compress it, and the force  from the repulsion between the eigenvalues $\gamma_i(t)$ and $\gamma_{i+1}(t)$ pushing to expand the gap.
More formally, this net pressure is $\mathrm{d} \left(\gamma_i(t) - \gamma_{i+1}(t)\right) = \mathrm{d}\gamma_i(t) -\mathrm{d} \gamma_{i+1}(t)$ and, thus, we can compute it using \eqref{eq_DBM_eigenvalues}:
\begin{eqnarray} \label{eq_pressure}
    \textstyle \mathrm{d} \gamma_i(t) - \mathrm{d} \gamma_{i+1} (t) =  \mathrm{d}B_{i, i}(t) +   \sum_{j \neq i} \frac{\beta \mathrm{d}t}{\gamma_i(t) - \gamma_j(t)}   -  \! \!\left(\mathrm{d}B_{i+1, i+1}(t) + \! \!  \sum_{j \neq i+1} \frac{\beta \mathrm{d}t}{\gamma_{i+1}(t) - \gamma_j(t)}  \right).
\end{eqnarray}
Ideally, we would like to show that at any time where all the gaps of $\gamma(t)$ are at least as large as all the gaps of $\xi(t)$, 
we have $\mathrm{d} \gamma_i(t) - \mathrm{d} \gamma_{i+1}(t)\geq \mathrm{d} \xi_i(t) - \mathrm{d} \xi_{i+1}(t)$.
This in turn {\em would} imply that the gaps of $\gamma(t)$ expand faster (or contract slower) than the corresponding gaps of $\xi(t)$, and hence that the gaps of  $\gamma(t)$ remain larger than those of $\xi(t)$ at every time $t \geq 0$.
Unfortunately, the opposite may be true: if the eigenvalue gap $\gamma_i(t) - \gamma_{i+1}(t)$ is much larger than the gap $\xi_i(t) - \xi_{i+1}(t)$ then the repulsion between $\gamma_i(t)$ and $\gamma_{i+1}(t)$  pushing to expand the gap $\gamma_i(t) - \gamma_{i+1}(t)$ is much smaller than the repulsion pushing to expand the gap $\xi_i(t) - \xi_{i+1}(t)$.

To solve this problem, we prove Lemma  \ref{lemma_gap_comparison} by a contradiction argument.
Towards this end, we first define 
$\tau := \inf\{t\geq 0: \xi_i(t) - \xi_{i+1}(t) > \gamma_i(t) - \gamma_{i+1}(t) \textrm{ for some } i \in [d]  \}$  to be the first time where for some $i$,  the size of the $i$'th gap  $\xi_i(t) - \xi_{i+1}(t)$ becomes larger than the $i$'th gap  $\gamma_i(t) - \gamma_{i+1}(t)$ of $\gamma(t)$.
We assume (falsely), that $\tau < \infty$ and show that this leads to a contradiction.

Since the initial gaps of $\gamma(0)$ are at least as large as those of $\xi(0)$, and since the trajectories $\gamma(t)$ and $\xi(t)$ are continuous w.p.\ $1$, by the intermediate value theorem there must be an $i\in[d]$ such that 
\begin{equation}\label{eq_t10}
   \textstyle  \gamma_i(\tau) - \gamma_{i+1}(\tau) = \xi_i(\tau) - \xi_{i+1}(\tau),
\end{equation}
and the other gaps at time $\tau$ satisfy
   $\gamma_j(\tau) - \gamma_{j+1}(\tau) \geq \xi_j(\tau) - \xi_{j+1}(\tau)$ for $j \in [d]$. 
Plugging \eqref{eq_t10} into \eqref{eq_pressure}, we
obtain the difference in net electrostatic pressure on the $i$'th gap of $\gamma$ and $\xi$ at time $\tau$:
\begin{eqnarray}
    \textstyle  (\mathrm{d} \gamma_i(\tau) - \mathrm{d} \gamma_{i+1} (\tau)) -   (\mathrm{d} \xi_i(\tau) - \mathrm{d} \xi_{i+1} (\tau)) 
     =  \sum_{j \neq i, i+1} \frac{\beta \mathrm{d}\tau}{\gamma_i(\tau) - \gamma_j(\tau)}  -  \frac{\beta  \mathrm{d}\tau  }{\xi_i(\tau) - \xi_j(\tau)}
     \geq 0. \label{eq_t12}
\end{eqnarray}
The Brownian motion terms $\mathrm{d}B$ from \eqref{eq_pressure} 
cancel as we have coupled the processes $\gamma$ and $\xi$ by setting their underlying Brownian motions $B(t)$ to be equal. 
The terms $\frac{1}{\gamma_i(\tau) - \gamma_{i+1}(\tau)}$ and  $\frac{1}{\xi_i(\tau) - \xi_{i+1}(\tau)}$ arising from  \eqref{eq_pressure} which describe repulsion between the $i$'th and $i+1$'th eigenvalues cancel 
by \eqref{eq_t10}.
Thus, we are only left with the forces from the other eigenvalues pushing to compress the $i$'th gap of $\xi(\tau)$ and $\gamma(\tau)$ from the outside, which 
 allows us to then show that since 
the gaps of $\gamma(\tau)$ are at least as large as the corresponding gaps of $\xi(\tau)$ at time $\tau$, the r.h.s. of \eqref{eq_t12} is $\geq 0$ (Proposition \ref{sec_gap_comparison_proof}).

Next, we would like to show that \eqref{eq_t12} implies that the $i$'th gap of $\xi$ does {\em not} become larger than the $i$'th gap of $\gamma$ at time $\tau$, leading to a contradiction.
Unfortunately, \eqref{eq_t12} is not sufficient to show this, since, if $(\mathrm{d} \gamma_i(\tau) - \mathrm{d} \gamma_{i+1} (\tau)) -   (\mathrm{d} \xi_i(\tau) - \mathrm{d} \xi_{i+1} (\tau)) = 0$ we might have that the {\em second} derivative of the gaps of $\gamma$,  $(\mathrm{d}^2 \gamma_i(\tau) - \mathrm{d}^2 \gamma_{i+1} (\tau))$ is strictly smaller than the second derivative of the gaps of $\xi$,  $(\mathrm{d}^2 \xi_i(\tau) - \mathrm{d}^2 \xi_{i+1} (\tau))$. 
 To overcome this problem, we observe that, since 
 the gaps of $\gamma(\tau)$ are at least the corresponding gaps of $\xi(\tau)$ at time $\tau$, the only way the r.h.s. of \eqref{eq_t12} could be $0$ is if all the gaps of $\gamma(\tau)$ are equal to the corresponding gaps of $\xi(\tau)$. 
 In this case, the gaps would be equal at {\em every} time $t$ since solutions of Dyson Brownian motion are unique w.r.t. the underlying Brownian motion which defines our coupling  (see e.g. \cite{anderson2010introduction}, restated as Lemma \ref{lemma_strong}).
Thus, without loss of generality, we may assume that there is at least one $j$ such that the $j$'th gap of $\gamma(\tau)$ is strictly greater than the $j$'th gap of $\xi(\tau)$.
This in turn implies the r.h.s. of \eqref{eq_t12} is {\em strictly} $>0$, and hence the $i$'th gap of $\gamma$ becomes strictly {\em larger} than the $i$'th gap of $\xi$ in an open neighborhood of the time $\tau$.
This contradicts the definition of 
$\tau$, and hence by contradiction, we have $\tau = \infty$, and therefore the gaps of $\gamma(t)$ are $\geq$ the corresponding gaps of $\xi(t)$ at every time $t \geq 0$.

\paragraph{Overview of proof of Lemma \ref{lemma_GUE_gaps}.}

Roughly speaking, Lemma \ref{lemma_GUE_gaps} requires us to show the conjectured lower bound of $\mathbb{P}\left(\eta_i - \eta_{i+1} \leq \frac{1}{\sqrt{d}} s\right) \leq s^3$ for any $i$ and $s \geq 0$,  when $\eta_1,\ldots, \eta_d$ are the eigenvalues of the GUE random matrix.
As a first approach, we would ideally like to integrate the formula for the joint eigenvalue density $f(\eta)$  \eqref{eq_joint_density} over the set $A(s):= \left\{\eta \in \mathcal{W}_d: \eta_i - \eta_{i+1} \leq \frac{1}{\sqrt{d}} s \right\}$:
\begin{eqnarray}\label{eq_t14}
  \textstyle \mathbb{P}\left(\eta_i - \eta_{i+1} \leq \frac{1}{\sqrt{d}} s\right) = \frac{1}{R_2} \int_{\{\eta \in \mathcal{W}_d: \eta_i - \eta_{i+1} \leq \frac{1}{\sqrt{d}} s \}}   \prod_{i<j} | \eta_i - \eta_j|^2 e^{-\frac{1}{2} \sum_{i=1}^d \eta_i^2} \mathrm{d} \eta.
\end{eqnarray}
Unfortunately, we do not know of a closed-form expression for the $d$-dimensional integral \eqref{eq_t14}.

To get around this, we observe that, given any $\eta$ such that $\eta_i - \eta_{i+1} \leq \frac{1}{\sqrt{d}} s$, if we can somehow find a map $\phi: \mathcal{W}_d \rightarrow \mathcal{W}_d$ such that the term $|\phi(\eta)[i] - \phi(\eta)[i+1]| \geq \frac{1}{s}|\eta_i - \eta_{i+1}|$, and all other terms ($|\phi(\eta_j) - \phi(\eta_\ell)|$ for $(j, \ell) \neq (i, i+1)$ and $e^{-\frac{1}{2} \sum_{i=1}^d \eta_i^2}$)  in the joint density remain unchanged, then we would have that $f(\phi(\eta)) \geq \frac{1}{s^2}f(\eta)$.
Moreover, roughly speaking, one might hope that, since $\phi$ expands one of the gaps by $\frac{1}{s}$ and leaves all the other gaps unchanged, the map $\phi$ would be invertible and the Jacobian determinant of such a map would satisfy $\mathrm{det}(J_\phi(\eta)) \geq \frac{1}{s}$ for all $\eta \in \mathcal{W}_d$.
This in turn would imply that the r.h.s. of \eqref{eq_t14} would satisfy
\begin{eqnarray}\label{eq_t15}
 \textstyle  \mathbb{P}\left(\eta_i - \eta_{i+1} \leq \frac{1}{\sqrt{d}} s\right)
  = s^3 \int_{A(s)}   f(\eta)  \frac{1}{s^3} \mathrm{d} \eta 
  \leq s^3 \int_{A(s)}  f(\eta)   \frac{f(\phi(\eta))}{f(\eta)}  \mathrm{det}(J_\phi(\eta))  \mathrm{d} \eta
  \leq s^3.
  \end{eqnarray}
The last step holds since $\phi$ is injective and $f$ is a probability density, implying the integral is $\leq 1$.

Unfortunately, one can easily see that there does not exist a map $\phi$ which expands the $i$'th gap term $|\eta_i-\eta_{i+1}|$ in the joint eigenvalue density  \eqref{eq_joint_density} by $\frac{1}{s}$, but leaves all other terms unchanged.
This is because, to expand $|\eta_i-\eta_{i+1}|$ but leave the other gap terms unchanged, one would, e.g., have to translate the other eigenvalues  $\eta_j$ for $j \leq i$ aside by an amount  $(\frac{1}{s}-1)|\eta_i-\eta_{i+1}|$.
To circumvent this problem, we instead consider a different map $\phi : \mathcal{W}_d \rightarrow \mathcal{W}_d$ which, roughly speaking, expands the eigenvalue gap $\eta_i-\eta_{i+1}$ by a factor of $\frac{1}{s}$, leaves all other gaps unchanged, and translates the eigenvalues of $\eta_j$ for $j \leq i$ to the left by an amount $\frac{1}{s}(\eta_i-\eta_{i+1})$ to make room for the expanded eigenvalue gap (see equations \eqref{eq_phi1}-\eqref{eq_phi3} for the full definition of $\phi$).
Since  $\eta_i = \Theta(\sqrt{d})$ w.h.p., when e.g. $|\eta_i-\eta_{i+1}| \geq \Theta(\frac{1}{\sqrt{d}})$ this would decrease the exponential term in the joint density by a factor of $$\textstyle e^{-\frac{1}{2} \sum_{j=1}^i (\phi(\eta)[i] - \eta_i)^2} \approx e^{-\frac{1}{2} \sum_{j=1}^i (\phi(\eta)[i] - \eta_i) \sqrt{d} } \geq  e^{-\frac{1}{2} \sum_{j=1}^i \frac{\sqrt{d}}{\sqrt{d}}} = e^i.$$
For $i \leq O(1)$, this is not an issue as then one has $e^i = O(1)$, and hence   $\frac{f(\phi(\eta))}{f(\eta)} \geq \Omega\left(\frac{1}{s^2}\right)$ (Lemma \ref{lemma_density_ratio_edge}).
Roughly speaking this fact, together with a bound on the Jacobian determinant of $\phi$ (Proposition \ref{prop_Jacobian_phi} which says  $\mathrm{det}(J_\phi(\eta)) \geq \frac{1}{s}$) and since $\phi$ is injective (Proposition \ref{prop_map_phi}), allows us to use the above map $\phi$ to show that \eqref{eq_t15} holds whenever the $i$'th eigenvalue gap is near the edge of the spectrum ($i \leq \tilde{O}(1)$).

To bound  $\gamma_i-\gamma_{i+1}$ for $i \geq \tilde{\Omega}(1)$, which are not near the edge of the spectrum, 
we will use the rigidity property of the GUE eigenvalues  \eqref{eq_ridgidity} 
(\cite{erdHos2012rigidity}; restated here as Lemma \ref{lemma_rigidity}).
Roughly, this rigidity property says that none of the eigenvalues $\eta_i$ fall more than a distance $\mathfrak{b} = O(\log(d)^{\log \log d}) = \tilde{O}(1)$ from their ``zero-temperature'' locations $\omega_i$.
Hence, $\eta_{j} \in [a,b]$ for all $i - \mathfrak{b} \leq j \leq i + \mathfrak{b}$,  where $a:=\eta_{i-\mathfrak{b}} \geq \omega_i - \mathfrak{b}^2\sqrt{d}$ and $b:= \eta_{i+\mathfrak{b}} \leq \omega_i + \mathfrak{b}^2\sqrt{d}$ w.h.p.

To apply this rigidity property, we define a new map  $g : \mathcal{W}_d \rightarrow \mathcal{W}_d$ where $g(\eta)$ leaves all eigenvalues $\eta_j$ outside $[a,b]$ fixed.
And $g(\eta)$ expands the $i$'th eigengap by a factor of $\frac{1}{s}$:   $g(\eta)[i] - g(\eta)[i+1] \geq \frac{1}{s}(\eta_i - \eta_{i+1})$. 
To ``make room'' for the expansion of the $i$'th gap without changing the locations of the eigenvalues outside $[a,b]$, it shrinks the eigengaps inside $[a,b]$ by a factor of $1-\alpha$ where $\alpha := (\frac{1}{s}-1)\frac{\eta_i - \eta_{i+1}}{b-a} \leq \mathfrak{b}^{-3}$ whenever $\eta \in A(\frac{s}{ \mathfrak{b}})$ because $\eta_i - \eta_{i+1} \leq s\frac{1}{ \mathfrak{b} \sqrt{d}}$  if $\eta \in A(\frac{s}{ \mathfrak{b}})$
(See \eqref{eq_g3}-\eqref{eq_g1} for the definition of $g$).
Thus, roughly, for all $\eta \in A(\frac{s}{\mathfrak{b}})$, 
\begin{eqnarray}\label{eq_t16}
\textstyle \frac{f(g(\eta))}{f(\eta)} \!= \!\prod_{j \neq \ell} \frac{|g(\eta)[j] - g(\eta)[\ell]|^2}{|\eta_j - \eta_{\ell}|^2}  e^{-\frac{1}{2} \sum_{j= i - \mathfrak{b}}^{i +\mathfrak{b}} \eta_j^2- g(\eta)[j]^2} 
    \geq \frac{1}{s^2}  (1-\alpha)^{2\mathfrak{b}^2}  e^{-\frac{1}{2} \sum_{j= i - \mathfrak{b}}^{i +\mathfrak{b}} \frac{1}{\mathfrak{b}}} 
    \geq \frac{1}{s^2}.
\end{eqnarray}
The first inequality holds since the product has $O(\mathfrak{b}^2)$ ``repulsion'' terms $\frac{|g(\eta)[j] - g(\eta)[\ell]|^2}{|\eta_j - \eta_{\ell}|^2} \geq (1-\alpha)^2$  where  $\ell, j \in [i-\mathfrak{b}, i+\mathfrak{b}]$ and one term $\frac{|g(\eta)[i] - g(\eta)[i+1]|^2}{|\eta_i - \eta_{i+1}|^2} \geq \frac{1}{s^2}$.
Replacing $g$ with $\phi$ and $A(s)$ with $A(\frac{s}{\mathfrak{b}})$ in $\eqref{eq_t15}$, and plugging in \eqref{eq_t16}  
we get, roughly, that for all $s \geq 0$ and all $i \in [d]$,
\begin{equation}
    \textstyle \mathbb{P}\left(\eta_i - \eta_{i+1} \leq \frac{1}{\mathfrak{b} \sqrt{d}} s \right) = \int_{A(\frac{s}{\mathfrak{b}})} f(\eta) \mathrm{d} \eta  \leq s^3 \int_{A(\frac{s}{\mathfrak{b}})}  f(\eta) \times  \frac{f(g(\eta))}{f(\eta)}  \mathrm{det}(J_g(\eta))  \mathrm{d} \eta \leq s^3.
\end{equation}
This completes the proof overview of Lemma \ref{lemma_GUE_gaps}.
To conclude, we note that one can convert the expectation bound in Theorem \ref{thm_rank_k_covariance_approximation_new} into the following high-probability bound  using the approach suggested in Appendix E of \cite{DBM_Neurips}:
 $P(\|\hat{M}_k - M_k\|_F > s \sqrt{k}\sqrt{d}) \leq  \tilde{O}(e^{-s})$ for $s>0$. 
 However, here there is an additional difficulty since the eigenvalue gaps of Dyson Brownian motion only satisfy high probability lower bounds with polynomial decay, but not exponential decay.  
 To overcome this, we believe one can use an approach similar to the proof of Lemma \ref{lemma_GUE_gaps} to show that the joint density of the GUE eigenvalues \eqref{eq_joint_density} implies that whenever roughly $|i-j| >\tilde{\Omega}(1)$, the gaps  $\eta_i - \eta_{i+1}$ and $\eta_j - \eta_{j+1}$ are (nearly) independent random variables in the sense that roughly $\mathbb{P}((\gamma_i(t) - \gamma_{i+1}(t) < x) \times (\gamma_j(t) - \gamma_{j+1}(t) < y)) \approx \mathbb{P}((\gamma_i(t) - \gamma_{i+1}(t) < x) \times \mathbb{P}(\gamma_j(t) - \gamma_{j+1}(t) < y)$ for $x,y >0$.
It then follows that with a probability that decays exponentially in $d$, all but $\tilde{O}(1)$ of the GUE eigenvalue gaps satisfy $\eta_i - \eta_{i+1} \geq \tilde{\Omega}\left(\frac{1}{\sqrt{d}}\right)$.

\section{Differentially private rank-$k$ approximation: Proof of Theorem \ref{thm_utility}} \label{sec_proof_utility_privacy}
\begin{proof}[Proof of Theorem \ref{thm_utility}] 
\paragraph{Privacy.}

The real Gaussian mechanism, $M + \sqrt{T}(W_1 + W_1^\top)$, where $W_1$ is a matrix with i.i.d.  $N(0,1)$ entries, was studied in \cite{dwork2014analyze} and shown to be $(\epsilon, \delta)$-differentially private for $T = \frac{2\log\frac{1.25}{\delta}}{\eps^2}$.
Our Algorithm \ref{alg_quaternion_Gaussian} is $(\epsilon, \delta)$-differentially private since it is a post-processing of the real Gaussian mechanism.
This is because any post-processing of an $(\varepsilon, \delta)$-differentially private mechanism (which does not have access to the original input matrix $M$) is guaranteed to be  $(\varepsilon, \delta)$-differentially private (see e.g. \cite{dwork2006our},  \cite{dwork2014algorithmic}).
To see why Algorithm \ref{alg_quaternion_Gaussian}  is a post-processing of the real Gaussian mechanism, observe that

\begin{align*}
    \hat{M} &= M + G\\
    &= M +  W + W^\ast\\
    &= M +  (W_1 +  W_2 \mathfrak{i}) + (W_1 +  W_2 \mathfrak{i})^\ast\\
    &= M + W_1 + W_1^\top + [W_2 \mathfrak{i} + (W_2 \mathfrak{i})^\ast].
\end{align*}

\paragraph{Utility of complex matrix $\hat{M}_k$ implies Utility of real matrix $Y$.}

Let $M = V \Sigma V^\top$ be a diagonalization of the real symmetric input matrix $M$ with eigenvalues $\sigma_1\geq \cdots \geq \sigma_d \geq 0$.
Let $M_k = V \Sigma_k V^\top$ be a (non-private) rank-$k$ approximation of $M$, where $\Sigma_k = \mathrm{diag}(\sigma_1,\ldots, \sigma_k, 0, \ldots,0)$.
Suppose we can show an upper bound on $\|\hat{M}_k -  M_k\|_F$, where $\hat{M}_k$ is the complex matrix in Algorithm \ref{alg_quaternion_Gaussian}.

Let $\mathfrak{R}_k:= \{A \in \mathbb{R}^{d\times d}: \mathrm{rank}(A) \leq k\}$ denote the set of real $d \times d$ rank-$k$ matrices.
 Since $Y = \mathrm{Real}(\hat{V} \hat{\Sigma}_k \hat{V}^\ast)$, we have that $Y \in \mathrm{argmin}_{Z \in \mathfrak{R}_k} \{\|\hat{M}_k-Z\|_F\}$.
This is because $\mathrm{Real}(\hat{V} \hat{\Sigma}_k \hat{V}^\ast)$ is a matrix of rank at most $k$ and the real and imaginary parts of $\hat{V} \hat{\Sigma}_k \hat{V}^\ast$ are orthogonal to each other in the Frobenius inner product. 
Thus, since $Y \in \mathrm{argmin}_{Z \in \mathfrak{R}_k} \{\|\hat{M}_k-Z\|_F\}$ and $M_k \in  \mathfrak{R}_k$ is also in the set of real-valued rank-$k$ matrices, we have that
\begin{equation*}
\|\hat{M}_k-Y\|_F \leq \|\hat{M}_k -  M_k\|_F.
\end{equation*}
Therefore, we have
\begin{equation}\label{eq_o2}
\|Y - M_k\|_F \leq \|\hat{M}_k-Y\|_F + \|\hat{M}_k -  M_k\|_F \leq 2  \|\hat{M}_k -  M_k\|_F.
\end{equation}
Plugging in our bound for $\sqrt{\mathbb{E}[\|\hat{M}_k -  M_k\|_F^2]}$ from  Theorem \ref{thm_rank_k_covariance_approximation_new} into \eqref{eq_o2}, we get that
$$\sqrt{\mathbb{E}[\|M_k - Y\|_F^2]}  \leq \tilde{O}\left(\sqrt{kd} \frac{\sigma_k}{\sigma_k - \sigma_{k+1}} \times  \frac{\sqrt{\log(\frac{1}{\delta})}}{\eps}\right).$$
\end{proof}

\section{Complex Gaussian perturbations: Proof of Theorem \ref{thm_rank_k_covariance_approximation_new}}\label{sec_utility}

\subsection{Defining the stochastic process on the space of rank-$k$ matrices}
At every time $t$, let $\Phi(t) = U(t) \Gamma(t) U(t)^\ast$  be a spectral decomposition of the symmetric matrix $\Phi(t)$,  where $\Gamma(t)$ is a diagonal matrix with diagonal entries $\gamma_1(t) \geq \cdots \geq \gamma_d(t)$ that are the eigenvalues of $\Phi(t)$, and $U(t) = [u_1(t), \ldots, u_d(t)]$ is a $d\times d$ unitary matrix whose columns $u_1(t), \ldots, u_d(t)$ are an orthonormal basis of eigenvectors of $\Phi(t)$.
At every time $t$, define $\Psi(t)$ to be the symmetric matrix with any eigenvalues $\lambda_1(t) \geq \cdots \lambda_d(t)$, where $\Lambda(t) :=\mathrm{diag}(\lambda_1(t), \ldots, \lambda_d(t))$,   and with eigenvectors given by the columns of $U(t)$: 
 $$   \Psi(t):= U(t) \Lambda(t) U(t)^\ast \qquad\forall t \in [0,T].$$
For all $t$, will fix $\lambda_i(t) = \gamma_i(t)$ for $i\leq k$ and $\lambda_i(t) = 0$ for all $i>k$.

\subsection{Preliminary results}

For every $\alpha>0$, let $\hat{E}_\alpha$ be the ``bad'' event that either $\sup_{t \in [0,T]}\|B(t)\|_2 > 4\sqrt{T}(\sqrt{d} + \alpha)$ or  $\sup_{t \in [0,t_0]}\|B(t)\|_2 > 4\sqrt{t_0}(\sqrt{d} + \alpha)$, or $\inf_{t_0 \leq t \leq T, 1\leq i < d  }\gamma_i(t) - \gamma_{i+1}(t) \leq \frac{1}{d^{10}} \frac{\sqrt{t}}{\mathfrak{b}\sqrt{d}}$.
In the following, we set $\alpha = 20\log^{\frac{1}{2}}(\sigma_1 + T)$ and
 $t_0 = \frac{1}{(kd)^{10} + k\alpha^2 +\sigma_1^2}$.
The following Lemma shows that  $\hat{E}_\alpha$ occurs with very low probability:
\begin{lemma}[\bf Probability of ``bad'' event occurring] \label{lemma_spectral_martingale}
 For every $T>0$ and every $\alpha>0$, we have,
      $ \mathbb{P}\left(\hat{E}_\alpha\right) \leq 4\sqrt{\pi} e^{-\frac{1}{8}\alpha^2}  + \frac{T}{d^{600}}.$
\end{lemma}
\begin{proof}
\begin{align*}
    \mathbb{P}\left(\hat{E}_\alpha\right) &\leq  \mathbb{P}\left(\sup_{t \in [0,T]}\|B(t)\|_2 > 4\sqrt{T}(\sqrt{d} + \alpha)\right) +  \mathbb{P}\left(\sup_{t \in [0,t_0]}\|B(t)\|_2 > 4\sqrt{t_0}(\sqrt{d} + \alpha)\right)\\
    &\qquad+ \mathbb{P}\left(\inf_{t_0 \leq t \leq T, 1\leq i < d  }\gamma_i(t) - \gamma_{i+1}(t) \leq \frac{1}{d^{10}} \frac{\sqrt{t}}{\mathfrak{b}\sqrt{d}}\right)\\
    & \stackrel{\textrm{Lemmas \ref{lemma_bad_event}, \ref{lemma_spectral_martingale_b}}}{\leq} 4\sqrt{\pi} e^{-\frac{1}{8}\alpha^2}  + \frac{T}{d^{600}}.
\end{align*}
\end{proof}

\begin{lemma}\label{lemma_utility_rare_event}
If $\alpha \geq 20\log^{\frac{1}{2}}(\sigma_1 +T)$, then we have
\begin{equation*}
    \mathbb{E}[\|\Psi(T) - \Psi(0)\|_F^2] \leq 4\mathbb{E}[\|\Psi(T) - \Psi(0)\|_F^2 \times \mathbbm{1}_{\hat{E}_\alpha^c}] +d.
    \end{equation*}
\end{lemma}
\begin{proof}

\begin{equation}\label{eq_u12}
    \mathbb{E}[\|\Psi(T) - \Psi(0)\|_F^2] \leq 4\mathbb{E}[\|\Psi(T) - \Psi(0)\|_F^2 \times \mathbbm{1}_{\hat{E}_\alpha^c}] + 4\mathbb{E}[\|\Psi(T) - \Psi(0)\|_F^2 \times \mathbbm{1}_{\hat{E}_\alpha}]. 
\end{equation}

\begin{align*}
    \|\Psi(T) - \Psi(0)\|_F &= \| U(T) \Gamma_k(T) U(T)^\ast - U(0) \Gamma_k(0) U(0)^\ast\|_F\\
&\leq \| U(T) \Gamma_k(T) U(T)^\ast\|_F + \|U(0) \Gamma_k(0) U(0)^\ast\|_F\\
    &\leq \| U(T) \Gamma(T) U(T)^\ast\|_F + \|U(0) \Gamma(0) U(0)^\ast\|_F\\
    &=\|\Phi(T)\|_F+ \|\Phi(0)\|_F\\
    &=\|M + B(T)\|_F + \|M\|_F\\
    &\leq 2\|M\|_F +\|B(T)\|_F.
\end{align*}
Therefore,
\begin{align} \label{eq_u13}
    \mathbb{E}[\|\Psi(T) - \Psi(0)\|_F^2 \times \mathbbm{1}_{\hat{E}_\alpha}] 
    &\leq \mathbb{E}[ (2\|M\|_F +\|B(T)\|_F)^2\times \mathbbm{1}_{\hat{E}_\alpha}]\nonumber\\
        &\leq \mathbb{E}[ (16\|M\|_F^2 + 4\|B(T)\|_F^2)\times \mathbbm{1}_{\hat{E}_\alpha}]\nonumber\\
    &= 16\|M\|_F^2\times \mathbb{P}(\hat{E}_\alpha)+ 4\mathbb{E}[\|B(T)\|_F^2\times \mathbbm{1}_{\hat{E}_\alpha}] \nonumber\\
      &\leq 16\|M\|_F^2\times \mathbb{P}(\hat{E}_\alpha)+ 4\sqrt{d}\mathbb{E}[\|B(T)\|_2^2\times \mathbbm{1}_{\hat{E}_\alpha}]\nonumber\\
      &=16\|M\|_F^2\times \mathbb{P}(\hat{E}_\alpha)+ 4\sqrt{d} \int_{16T(\sqrt{d} + \alpha)^2}^\infty\mathbb{P}(\|B(T)\|_2^2 >s)\mathrm{d}s \nonumber\\
            &=16\|M\|_F^2\times \mathbb{P}(\hat{E}_\alpha)+ 4\sqrt{d} \int_{4T(\sqrt{d} + \alpha)^2}^\infty\mathbb{P}(T\|W\|_2^2 >s)\mathrm{d}s \nonumber\\
             &=16\|M\|_F^2\times \mathbb{P}(\hat{E}_\alpha)+ 4\sqrt{d} \int_{16T(\sqrt{d} + \alpha)^2}^\infty\mathbb{P}\left(\|W\|_2 >\frac{\sqrt{s}}{\sqrt{T}}\right)\mathrm{d}s \nonumber\\
                      &=16\|M\|_F^2\times \mathbb{P}(\hat{E}_\alpha)+ 4\sqrt{d} \int_{16T(\sqrt{d} + \alpha)^2}^\infty 2e^{-\left(\frac{\sqrt{s}}{\sqrt{T}} -2\sqrt{d}\right)^2}\mathrm{d}s \nonumber\\
                       &=16\|M\|_F^2\times \mathbb{P}(\hat{E}_\alpha)+ 4\sqrt{d} \int_{16T(\sqrt{d} + \alpha)^2}^\infty 2e^{-\left(\frac{s}{T} -2\sqrt{d}\frac{\sqrt{s}}{\sqrt{T}} +4d\right)}\mathrm{d}s\nonumber\\
                            &\leq  16\|M\|_F^2\times \mathbb{P}(\hat{E}_\alpha)+ 4\sqrt{d} \int_{16T(\sqrt{d} + \alpha)^2}^\infty 2e^{-\left(\frac{s}{2T} +4d\right)}\mathrm{d}s\nonumber\\
                              &=  16\|M\|_F^2\times \mathbb{P}(\hat{E}_\alpha)+ 4\sqrt{d} e^{-4d} \int_{16T(\sqrt{d} + \alpha)^2}^\infty 2e^{-\frac{s}{2T}}\mathrm{d}s\nonumber\\
                                 &=  16\|M\|_F^2\times \mathbb{P}(\hat{E}_\alpha)+ 4\sqrt{d} e^{-4d} 4T e^{-\frac{16T(\sqrt{d} + \alpha)^2}{2T}}\nonumber\\
                               &=  16\|M\|_F^2\times \mathbb{P}(\hat{E}_\alpha)+ 4\sqrt{d} e^{-4d} 4T e^{-8(\sqrt{d} + \alpha)^2}\nonumber\\
                               &\leq  16\|M\|_F^2\times \mathbb{P}(\hat{E}_\alpha)+ 1\nonumber\\ 
                               &\stackrel{\textrm{Lemma \ref{lemma_spectral_martingale}}}{\leq}   16\|M\|_F^2\times \left(4\sqrt{\pi} e^{-\frac{1}{8}\alpha^2}+ \frac{T}{d^{600}}\right)+ 1\nonumber\\
                            &\leq 16 d \sigma_1^2 \times \left(4\sqrt{\pi} e^{-\frac{1}{8}\alpha^2}+ \frac{T}{d^{600}}\right) +1 \nonumber\\
                            &\leq \frac{1}{4}d + \frac{T}{d^{200}},
\end{align}
where $W$ is a matrix with i.i.d. $N(0,1)$ entries, and $Y\sim(0,\frac{1}{2})$.
The last inequality is because $\alpha \geq 20\log^{\frac{1}{2}}(d\sigma_1 +T)$ and $\sigma_1^2 \leq d^{100}$.
Plugging in \eqref{eq_u13} into \eqref{eq_u12} completes the proof.
\end{proof}
The following lemma will be useful in bounding the gaps $\gamma_i(t) - \gamma_j(t)$ for $i\leq k < j$.

\begin{lemma}[\bf ``Worst-case'' eigenvalue gap bound]\label{lemma_gap_concentration}
Whenever  $\gamma_i(0) - \gamma_{i+1}(0) \geq 8 \sqrt{T} \sqrt{d}$ for every $i \in S$ and $T>0$ and some subset $S \subset [d-1]$, we have that for any $\alpha>0$, 
\begin{equation*}
    \bigcup_{i\in S} \left\{\inf_{t \in [0,T]} \gamma_i(t) - \gamma_{i+1}(t) <  \frac{1}{2}(\gamma_i(0) - \gamma_{i+1}(0)) - \alpha) \right\} \subseteq \hat{E}_\alpha.
    \end{equation*}
\end{lemma}
\begin{proof}[Proof of Lemma \ref{lemma_gap_concentration}]
Since, at every time $t$, $\Phi(t)= M + B(t)$ and  $\gamma_1(t) \geq \cdots \geq \gamma_d(t)$ are the eigenvalues of $\Phi(t)$, Weyl's inequality implies that
\begin{equation} \label{eq_gap1}
    \gamma_i(t) - \gamma_{i+1}(t) \geq \gamma_i(0) - \gamma_{i+1}(0) - \|B(t)\|_2, \qquad \forall t\in[0,T], i \in [d].
\end{equation}
Therefore, \eqref{eq_gap1} implies that
    \begin{align*}
             &\bigcup_{i\in S} \left\{\inf_{t \in [0,T]} \gamma_i(t) - \gamma_{i+1}(t) <  \frac{1}{2}(\gamma_i(0) - \gamma_{i+1}(0)) - \alpha) \right\}\\
              &\stackrel{\textrm{Eq. \eqref{eq_gap1}}}{\subseteq} \bigcup_{i\in S} \left\{ \gamma_i(0) - \gamma_{i+1}(0) - \sup_{t \in [0,T]} 2\|B(t)\|_2  <  \frac{1}{2}(\gamma_i(0) - \gamma_{i+1}(0)) - \alpha) \right\}  \\
         &=  \bigcup_{i\in S} \left\{ \sup_{t \in [0,T]} \|B(t)\|_2  >  \frac{1}{4}(\gamma_i(0) - \gamma_{i+1}(0)) + \frac{1}{2} \alpha) \right\} \\
                  &\subseteq  \bigcup_{i\in S} \left\{ \sup_{t \in [0,T]} \|B(t)\|_2  >   2\sqrt{T} \sqrt{d} + \frac{1}{2} \alpha) \right\}\\
                  &= \left\{\sup_{t \in [0,T]} \|B(t)\|_2  >   2\sqrt{T} \sqrt{d} + \frac{1}{2} \alpha \right\}\\ 
                  &= \hat{E}_\alpha.
    \end{align*}
The first inequality holds by \eqref{eq_gap1},
and the second inequality holds since the statement of Lemma \ref{lemma_gap_concentration} assumes that $\gamma_i(0) - \gamma_{i+1}(0) \geq 8 \sqrt{T} \sqrt{d}$.
\end{proof}
The following proposition provides a crude bound on the Frobenius distance over the very short time interval $[0,t_0]$, which we will use to ``jump-start'' our more sophisticated bound on the much longer interval $[t_0, T]$:
\begin{proposition}\label{lemma_t0}
Suppose that  $\sigma_k - \sigma_{k+1} \geq 4T \sqrt{d} + 2\alpha$.
Then for every $0 \leq t_0 < 1$ we have
\begin{equation*}
    \|\Psi(t_0) - \Psi(0)\|_F \times \mathbbm{1}_{\hat{E}_\alpha^c} \leq \sqrt{t_0}\left(2 \sqrt{k} (\sqrt{d} + \alpha)  + 8\sigma_1\right)
\end{equation*}
with probability $1$.
\end{proposition}

\begin{proof}
At every time $t\geq 0$, let $U_k(t)$ denote the $d\times k$ matrix consisting of the first $k$ columns of $U(t)$.
Further, let $\Gamma_k(t)$ denote the $k\times k$ matrix consisting of the first $k$ rows and columns of $\Gamma(t)$.

\begin{align} \label{eq_u4}
    \|\Psi(t_0) - \Psi(0)\|_F &= \|U_k(t_0) \Gamma_k(t_0) U_k(t_0)^\ast - U_k(0) \Gamma_k(0) U_k(0)^\ast\|_F \nonumber\\
        &\leq \|U_k(t_0) \Gamma_k(t_0) U_k(t_0)^\ast -  U_k(t_0) \Gamma_k(0) U_k(t_0)^\ast\|_F \nonumber \\
        &+ \|U_k(t_0) \Gamma_k(0) U_k(t_0)^\ast - U_k(t_0) \Gamma_k(0) U_k(0)^\ast\|_F\nonumber\\
        &+\|U_k(t_0) \Gamma_k(0) U_k(0)^\ast - U_k(0) \Gamma_k(0) U_k(0)^\ast \|_F\nonumber\\
                &\leq \|U_k(t_0)\|_2^2 \times  \|\Gamma_k(t_0) -\Gamma_k(0)\|_F\nonumber \\
                &+ \|U_k(t_0)\|_2 \times \|\Gamma_k(0)\|_2 \times \| U_k(t_0)^\ast - U_k(0)^\ast\|_F \nonumber\\
        &+\|U_k(t_0) -U_k(0)\|_F \times \|\Gamma_k(0)\|_2 \times \|U_k(0)^\ast \|_2 \nonumber\\
        &= \|\Gamma_k(t_0) -\Gamma_k(0)\|_F  +  \sigma_1  \| U_k(t_0)^\ast - U_k(0)^\ast\|_F + \sigma_1 \|U_k(t_0) -U_k(0)\|_F \nonumber\\
                &= \|\Gamma_k(t_0) -\Gamma_k(0)\|_F  + 2\sigma_1 \|U_k(t_0) -U_k(0)\|_F,
        \end{align}
        where second equality holds since $\|U_k(t)\|_2 = 1$ for all $t\geq 0$, and since $\|\Gamma_k(0)\|_2 = \sigma_1$ since $\Gamma_k(0) = M$.

By Lemma \ref{lemma_SinTheta}, we have
\begin{align}\label{eq_u5}
 \|U_k(t_0)U_k^\ast(t_0) -U_k(0)U_k^\ast(0)\|_F  &\stackrel{\textrm{Lemma \ref{lemma_SinTheta}}}{\leq} \frac{\|\Phi(t_0) - \Phi(0)\|_F}{\gamma_k(0) - \gamma_{k+1}(t_0)}\nonumber\\
 &=  \frac{\|B(t_0)\|_F}{\gamma_k(0) - \gamma_{k+1}(t_0)}.
\end{align}
By Weyl's inequality (Lemma \ref{lemma_weyl}), we have that, whenever the event  $\hat{E}^c_\alpha$ occurs,
\begin{align}\label{eq_u6}
\gamma_{k+1}(t_0) &\stackrel{\textrm{Lemma \ref{lemma_weyl}}}{\leq} \gamma_{k+1}(0) + \|B(t_0)\|_2\nonumber\\
&\leq \gamma_{k+1}(0) + 2\sqrt{t_0}(\sqrt{d} + \alpha) \nonumber\\
&= \sigma_{k+1} + 2\sqrt{t_0}(\sqrt{d} + \alpha)
\end{align}
where the second inequality is by the definition of the event $\hat{E}^c_\alpha$.
Thus, \eqref{eq_u6} implies that
\begin{align}\label{eq_u7}
\gamma_k(0)-\gamma_{k+1}(t_0) &\stackrel{\textrm{Eq. \ref{eq_u6}}}{\geq} \gamma_k(0)- \sigma_{k+1} - 2\sqrt{t_0}(\sqrt{d} + \alpha)\nonumber\\
&=  \sigma_k - \sigma_{k+1} - 2\sqrt{t_0}(\sqrt{d} + \alpha)\nonumber\\
&\geq \frac{1}{2}(\sigma_k - \sigma_{k+1}).
\end{align}
where the second inequality holds because $\sigma_k - \sigma_{k+1} \geq 4T \sqrt{d} + 2\alpha$ and $T \geq 1> t_0$.
Thus, plugging \eqref{eq_u7} into \eqref{eq_u5}, we have that whenever the event  $\hat{E}^c_\alpha$ occurs,
\begin{align}\label{eq_u8}
 \|U_k(t_0)U_k^\ast(t_0) -U_k(0)U_k^\ast(0)\|_F  &\leq  \frac{2\|B(t_0)\|_F}{\sigma_k - \sigma_{k+1}}\\
 &\leq  \frac{4\sqrt{t_0}(\sqrt{d} + \alpha)}{\sigma_k - \sigma_{k+1}},
\end{align}
where the second inequality is by the definition of the event $\hat{E}^c_\alpha$.
We also have (by, e.g., Inequality (27) in \cite{mangoubi2022private}) that
\begin{equation}\label{eq_u9}
    \|U_k(t_0) -U_k(0)\|_F \leq \|U_k(t_0)U_k^\ast(t_0) -U_k(0)U_k^\ast(0)\|_F.
\end{equation}
Therefore, plugging in \eqref{eq_u9} into \eqref{eq_u8}, we get that, whenever the event  $\hat{E}^c_\alpha$ occurs,
\begin{equation}\label{eq_u10}
    \|U_k(t_0) -U_k(0)\|_F \leq \frac{4\sqrt{t_0}(\sqrt{d} + \alpha)}{\sigma_k - \sigma_{k+1}}.
\end{equation}
Plugging in \eqref{eq_u10} into \eqref{eq_u4} we get
\begin{align} \label{eq_u11}
    \|\Psi(t_0) - \Psi(0)\|_F \times \mathbbm{1}_{\hat{E}_\alpha^c} &\stackrel{\textrm{Eq. \eqref{eq_u4}}}{\leq} \|\Gamma_k(t_0) -\Gamma_k(0)\|_F \times\mathbbm{1}_{\hat{E}_\alpha^c}  + 2\sigma_1 \|U_k(t_0) -U_k(0)\|_F\times \mathbbm{1}_{\hat{E}_\alpha^c} \nonumber\\
    &\stackrel{\textrm{Eq. \eqref{eq_u10}}}{\leq} \|\Gamma_k(t_0) -\Gamma_k(0)\|_F \times\mathbbm{1}_{\hat{E}_\alpha^c}  + 2\sigma_1 \frac{4\sqrt{t_0}(\sqrt{d} + \alpha)}{\sigma_k - \sigma_{k+1}} \nonumber\\
    &\leq \sqrt{k} \times \|B(t_0)\|_2 \times\mathbbm{1}_{\hat{E}_\alpha^c}  + 2\sigma_1 \frac{4\sqrt{t_0}(\sqrt{d} + \alpha)}{\sigma_k - \sigma_{k+1}} \nonumber\\
        &\leq \sqrt{k} \times  2\sqrt{t_0}(\sqrt{d} + \alpha)  + 2\sigma_1 \frac{4\sqrt{t_0}(\sqrt{d} + \alpha)}{\sigma_k - \sigma_{k+1}} \nonumber\\
                &\leq \sqrt{k} \times  2\sqrt{t_0}(\sqrt{d} + \alpha)  + 2\sigma_1 \frac{4\sqrt{t_0}(\sqrt{d} + \alpha)}{\sigma_k - \sigma_{k+1}} \nonumber\\
                &\leq \sqrt{k} \times  2\sqrt{t_0}(\sqrt{d} + \alpha)  + 2\sigma_1 \frac{4\sqrt{t_0}(\sqrt{d} + \alpha)}{\sqrt{d} + \alpha} \nonumber\\
                 &\leq   2\sqrt{t_0} \sqrt{k} (\sqrt{d} + \alpha)  + 8\sigma_1 \sqrt{t_0} \nonumber\\
                 &=   \sqrt{t_0}(2 \sqrt{k} (\sqrt{d} + \alpha)  + 8\sigma_1),
        \end{align}
        where the fourth inequality is by the definition of the event $\hat{E}^c_\alpha$, and the fifth inequality holds by our assumption that $\sigma_k - \sigma_{k+1} \geq  4T \sqrt{d} + 2\alpha$ and since $T\geq 1$.
\end{proof}

\subsection{Proof of Theorem \ref{thm_rank_k_covariance_approximation_new}} \label{sec_proof_of_rank_k_covariance_approximation_new}

\begin{lemma}[\bf  It\^o derivative $\mathrm{d} u_i(t) u_j^\ast(t)$]\label{Lemma_projection_differntial}
For all $t \in [0,T]$, 
\begin{align*}
    &\mathrm{d}(u_i(t) u_i^\ast(t))\\ 
    & = \sum_{j \neq i} \frac{1}{\gamma_i(t) - \gamma_j(t)}(u_i(t) u_j^\ast(t)\mathrm{d}B_{ij}(t) + u_j(t) u_i^\ast(t)\mathrm{d}B_{ij}^\ast(t))\\
    &\qquad - \sum_{j\neq i} \frac{\mathrm{d}t}{(\gamma_i(t)- \gamma_j(t))^2} (u_i(t) u_i^\ast(t) - u_j(t)u_j^\ast(t)).
\end{align*}
\end{lemma}
The proof of Lemma \ref{Lemma_projection_differntial} is given in   Section \ref{sec_proof_Lemma_projection_differntial} and is an adaptation of the Proof of Lemma 4.1 of \cite{DBM_Neurips} to the setting of complex matrices.

Define $\Delta_{ij}(t) := \gamma_i(t)-\gamma_j(t)$ for all $i,j \in [d]$ and all $t \geq 0$.
Fix any  $\eta \in \mathbb{R}^{d\times d}$, define
$\mu_{ij}(t):=\max(\Delta_{ij}(t), \eta_{ij})$ for $i\leq j$ and $\mu_{ij}(t):= - \mu_{ji}(t)$ for $i>j$.
Further, define the following matrix-valued It\^o diffusion $Z_\eta(t)$ via its It\^o derivative  $\mathrm{d}Z_\eta(t)$:
\begin{eqnarray} \label{eq_orbit_diffusion2}
        \mathrm{d}Z_\eta(t)   & := &   \frac{1}{2}\sum_{i=1}^{d} \sum_{j \neq i} (\lambda_i(t) - \lambda_j(t)) \frac{1}{\mu_{ij}(t)}(u_i(t) u_j^\ast(t)\mathrm{d}B_{ij}(t) + u_j(t) u_i^\ast(t)\mathrm{d}B_{ij}^\ast(t)) \nonumber \\
     & &  + \sum_{i=1}^{d} \sum_{j\neq i} (\lambda_i(t) - \lambda_j(t)) \frac{\mathrm{d}t}{\mu^2_{ij}(t)} u_i(t) u_i^\ast(t),
\end{eqnarray}
with  initial condition $Z_\eta(0):= \Psi(0)$.
Thus, $Z_\eta(t) = \Psi(0)+ \int_0^t \mathrm{d}Z_\eta(s)$ for all $t \geq 0$.
We then integrate  $\mathrm{d}Z_\eta(t)$ over the time interval $[0,T]$, and apply  It\^o's Lemma (Lemma \ref{lemma_ito_lemma_new}) to obtain an expression for the Frobenius norm of this integral.
In the following, we fix $\eta_{ij} = 0$ for all $i,j$.

\begin{lemma} \label{Lemma_integral}
For any $T\geq t_0 \geq 0$, 
\begin{eqnarray*}\mathbb{E}\left[\left\|Z_\eta\left(T\right) -   Z_\eta(t_0)\right \|_F^2 \right] &=& 
 32\int_{t_0}^{T}   \mathbb{E}\left[ \sum_{i=1}^{d}  \sum_{j \neq i}  \frac{(\lambda_i(t) - \lambda_j(t))^2}{\mu^2_{ij}(t)} \mathrm{d}t \right] \nonumber\\ 
     & &  +    T \int_{t_0}^{T}\mathbb{E}\left[\sum_{i=1}^{d}\left(\sum_{j\neq i} \frac{\lambda_i(t) - \lambda_j(t)}{\mu^2_{ij}(t)}\right)^2   \right]\mathrm{d}t.
\end{eqnarray*}
\end{lemma}
The proof of Lemma \ref{Lemma_integral} is a slight modification of the proof of Lemma 4.5 in \cite{DBM_Neurips}, to accommodate the fact that $\lambda_i(t)$ varies over time (which is in contrast to \cite{DBM_Neurips} where $\lambda_i$ is constant), and the fact that our matrices have complex-valued entries. 
The proof of Lemma \ref{Lemma_integral} is given in  Section \ref{Appendix_proof_of_Lemma_integral}.
We now proceed to the proof of Theorem \ref{thm_rank_k_covariance_approximation_new}.

\begin{proof}[Proof of theorem \ref{thm_rank_k_covariance_approximation_new}]
In the following, we set $t_0 = \frac{1}{(kd)^{10} + k\alpha^2 +\sigma_1^2}$.
We first compute the Ito derivative of $\Psi(t) := \sum_{i=1}^d \lambda_i(t) u_i(t) u_i^\ast(t)$:
\begin{align}\label{eq_ito_derivative_b}
    \mathrm{d} \Psi(t) &= \sum_{i=1}^d(\lambda_i(t) + \mathrm{d}\lambda_i(t)) (u_i(t) u_i^\ast(t) + \mathrm{d}(u_i(t) u_i^\ast(t))) - \lambda_i(t) u_i(t) u_i^\ast(t) \nonumber \\
    &=\sum_{i=1}^d \lambda_i(t) \mathrm{d}(u_i(t) u_i^\ast(t))) + (\mathrm{d}\lambda_i(t)) (u_i(t) u_i^\ast(t)) + \mathrm{d}\lambda_i(t) \mathrm{d}( u_i(t) u_i^\ast(t)).
\end{align}
From Lemma \ref{Lemma_projection_differntial}, we have that, for all $t \in [0,T]$, 
\begin{eqnarray}\label{eq_duu}
\mathrm{d}(u_i(t) u_i^\ast(t))&=&        \sum_{j \neq i} \frac{1}{\gamma_i(t) - \gamma_j(t)}(u_i(t) u_j^\ast(t)\mathrm{d}B_{ij}(t) + u_j(t) u_i^\ast(t)\mathrm{d}B_{ij}^\ast(t)) \nonumber \\ & & +  \sum_{j \neq i} \frac{\mathrm{d}t}{(\gamma_i(t)- \gamma_j(t))^2} (u_i(t) u_i^\ast(t) - u_j(t)u_j^\ast(t)).
\end{eqnarray}
For all $i\leq k$, we have that $\lambda_i(t) = \gamma_i(t)$ for all $t\geq0$.
Thus, by \eqref{eq_duu}, we have that
\begin{align}\label{eq_dlambda_du}
&\mathrm{d}\lambda_i(t) \mathrm{d}( u_i(t) u_i^\ast(t))\nonumber\\
&= \mathrm{d}\gamma_i(t) \mathrm{d}( u_i(t) u_i^\ast(t))\nonumber\\
&\stackrel{\textrm{Eq. \eqref{eq_duu}}}{=} \mathrm{d}\gamma_i(t) \bigg[\sum_{j \neq i} \frac{1}{\gamma_i(t) - \gamma_j(t)}(u_i(t) u_j^\ast(t)\mathrm{d}B_{ij}(t) + u_j(t) u_i^\ast(t)\mathrm{d}B_{ij}^\ast(t))\nonumber\\
&\qquad \qquad \qquad \qquad \qquad \qquad+ \sum_{j \neq i} \frac{\mathrm{d}t}{(\gamma_i(t)- \gamma_j(t))^2} (u_i(t) u_i^\ast(t) - u_j(t)u_j^\ast(t))\bigg]\nonumber\\
&= \sum_{j \neq i} \frac{\mathrm{d}\gamma_i(t)}{\gamma_i(t) - \gamma_j(t)}(u_i(t) u_j^\ast(t)\mathrm{d}B_{ij}(t) + u_j(t) u_i^\ast(t) \mathrm{d}B_{ij}^\ast(t) )\nonumber\\
&\qquad \qquad \qquad \qquad \qquad \qquad + \sum_{j \neq i} \frac{\mathrm{d}\gamma_i(t)\mathrm{d}t}{(\gamma_i(t)- \gamma_j(t))^2} (u_i(t) u_i^\ast(t) - u_j(t)u_j^\ast(t))\nonumber\\
&= \sum_{j \neq i} \left(\mathrm{d}B_{i i}(t) +   \sum_{j \neq i} \frac{1}{\gamma_i(t) - \gamma_j(t)} \mathrm{d}t \right)\frac{1}{\gamma_i(t) - \gamma_j(t)}(u_i(t) u_j^\ast(t)\mathrm{d}B_{ij}(t) + u_j(t) u_i^\ast(t)\mathrm{d}B_{ij}^\ast(t))\nonumber\\
&+ \sum_{j \neq i} \left(\mathrm{d}B_{i i}(t) +   \sum_{j \neq i} \frac{1}{\gamma_i(t) - \gamma_j(t)} \mathrm{d}t\right) \frac{\mathrm{d}t}{(\gamma_i(t)- \gamma_j(t))^2} (u_i(t) u_i^\ast(t) - u_j(t)u_j^\ast(t))\nonumber\\
&= \sum_{j \neq i} \left(\mathrm{d}B_{i i}(t)\mathrm{d}B_{ij}(t) +  2 \sum_{j \neq i} \frac{1}{\gamma_i(t) - \gamma_j(t)} \mathrm{d}t \, \mathrm{d}B_{ij}(t) \right)\frac{1}{\gamma_i(t) - \gamma_j(t)}u_i(t) u_j^\ast(t)\nonumber\\
&+\sum_{j \neq i} \left(\mathrm{d}B_{i i}(t)\mathrm{d}B_{ij}^\ast(t) +  2 \sum_{j \neq i} \frac{1}{\gamma_i(t) - \gamma_j(t)} \mathrm{d}t \, \mathrm{d}B_{ij}^\ast(t) \right)\frac{1}{\gamma_i(t) - \gamma_j(t)} u_j(t) u_i^\ast(t)\nonumber\\
&+ \sum_{j \neq i} \left(\mathrm{d}B_{i i}(t) \mathrm{d}t +  2 \sum_{j \neq i} \frac{1}{\gamma_i(t) - \gamma_j(t)} (\mathrm{d}t)^2\right) \frac{1}{(\gamma_i(t)- \gamma_j(t))^2} (u_i(t) u_i^\ast(t) - u_j(t)u_j^\ast(t))\nonumber\\
&= 0, \qquad \qquad \forall i \leq k,
\end{align}
where the last equality holds since, for all $i,j \in [d]$, the Ito differentials $\mathrm{d}B_{i i}(t)\mathrm{d}B_{ij}(t)$ and $\mathrm{d}B_{i i}(t)\mathrm{d}B_{ij}^\ast(t)$ vanish because $\mathrm{d}B_{i i}(t)$ and $\mathrm{d}B_{ij}(t)$ are uncorrelated with mean zero, and the Ito differentials $\mathrm{d}B_{i i}(t) \mathrm{d}t$ and $(\mathrm{d}t)^2$ vanish because they are higher-order terms.
Therefore, plugging in \eqref{eq_dlambda_du} into \eqref{eq_ito_derivative_b}, we have that
\begin{align}\label{eq_ito_derivative}
    \mathrm{d} \Psi(t)  &\stackrel{\textrm{Eq. \eqref{eq_ito_derivative_b}}}{=}\sum_{i=1}^d \lambda_i(t) \mathrm{d}(u_i(t) u_i^\ast(t))) + (\mathrm{d}\lambda_i(t)) (u_i(t) u_i^\ast(t)) + \mathrm{d}\lambda_i(t) \mathrm{d}( u_i(t) u_i^\ast(t)) \nonumber\\
    &=\sum_{i=1}^k \lambda_i(t) \mathrm{d}(u_i(t) u_i^\ast(t))) + (\mathrm{d}\lambda_i(t)) (u_i(t) u_i^\ast(t)) + \mathrm{d}\lambda_i(t) \mathrm{d}( u_i(t) u_i^\ast(t)) \nonumber\\
     &\stackrel{\textrm{Eq. \eqref{eq_dlambda_du}}}{=}\sum_{i=1}^k \lambda_i(t) \mathrm{d}(u_i(t) u_i^\ast(t))) + (\mathrm{d}\lambda_i(t)) (u_i(t) u_i^\ast(t)),
\end{align}
where the second equality holds since  $\lambda_i(t) = 0$ for all $i >k$ and all $t\geq 0$.
Therefore, we have
 \begin{eqnarray} \label{eq_ito_integral_1}
&&\mathbb{E}\left[\left\|\Psi (T) -  \Psi(t_0)\right \|_F^2 \times \mathbbm{1}_{\hat{E}_\alpha^c} \right] \nonumber\\
&&=  \mathbb{E}\left[\left\|\int_{t_0}^T \sum_{i=1}^d   \lambda_i(t) \mathrm{d}( u_i(t) u_i^\ast(t)) + (\mathrm{d} \lambda_i(t)) u_i(t) u_i^\ast(t) \right\|_F^2 \times \mathbbm{1}_{\hat{E}_\alpha^c} \right] \nonumber\\
&&\leq  \mathbb{E}\left[\left\|\int_{t_0}^T \sum_{i=1}^d   \lambda_i(t) \mathrm{d}( u_i(t) u_i^\ast(t))\right\|_F^2 \times \mathbbm{1}_{\hat{E}_\alpha^c} \right]\nonumber \\
& & \qquad+ \mathbb{E}\left[ \left\|\int_{t_0}^T \sum_{i=1}^d  (\mathrm{d} \lambda_i(t)) u_i(t) u_i^\ast(t) \right\|_F^2 \times \mathbbm{1}_{\hat{E}_\alpha^c}\right]  \nonumber\\
&&\leq  \mathbb{E}\left[\left\|\int_{t_0}^T \sum_{i=1}^d   \lambda_i(t) \mathrm{d}( u_i(t) u_i^\ast(t))\right\|_F^2 \times \mathbbm{1}_{\hat{E}_\alpha^c} \right] \nonumber\\
& &\qquad + \mathbb{E}\left[ \left\|\int_{t_0}^T \sum_{i=1}^d  (\mathrm{d} \lambda_i(t)) u_i(t) u_i^\ast(t) \right\|_F^2\right]  \nonumber\\
&&\stackrel{\textrm{Lem. \ref{lemma_gap_concentration}}}{\leq}  \mathbb{E}\left[\left\|Z_\eta\left(T\right) -   Z_\eta(t_0)\right \|_F^2  \times \mathbbm{1}_{\hat{E}_\alpha^c}\right] \nonumber\\
& & \qquad + \mathbb{E}\left[\left\|\int_{t_0}^T \sum_{i=1}^d  (\mathrm{d} \lambda_i(t)) u_i(t) u_i^\ast(t) \right\|_F^2 \times \mathbbm{1}_{\hat{E}_\alpha^c}\right] \nonumber\\
&&\stackrel{\textrm{Lem. \ref{Lemma_integral}}}{\leq}   32\int_{t_0}^{T}  \mathbb{E}\left[ \sum_{i=1}^{d}  \sum_{j \neq i}  \frac{(\lambda_i(t) - \lambda_j(t))^2}{(\gamma_i(t)- \gamma_j(t))^2} \times \mathbbm{1}_{\hat{E}_\alpha^c}\right]\mathrm{d}t \nonumber \\
    & & \quad +    T \int_{t_0}^{T}\mathbb{E}\left[\sum_{i=1}^{d}\left(\sum_{j\neq i} \frac{\lambda_i(t) - \lambda_j(t)}{(\gamma_i(t)- \gamma_j(t))^2}\right)^2 \times \mathbbm{1}_{\hat{E}_\alpha^c}\right] \mathrm{d}t \nonumber\\
    & &\quad + \mathbb{E}\left[\left\|\int_{t_0}^T \sum_{i=1}^d (\mathrm{d} \lambda_i(t))  u_i(t) u_i^\ast(t)\right\|_F^2 \right],
\end{eqnarray}
Where the last inequality holds by Lemma \ref{Lemma_integral}.
Plugging in $\lambda_i(t) = \gamma_i(t)$ for $i \leq k$ and $\lambda_i(t) = 0$  for $i>k$ into \eqref{eq_ito_integral_1}, we have

 \begin{align}\label{eq_u1}
&\mathbb{E}\left[\left\|\Psi (T) -  \Psi(t_0)\right \|_F^2 \times \mathbbm{1}_{\hat{E}_\alpha^c}\right] \nonumber\\
&\leq 2\int_{t_0}^{T}   \mathbb{E}\left[ \sum_{i=1}^{d}  \sum_{j \neq i}  \frac{(\lambda_i(t) - \lambda_j(t))^2}{(\gamma_i(t)- \gamma_j(t))^2} \times \mathbbm{1}_{\hat{E}_\alpha^c}\right]\mathrm{d}t \nonumber \\
    & +    T \int_{t_0}^{T}\mathbb{E}\left[\sum_{i=1}^{d}\left(\sum_{j\neq i} \frac{\lambda_i(t) - \lambda_j(t)}{(\gamma_i(t)- \gamma_j(t))^2}\right)^2 \times \mathbbm{1}_{\hat{E}_\alpha^c}\right] \mathrm{d}t + \mathbb{E}\left[\left\|\int_{t_0}^T \sum_{i=1}^d (\mathrm{d} \lambda_i(t))  u_i(t) u_i^\ast(t)\right\|_F^2  \times \mathbbm{1}_{\hat{E}_\alpha^c}\right] \nonumber\\
&\leq  32\int_{t_0}^{T}   \mathbb{E}\left[ \sum_{i=1}^{k} \left( k +  \sum_{j >k}  \frac{(\gamma_i(t))^2}{(\gamma_i(t)- \gamma_j(t))^2} \right)  \times \mathbbm{1}_{\hat{E}_\alpha^c}\right]\mathrm{d}t \nonumber \\
    & +    T \int_{t_0}^{T}\mathbb{E}\left[\sum_{i=1}^{k}\left(\sum_{j \neq i: j\leq k} \frac{1}{\gamma_i(t)- \gamma_j(t)} + \sum_{j> k} \frac{\gamma_i(t)}{(\gamma_i(t)- \gamma_j(t))^2}\right)^2  \times \mathbbm{1}_{\hat{E}_\alpha^c}\right] \mathrm{d}t \nonumber\\
    &+ \mathbb{E}\left[\left\|\int_{t_0}^T \sum_{i=1}^d (\mathrm{d} \lambda_i(t))  u_i(t) u_i^\ast(t)\right\|_F^2  \times \mathbbm{1}_{\hat{E}_\alpha^c}\right]  \nonumber\\
    &\leq  32\int_{t_0}^{T}   \mathbb{E}\left[ \sum_{i=1}^{k} \left( k +  \sum_{j >k}  \frac{(\gamma_i(t))^2}{(\gamma_i(t)- \gamma_j(t))^2} \right)  \times \mathbbm{1}_{\hat{E}_\alpha^c}\right]\mathrm{d}t \nonumber \\
    & +    4T \int_{t_0}^{T}\sum_{i=1}^{k}\mathbb{E}\left[\left(\sum_{j \neq i: j\leq k} \frac{1}{\gamma_i(t)- \gamma_j(t)}\right)^2  \times \mathbbm{1}_{\hat{E}_\alpha^c}\right] +  \mathbb{E}\left[\left(\sum_{j> k} \frac{\gamma_i(t)}{(\gamma_i(t)- \gamma_j(t))^2}\right)^2  \times \mathbbm{1}_{\hat{E}_\alpha^c}\right] \mathrm{d}t\nonumber\\
    &+ \mathbb{E}\left[\left\|\int_{t_0}^T \sum_{i=1}^d (\mathrm{d} \lambda_i(t))  u_i(t) u_i^\ast(t)\right\|_F^2  \times \mathbbm{1}_{\hat{E}_\alpha^c}\right].
\end{align}

\paragraph{Bounding the second moment of the inverse gaps:}
By Corollary \ref{lemma_gaps_any_start} we have that for every $1\leq i < j \leq d$,
\begin{equation}\label{eq_v1}
    \mathbb{P}\left(\left\{ \gamma_i(t) - \gamma_{j}(t) \leq (j-i) \times s \frac{\sqrt{t}}{\mathfrak{b}\sqrt{d}} \right\} \cap \hat{E}_\alpha^c\right) \leq  s^3 \qquad \forall s>0, t>0.
\end{equation}
Thus, for $t \leq T$,
\begin{align} \label{eq_second_inverse_moment}
    &\mathbb{E}\left[\frac{1}{(\gamma_i(t) - \gamma_{j}(t))^2}  \times \mathbbm{1}_{\hat{E}_\alpha^c}\right]\nonumber\\
    &\leq \mathbb{E}\left[\frac{1}{(\gamma_i(t) - \gamma_{j}(t))^2} \times \mathbbm{1}\left\{\gamma_i(t) - \gamma_{j}(t) \leq (j-i) \times \frac{\sqrt{t}}{ \mathfrak{b}\sqrt{d}}\right\}  \times \mathbbm{1}_{\hat{E}_\alpha^c}\right] + \frac{\mathfrak{b}^2 d}{(j-i)^2 t} \nonumber\\
    &= \int_{\frac{\mathfrak{b}^2d}{(j-i)^2 t}}^{\infty}   \mathbb{P}\left(\left\{\frac{1}{(\gamma_i(t) - \gamma_{i+1}(t))^2} \geq s \right\} \cap \hat{E}_\alpha^c \right) \mathrm{d}s + \frac{\mathfrak{b}^2 d}{(j-i)^2 t} \nonumber\\
    &= \int_{\frac{\mathfrak{b}^2 d}{(j-i)^2 t}}^{\infty}   \mathbb{P}\left(\left\{(\gamma_i(t) - \gamma_{i+1}(t))^2 \leq s^{-1}  \right\} \cap \hat{E}_\alpha^c \right)  \mathrm{d}s + \frac{\mathfrak{b}^2 d}{(j-i)^2 t} \nonumber\\
    &= \int_{\frac{\mathfrak{b}^2 d}{(j-i)^2 t}}^{\infty}   \mathbb{P}\left(\left\{\gamma_i(t) - \gamma_{i+1}(t) \leq s^{-\frac{1}{2}}\right\} \cap \hat{E}_\alpha^c\right) \mathrm{d}s + \frac{\mathfrak{b}^2 d}{(j-i)^2 t}\nonumber\\
    &\stackrel{\textrm{Eq. \eqref{eq_v1}}}{\leq} \int_{\frac{\mathfrak{b}^2 d}{(j-i)^2 t}}^{\infty}  \left(\frac{\mathfrak{b}^2 d}{(j-i)^2 t}\right)^{\frac{3}{2}} s^{-\frac{3}{2}} \mathrm{d}s + \frac{\mathfrak{b}^2 d}{(j-i)^2 t} \nonumber\\
    & = -\frac{3}{2} \left(\frac{\mathfrak{b}^2 d}{(j-i)^2 t}\right)^{\frac{3}{2}}  s^{-\frac{1}{2}} \bigg|_{\frac{\mathfrak{b}^2 d}{(j-i)^2 t}}^{\infty} + \frac{\mathfrak{b}^2 d}{(j-i)^2 t}\nonumber\\
      & = \frac{3}{2} \frac{\mathfrak{b}^2 d}{(j-i)^2 t} + \frac{\mathfrak{b}^2 d}{(j-i)^2 t}\nonumber\\
        & \leq 3 \frac{\mathfrak{b}^2 d}{(j-i)^2 t}.
\end{align}
Thus, for any $t_0>0$,
\begin{align*}
    \mathbb{E}\left[\int_{t_0}^T \frac{1}{(\gamma_i(t) - \gamma_{j}(t))^2} \mathrm{d}t  \times \mathbbm{1}_{\hat{E}_\alpha^c}\right] &=\int_{t_0}^T \mathbb{E}\left[ \frac{1}{(\gamma_i(t) - \gamma_{j}(t))^2}  \times \mathbbm{1}_{\hat{E}_\alpha^c}\right]  \mathrm{d}t\\
     &\leq  3 \int_{t_0}^T  \frac{\mathfrak{b}^2 d}{(j-i)^2 t} \mathrm{d}t\\
     &=  3 \frac{\mathfrak{b}^2 d}{(j-i)^2} \log(t)|_{t_0}^T\\
     &= 3 \frac{\mathfrak{b}^2 d}{(j-i)^2}\times (\log(T) - \log(t_0)).
    \end{align*}

\paragraph{Bounding the term $\mathbb{E}\left[\left\|\int_0^T \sum_{i=1}^d (\mathrm{d} \lambda_i(t))  u_i(t) u_i^\ast(t)\right\|_F^2  \times \mathbbm{1}_{\hat{E}_\alpha^c}\right]$:}

 For $i > k$, $\mathrm{d} \lambda_i(t) = 0$.
 For $i \leq k$, we have $\lambda_i(t) = \gamma_i(t)$ and thus,
\begin{align*}
  (\mathrm{d} \lambda_i(t))  u_i(t) u_i^\ast(t) &=   (\mathrm{d} \gamma_i(t))  u_i(t) u_i^\ast(t)\\
  &= \left(\mathrm{d}B_{i i}(t) +  2 \sum_{j \neq i} \frac{1}{\gamma_i(t) - \gamma_j(t)} \mathrm{d}t \right) u_i(t) u_i^\ast(t).
\end{align*}
 Thus,
 \begin{align}\label{eq_a1}
    &\mathbb{E}\left[\left\|\int_{t_0}^T \sum_{i=1}^d (\mathrm{d} \lambda_i(t))  u_i(t) u_i^\ast(t)\right\|_F^2  \times \mathbbm{1}_{\hat{E}_\alpha^c} \right]
    =  \mathbb{E}\left[\left\|\int_{t_0}^T \sum_{i=1}^k (\mathrm{d} \gamma_i(t))  u_i(t) u_i^\ast(t)\right\|_F^2 \times \mathbbm{1}_{\hat{E}_\alpha^c} \right] \nonumber\\
     &=\mathbb{E}\left[\left\|\int_{t_0}^T \sum_{i=1}^k \left(\mathrm{d}B_{i i}(t) +  2 \sum_{j \neq i} \frac{1}{\gamma_i(t) - \gamma_j(t)} \mathrm{d}t \right) u_i(t) u_i^\ast(t)\right\|_F^2 \times \mathbbm{1}_{\hat{E}_\alpha^c} \right]\nonumber\\
     &\leq 3\mathbb{E}\left[\left\|\int_{t_0}^T \sum_{i=1}^k  u_i(t) u_i^\ast(t) \mathrm{d}B_{i i}(t) \right\|_F^2 \times \mathbbm{1}_{\hat{E}_\alpha^c} \right] \nonumber\\
     &\qquad +  6 \mathbb{E}\left[\left\|\int_{t_0}^T \sum_{i=1}^k \sum_{j \neq i} \frac{1}{\gamma_i(t) - \gamma_j(t)}  u_i(t) u_i^\ast(t) \mathrm{d}t\right\|_F^2 \times \mathbbm{1}_{\hat{E}_\alpha^c} \right].
 \end{align}
    To bound the first term on the r.h.s. of \eqref{eq_a1}, we note that since $\mathrm{d}B_{ii}(t)$ are i.i.d. for all $t>0$, we have if we plug in $f(X):= \|X\|_F^2 = \sum_{i=1}^d \sum_{j=1}^d X_{ij}^2$ into Ito's Lemma (Lemma \ref{lemma_ito_lemma_new}) that
    \begin{align}\label{eq_a2}
         \mathbb{E}\left[\left\|\int_{t_0}^T \sum_{i=1}^d  u_i(t) u_i^\ast(t) \mathrm{d}B_{i i}(t) \right\|_F^2 \times \mathbbm{1}_{\hat{E}_\alpha^c} \right]
         &= \mathbb{E}\left[\left\|\int_{t_0}^T \sum_{i=1}^d  u_i(t) u_i^\ast(t) \mathrm{d}B_{i i}(t) \right\|_F^2 \right] \nonumber\\
         &\stackrel{\textrm{Lemma \ref{lemma_ito_lemma_new}}}{=} \sum_{i=1}^d \int_{t_0}^T \|u_i(t) u_i^\ast(t) \|_F^2 \mathrm{d}t \nonumber\\
         & = (T-t_0) d.
    \end{align}
To bound the second term on the R.H.S. of \eqref{eq_a1}, we have
\begin{align}\label{eq_a3}
    &\mathbb{E}\left[\left\|\int_{t_0}^T \sum_{i=1}^k \sum_{j \neq i} \frac{1}{\gamma_i(t) - \gamma_j(t)}  u_i(t) u_i^\ast(t) \mathrm{d}t\right\|_F^2 \times \mathbbm{1}_{\hat{E}_\alpha^c} \right]\nonumber\\
    &\leq    \mathbb{E}\left[\int_{t_0}^T \left\|\sum_{i=1}^k \sum_{j \neq i} \frac{1}{\gamma_i(t) - \gamma_j(t)}  u_i(t) u_i^\ast(t) \right\|_F^2 \mathrm{d}t \times \int_{t_0}^T 1^2 \mathrm{d}t \times \mathbbm{1}_{\hat{E}_\alpha^c} \right]\nonumber\\
       &=    (T-t_0)\mathbb{E}\left[\int_{t_0}^T \left\|\sum_{i=1}^k \sum_{j \neq i} \frac{1}{\gamma_i(t) - \gamma_j(t)}  u_i(t) u_i^\ast(t) \right\|_F^2 \mathrm{d}t \times \mathbbm{1}_{\hat{E}_\alpha^c} \right]\nonumber\\
             &=    (T-t_0)\mathbb{E}\left[\int_{t_0}^T \sum_{i=1}^k \left\| \sum_{j \neq i} \frac{1}{\gamma_i(t) - \gamma_j(t)}  u_i(t) u_i^\ast(t) \right\|_F^2 \mathrm{d}t \times \mathbbm{1}_{\hat{E}_\alpha^c} \right]\nonumber\\
                          &=    (T-t_0)\mathbb{E}\left[\int_{t_0}^T \sum_{i=1}^k  \left(\sum_{j \neq i} \frac{1}{\gamma_i(t) - \gamma_j(t)}\right)^2 \left\| u_i(t) u_i^\ast(t) \right\|_F^2 \mathrm{d}t \times \mathbbm{1}_{\hat{E}_\alpha^c} \right],\nonumber\\
                            &=    (T-t_0)\int_{t_0}^T \sum_{i=1}^k\mathbb{E}\left[  \left(\sum_{j \neq i} \frac{1}{\gamma_i(t) - \gamma_j(t)}\right)^2 \times \mathbbm{1}_{\hat{E}_\alpha^c} \right] \mathrm{d}t,
\end{align}
where the first inequality is by the Cauchy-Schwartz inequality,
and the second equality holds since $\langle u_i(t) u_i^\ast(t), \, \, u_\ell(t) u_\ell^\ast(t) \rangle = 0$ for all $i \neq \ell$.
Therefore, plugging in \eqref{eq_a2} and \eqref{eq_a3} into \eqref{eq_a1}, we have
 \begin{align}\label{eq_a4}
    &\mathbb{E}\left[\left\|\int_{t_0}^T \sum_{i=1}^d (\mathrm{d} \lambda_i(t))  u_i(t) u_i^\ast(t)\right\|_F^2 \times \mathbbm{1}_{\hat{E}_\alpha^c} \right]\nonumber\\
     &\leq 3(T-t_0) d +  6 (T-t_0)\int_{t_0}^T \sum_{i=1}^k\mathbb{E}\left[  \left(\sum_{j \neq i} \frac{1}{\gamma_i(t) - \gamma_j(t)}\right)^2 \times \mathbbm{1}_{\hat{E}_\alpha^c} \right] \mathrm{d}t.
 \end{align}

\paragraph{Bounding the term $\mathbb{E}\left[  \left(\sum_{j \neq i} \frac{1}{\gamma_i(t) - \gamma_j(t)}\right)^2 \times \mathbbm{1}_{\hat{E}_\alpha^c} \right]$:}

Consider any subset $S \subseteq \{1,\ldots,d\}$.
Then

\begin{align}\label{eq_v2}
    &\mathbb{E}\left[  \left(\sum_{j \in S, j \neq i} \frac{1}{\gamma_i(t) - \gamma_j(t)}\right)^2 \times \mathbbm{1}_{\hat{E}_\alpha^c} \right] \nonumber\\
    &=   \mathbb{E}\left[  \sum_{j \in S, j \neq i} \, \, \sum_{\ell \in S, \ell \neq i} \frac{1}{(\gamma_i(t) - \gamma_j(t))(\gamma_i(t) - \gamma_\ell(t))}  \times \mathbbm{1}_{\hat{E}_\alpha^c} \right] \nonumber\\
    & =   \sum_{j \in S, j \neq i} \, \, \sum_{\ell \in S, \ell \neq i}   \mathbb{E}\left[ \frac{1}{(j-i)(
    \ell - i)\frac{\gamma_i(t) - \gamma_j(t)}{j-i}\times \frac{\gamma_i(t) - \gamma_\ell(t)}{\ell-i}}  \times \mathbbm{1}_{\hat{E}_\alpha^c} \right] \nonumber\\
     & =   \sum_{j \in S, j \neq i} \, \, \sum_{\ell \in S, \ell \neq i}  \frac{1}{(j-i)(
    \ell - i)}\mathbb{E}\left[ \frac{1}{\frac{\gamma_i(t) - \gamma_j(t)}{j-i} \times \frac{\gamma_i(t) - \gamma_\ell(t)}{\ell-i}}  \times \mathbbm{1}_{\hat{E}_\alpha^c} \right] \nonumber\\
    &\leq 2 \sum_{j \in S, j \neq i} \, \, \sum_{\ell \in S, \ell \neq i}   \frac{1}{|(j-i)(
    \ell - i)|}\left(\mathbb{E}\left[ \frac{1}{\left(\frac{\gamma_i(t) - \gamma_j(t)}{j-i}\right)^2} \times \mathbbm{1}_{\hat{E}_\alpha^c} \right] + \mathbb{E}\left[ \frac{1}{ \left(\frac{\gamma_i(t) - \gamma_\ell(t)}{\ell-i}\right)^2}  \times \mathbbm{1}_{\hat{E}_\alpha^c} \right]\right) \nonumber\\
    &\stackrel{\textrm{Eq. \eqref{eq_second_inverse_moment}}}{\leq}
    2 \sum_{j \in S, j \neq i} \, \, \sum_{\ell \in S, \ell \neq i}   \frac{1}{|(j-i)(
    \ell - i)|}\left(3 \mathfrak{b}^2\frac{ d}{t} + 3 \mathfrak{b}^2\frac{ d}{t}\right) \nonumber\\
    &=12 \mathfrak{b}^2 \frac{d}{t}  \sum_{j \in S, j \neq i} \, \, \sum_{\ell \in S, \ell \neq i}  \frac{1}{|(j-i)(
    \ell - i)|} \nonumber\\
    &=12 \mathfrak{b}^2 \frac{d}{t} \sum_{j \in S, j \neq i} \frac{1}{|j-i|} \sum_{\ell \in S, \ell \neq i}  \frac{1}{|
    \ell - i|} \nonumber\\
    &\leq 12 \mathfrak{b}^2 \frac{d}{t}\sum_{\ell \in S, \ell \neq i}  \frac{1}{|j-i|} \log d  \nonumber\\
    &\leq 12 \mathfrak{b}^2 \frac{d}{t} \log^2 d .
\end{align}

\paragraph{Completing the proof:}

 \begin{align} \label{eq_u2}
&\mathbb{E}\left[\left\|\Psi (T) -  \Psi(t_0)\right \|_F^2 \times   \mathbbm{1}_{\hat{E}_\alpha^c}\right]\nonumber\\ 
    &\stackrel{\textrm{Eq. \eqref{eq_u1}, \eqref{eq_a4}}}{\leq} 32\int_{t_0}^{T}   \mathbb{E}\left[ \sum_{i=1}^{k} \left( k +  \sum_{j >k}  \frac{(\gamma_i(t))^2}{(\gamma_i(t) - \gamma_j(t))^2} \right)  \times \mathbbm{1}_{\hat{E}_\alpha^c}\right]\mathrm{d}t \nonumber \\
    & +    4T \int_{t_0}^{T}\sum_{i=1}^{k}\mathbb{E}\left[\left(\sum_{j \neq i: j\leq k} \frac{1}{\gamma_i(t) - \gamma_j(t)}\right)^2 \times \mathbbm{1}_{\hat{E}_\alpha^c}\right] +  \mathbb{E}\left[\left(\sum_{j> k} \frac{\gamma_i(t)}{(\gamma_i(t) - \gamma_j(t))^2}\right)^2  \times \mathbbm{1}_{\hat{E}_\alpha^c}\right] \mathrm{d}t\nonumber\\
    &+ 3(T-t_0) d +  6 (T-t_0)\int_{t_0}^T \sum_{i=1}^k\mathbb{E}\left[  \left(\sum_{j \neq i} \frac{1}{\gamma_i(t) - \gamma_j(t)}\right)^2 \times \mathbbm{1}_{\hat{E}_\alpha^c} \right] \mathrm{d}t. \nonumber\\
    &\stackrel{\textrm{Eq. \eqref{eq_v2}}}{\leq} 
    32\int_{t_0}^{T}   \mathbb{E}\left[ \sum_{i=1}^{k} \left( k +  \sum_{j >k}  \frac{(\gamma_i(t))^2}{(\gamma_i(t) - \gamma_j(t))^2} \right)  \times \mathbbm{1}_{\hat{E}_\alpha^c}\right]\mathrm{d}t \nonumber \\
    & +    4T \int_{t_0}^{T}\sum_{i=1}^{k}12 \mathfrak{b}^2 \frac{d}{t} \log^2 d  +  \mathbb{E}\left[\left(\sum_{j> k} \frac{\gamma_i(t)}{(\gamma_i(t) - \gamma_j(t))^2}\right)^2  \times \mathbbm{1}_{\hat{E}_\alpha^c}\right] \mathrm{d}t\nonumber\\
    &+ 3(T-t_0) d +  6 (T-t_0)\int_{t_0}^T \sum_{i=1}^k 12 \mathfrak{b}^2  \frac{d}{t} \log^2 d  \nonumber\\
    &\leq     32\int_{t_0}^{T}   \mathbb{E}\left[ \sum_{i=1}^{k} \left( k +  \sum_{j >k}  16 \frac{\sigma_k^2}{(\sigma_k-\sigma_{k+1})^2} \right)  \right]\mathrm{d}t \nonumber \\
    & +    48 \mathfrak{b}^2 T kd (\log T- \log t_0) \log^2 d  +  \int_{t_0}^{T}  \mathbb{E}\left[\left(\sum_{j> k} 16 \frac{\sigma_k}{(\sigma_k-\sigma_{k+1}) \sqrt{T} \sqrt{d}}\right)^2 \right] \mathrm{d}t\nonumber\\
    &+ 3(T-t_0) d +  6 (T-t_0) 12 \mathfrak{b}^2  kd (\log^2 d) (\log(T)-\log(t_0)) \nonumber\\
     &\leq     32T\left( k^2 +   16 kd \frac{\sigma_k^2}{(\sigma_k-\sigma_{k+1})^2} \right)\nonumber \\
    & +     48 \mathfrak{b}^2 T kd (\log T- \log t_0) \log^2 d  +  \frac{16^2}{T} d \frac{\sigma_k^2}{(\sigma_k-\sigma_{k+1})^2}(T-t_0) \nonumber\\
    &+ 3(T-t_0) d +  6 (T-t_0) 12 \mathfrak{b}^2  kd (\log^2 d) (\log(T)-\log(t_0))\nonumber\\
    &\leq \frac{1}{2}10^4 \mathfrak{b}^2  kd T  \frac{\sigma_k^2}{(\sigma_k -\sigma_{k+1})^2} (\log^3 d) \log(\sigma_1 + T),
\end{align}
where the last inequality of \eqref{eq_u2} holds because $-\log t_0 \leq 20 \log(d)$ since 
$t_0 = \frac{1}{(kd)^{10} + k\alpha^2 +\sigma_1^2} \geq \frac{1}{(kd)^{10} + 400k\log(\sigma_1 + T) +\sigma_1^2} $.
Moreover, the third inequality of \eqref{eq_u2} holds because Lemma \ref{lemma_gap_concentration} implies that, since  $\sigma_k - \sigma_{k+1} \geq \sqrt{T} \sqrt{d} + 40\log^{\frac{1}{2}}(\sigma_1 + T)$, whenever $\hat{E}_\alpha^c$ occurs, we have 
\begin{align*}
    \gamma_k(t) - \gamma_{k+1}(t) \geq  \frac{1}{2}((\sigma_k - \sigma_{k+1}) - \alpha) \geq \frac{1}{4}((\sigma_k - \sigma_{k+1}) - \alpha) 
    &\geq \frac{1}{4}\sqrt{T} \sqrt{d} \qquad \forall t \geq 0,
    \end{align*}
    because $\alpha = 20\log^{\frac{1}{2}}(\sigma_1 + T)$.
    Therefore, plugging in \eqref{eq_u2} into Lemma \ref{lemma_utility_rare_event}, we have that
 \begin{align} \label{eq_u3}
    &\mathbb{E}[\|\Psi(T) - \Psi(0)\|_F^2] \nonumber\\
    &\stackrel{\textrm{Lemma \ref{lemma_utility_rare_event}}}{\leq} 4\mathbb{E}[\|\Psi(T) - \Psi(0)\|_F^2 \times \mathbbm{1}_{\hat{E}_\alpha^c}] +d + \frac{T}{d^{200}} \nonumber\\   
    &\leq 16\mathbb{E}[\|\Psi(T) - \Psi(t_0)\|_F^2 \times \mathbbm{1}_{\hat{E}_\alpha^c}] +  16\mathbb{E}[\|\Psi(t_0) - \Psi(0)\|_F^2 \times \mathbbm{1}_{\hat{E}_\alpha^c}] + d + \frac{T}{d^{200}}  \nonumber\\   
    &\stackrel{\textrm{Eq. \eqref{eq_u2}}}{\leq} \frac{1}{4} 10^6 \mathfrak{b}^2  kd T  \frac{\sigma_k^2}{(\sigma_k-\sigma_{k+1})^2} (\log^3 d) + 16\mathbb{E}[\|\Psi(t_0) - \Psi(0)\|_F^2 \times \mathbbm{1}_{\hat{E}_\alpha^c}]  +d + \frac{T}{d^{200}} \nonumber\\
    &\leq \frac{1}{2} 10^6 \mathfrak{b}^2  kd T  \frac{\sigma_k^2}{(\sigma_k-\sigma_{k+1})^2} (\log^3 d) + \mathbb{E}[\|\Psi(t_0) - \Psi(0)\|_F^2 \times \mathbbm{1}_{\hat{E}_\alpha^c}] \nonumber\\
       &\stackrel{\textrm{Lemma \ref{lemma_t0}}}{\leq} \frac{1}{2}  10^6 \mathfrak{b}^2  kd T  \frac{\sigma_k^2}{(\sigma_k-\sigma_{k+1})^2} (\log^3 d) + 40^2 \nonumber\\
       &\leq  10^6 \mathfrak{b}^2  kd T  \frac{\sigma_k^2}{(\sigma_k-\sigma_{k+1})^2} (\log^3 d) \log(\sigma_1 + T),
\end{align}
where the fifth inequality holds by Lemma \ref{lemma_t0} since $t_0 = \frac{1}{(kd)^{10} + k\alpha^2 +\sigma_1^2}$.
\end{proof}

\section{Eigenvalue gap comparison result: Proof of Lemma \ref{lemma_gap_comparison}} \label{sec_gap_comparison_proof}

\begin{proposition}\label{prop_stochastic_derivative_comparison}
Consider any strong solutions $\gamma, \xi$ to \eqref{eq_DBM_eigenvalues}, for any $\beta \geq 1$. 
Suppose that for some $i \in [d]$ and at some time $t \geq 0$,
\begin{equation} \label{eq_z3}
\gamma_i(t) - \gamma_{i+1}(t) = \xi_i(t) - \xi_{i+1}(t) >0
\end{equation}
and
\begin{equation}  \label{eq_z4}
\gamma_j(t) - \gamma_{j+1}(t) \geq \xi_j(t) - \xi_{j+1}(t) > 0 \qquad \forall j \in [d-1].
\end{equation}
Then
\begin{equation}  \label{eq_z5}
\mathrm{d} \gamma_i(t) - \mathrm{d}  \gamma_{i+1}(t) \geq \mathrm{d} \xi_i(t) - \mathrm{d}  \xi_{i+1}(t).
\end{equation}

\end{proposition}

\begin{proof}
First, note that for any numbers $b \geq c >0$ and all $a >0$ we have that
\begin{equation} \label{eq_z1}
\frac{1}{a+b} - \frac{1}{b} \geq \frac{1}{a+c} - \frac{1}{c}.
\end{equation}

\paragraph{Bounding the repulsion forces when  $j > i+1$:}
For any $j > i+1$ we have that by \eqref{eq_z1} (setting $a=\gamma_i(t) - \gamma_{i+1}(t)$,  $b=\gamma_{i+1}(t) - \gamma_j(t)$, and $c=\xi_{i+1}(t) - \xi_j(t)$,  and noting that \eqref{eq_z4} implies that $b \geq c >0$ since $j > i+1$),
\begin{align} \label{eq_z2}
&\frac{1}{\gamma_i(t) - \gamma_{i+1}(t) +(\gamma_{i+1}(t) - \gamma_j(t))} - \frac{1}{\gamma_{i+1}(t) - \gamma_j(t)} \nonumber\\
&\geq \frac{1}{\gamma_i(t) - \gamma_{i+1}(t) +(\xi_{i+1}(t) - \xi_j(t))} - \frac{1}{\xi_{i+1}(t) - \xi_j(t)}.
\end{align}
Plugging in \eqref{eq_z3} into \eqref{eq_z2}, we get that
\begin{align} \label{eq_z6}
&\frac{1}{\gamma_i(t) - \gamma_{i+1}(t) +(\gamma_{i+1}(t) - \gamma_j(t))} - \frac{1}{\gamma_{i+1}(t) - \gamma_j(t)}\nonumber\\
&\geq \frac{1}{\xi_i(t) - \xi_{i+1}(t) +(\xi_{i+1}(t) - \xi_j(t))} - \frac{1}{\xi_{i+1}(t) - \xi_j(t)}.
\end{align}
Simplifying \eqref{eq_z6}, we get
\begin{equation} \label{eq_z7}
\frac{1}{\gamma_i(t) - \gamma_j(t)} - \frac{1}{\gamma_{i+1}(t) - \gamma_j(t)} \geq \frac{1}{\xi_i(t) - \xi_j(t)} - \frac{1}{\xi_{i+1}(t) - \xi_j(t)} \qquad \qquad \forall j > i+1.
\end{equation}

\paragraph{Bounding the repulsion forces when $j < i$:}  Next, consider any $j < i$.
Then by \eqref{eq_z1}  (setting $a = \gamma_i(t) - \gamma_{i+1}(t)$, $b=\gamma_{j}(t) - \gamma_i(t)$, and $c=\xi_{j}(t) - \xi_i(t)$ into \eqref{eq_z1}, and noting that \eqref{eq_z4} implies that $b \geq c >0$ since $j < i$), we get
\begin{align} \label{eq_z8}
&\frac{1}{\gamma_{j}(t) - \gamma_i(t) + (\gamma_{i}(t) - \gamma_{i+1}(t))} - \frac{1}{\gamma_{j}(t) - \gamma_i(t)}\nonumber\\
&\geq \frac{1}{\xi_{j}(t) - \xi_i(t) + (\gamma_{i}(t) - \gamma_{i+1}(t))} - \frac{1}{\xi_{j}(t) - \xi_i(t)}.
\end{align}
Then plugging in \eqref{eq_z3} into \eqref{eq_z8}, we have
\begin{align} \label{eq_z9}
&\frac{1}{\gamma_{j}(t) - \gamma_i(t) + (\gamma_{i}(t) - \gamma_{i+1}(t))} - \frac{1}{\gamma_{j}(t) - \gamma_i(t)}\nonumber\\
&\geq \frac{1}{\xi_{j}(t) - \xi_i(t) + (\xi_{i}(t) - \xi_{i+1}(t))} - \frac{1}{\xi_{j}(t) - \xi_i(t)}.
\end{align}
Simplifying \eqref{eq_z9}, we get
\begin{equation} \label{eq_z10}
\frac{1}{\gamma_i(t) - \gamma_{j}(t)} -\frac{1}{\gamma_{i+1}(t) - \gamma_{j}(t)}  \geq   \frac{1}{\xi_i(t)- \xi_{j}(t)} - \frac{1}{\xi_{i+1}(t)- \xi_{j}(t)} \qquad \qquad \forall j < i.
\end{equation}
Therefore,  \eqref{eq_z7} and \eqref{eq_z10} together imply that
\begin{equation} \label{eq_z11}
\frac{1}{\gamma_i(t) - \gamma_{j}(t)} -\frac{1}{\gamma_{i+1}(t) - \gamma_{j}(t)}  \geq   \frac{1}{\xi_i(t)- \xi_{j}(t)} - \frac{1}{\xi_{i+1}(t)- \xi_{j}(t)} \qquad \qquad \forall j \in [d] \backslash \{i, i+1\}.
\end{equation}

\paragraph{Bounding the gap derivative:}
By \eqref{eq_DBM_eigenvalues} and \eqref{eq_z11} we have that
\begin{align*}
&\mathrm{d} \gamma_i(t) - \mathrm{d} \gamma_{i+1} (t) \nonumber\\
&\stackrel{\textrm{Eq. \eqref{eq_DBM_eigenvalues}}}{=}  \left(\mathrm{d}B_{i, i}(t) +  \beta \sum_{j \neq i} \frac{1}{\gamma_i(t) - \gamma_j(t)} \mathrm{d}t \right) -  \left(\mathrm{d}B_{i+1, i+1}(t) +  \beta \sum_{j \neq i+1} \frac{1}{\gamma_{i+1}(t) - \gamma_j(t)} \mathrm{d}t \right)\\
&= \mathrm{d}B_{i, i}(t)  - \mathrm{d}B_{i+1, i+1}(t)  + \beta  \mathrm{d}t  \sum_{j \in [d] \backslash \{i, i+1\}}   \frac{1}{\gamma_i(t) - \gamma_j(t)} -  \frac{1}{\gamma_{i+1}(t) - \gamma_j(t)}\\
&\stackrel{\textrm{Eq. \eqref{eq_z11}}}{\geq}   \mathrm{d}B_{i, i}(t)  - \mathrm{d}B_{i+1, i+1}(t)  + \beta  \mathrm{d}t  \sum_{j \in [d] \backslash \{i, i+1\}}  \frac{1}{\xi_i(t)- \xi_{j}(t)} - \frac{1}{\xi_{i+1}(t)- \xi_{j}(t)}\\
&=  \left(\mathrm{d}B_{i, i}(t) +  \beta \sum_{j \neq i} \frac{1}{\xi_i(t) - \xi_j(t)} \mathrm{d}t \right) -  \left(\mathrm{d}B_{i+1, i+1}(t) +  \beta \sum_{j \neq i+1} \frac{1}{\xi_{i+1}(t) - \xi_j(t)} \mathrm{d}t \right)\\
&= \mathrm{d} \xi_i(t) - \mathrm{d} \xi_{i+1} (t).
\end{align*}
This proves \eqref{eq_z5} and completes the proof of the proposition.
\end{proof}

\begin{proposition}\label{prop_stochastic_derivative_comparison_equality}
Consider any strong solutions $\gamma, \xi$ to \eqref{eq_DBM_eigenvalues}, for any $\beta \geq 1$. 
Suppose that for some $i \in [d]$ and at some time $t \geq 0$,
\begin{equation}  \label{eq_y1}
\gamma_i(t) - \gamma_{i+1}(t) = \xi_i(t) - \xi_{i+1}(t) > 0 \qquad \forall i \in [d-1].
\end{equation}
Then
\begin{equation}  \label{eq_y2}
\mathrm{d} \gamma_i(t) - \mathrm{d}  \gamma_{i+1}(t) = \mathrm{d} \xi_i(t) - \mathrm{d}  \xi_{i+1}(t) \qquad \forall i \in [d-1].
\end{equation}

\begin{proof}
\begin{align*}
&\mathrm{d} \gamma_i(t) - \mathrm{d} \gamma_{i+1} (t) \nonumber\\
&\stackrel{\textrm{Eq. \eqref{eq_DBM_eigenvalues}}}{=}  \left(\mathrm{d}B_{i, i}(t) +  \beta \sum_{j \neq i} \frac{1}{\gamma_i(t) - \gamma_j(t)} \mathrm{d}t \right) -  \left(\mathrm{d}B_{i+1, i+1}(t) +  \beta \sum_{j \neq i+1} \frac{1}{\gamma_{i+1}(t) - \gamma_j(t)} \mathrm{d}t \right)\\
&\stackrel{\textrm{Eq. \eqref{eq_z11}}}{=}   \mathrm{d}B_{i, i}(t)  - \mathrm{d}B_{i+1, i+1}(t)  + \beta  \mathrm{d}t  \sum_{j \in [d] \backslash \{i, i+1\}}  \frac{1}{\xi_i(t)- \xi_{j}(t)} - \frac{1}{\xi_{i+1}(t)- \xi_{j}(t)}\\
&= \mathrm{d} \xi_i(t) - \mathrm{d} \xi_{i+1} (t).
\end{align*}
\end{proof}

\end{proposition}

\begin{proof}[Proof of Lemma \ref{lemma_gap_comparison}]
First, we note that since by Lemma \ref{lemma_continuity} at every time $t \geq 0$ the strong solution $\gamma(t)$ is a continuous function of the initial conditions $\gamma(0)$, without loss of generality we may assume that that the initial eigenvalue gaps of $\gamma$ are strictly greater than the corresponding eigenvalue gaps of $\xi$:
\begin{equation}\label{eq_initial_gaps}
    \xi_i(0) - \xi_{i+1}(0)  < \gamma_i(0) - \gamma_{i+1}(0) \qquad \qquad 1\leq i < d.
\end{equation}
We prove Lemma \ref{lemma_gap_comparison} by contradiction.
Let $\tau := \inf\{t\geq 0: \xi_i(t) - \xi_{i+1}(t) > \gamma_i(t) - \gamma_{i+1}(t) \textrm{ for some } i \in [d]  \}$ to be the first time where the size of the $i$'th gaps ``cross'' for some $i \in [d]$ (in other words $\tau$ be the first time when the conclusion of Lemma \ref{lemma_gap_comparison} fails to hold).

\medskip\noindent
{\em Assumption towards a contradiction:} Suppose (towards a contradiction) that $\tau < \infty$.
By definition of strong solutions, strong solutions to stochastic differential are almost surely continuous on $[0,\infty)$, we have that both $\gamma$ and $\xi$ are almost surely continuous on $[0,\infty)$.
Therefore, since $\tau< \infty$, by the intermediate value theorem, we must have that, for some $i\in [d]$ $i$'th gap of $\xi$ and the $i$'th gap of $\gamma$ are equal at the time $\tau$, and that at this time $\tau$ all the other gaps of $\gamma$ are  at least as large as the corresponding gaps of $\xi$:
\begin{align}
      \gamma_i(\tau) - \gamma_{i+1}(\tau) &= \xi_i(\tau) - \xi_{i+1}(\tau),\label{eq_w1}\\ 
\nonumber\\
      \gamma_j(\tau) - \gamma_{j+1}(\tau) &\geq \xi_j(\tau) - \xi_{j+1}(\tau) \qquad \forall j \in[d-1]. \label{eq_w2} 
\end{align}
Moreover, by Lemma \ref{lemma_DBM_collision} we have that, almost surely, the particles of Dyson Brownian motion do not collide with each other on all of $(0,\infty)$.
Therefore, we have that, almost surely, 
\begin{align} \label{eq_w3}
\gamma_i(t) - \gamma_{i+1}(t) >0 \qquad \forall t \in (0,\infty), \,\, i \in [d-1]\\
\xi_i(t) - \xi_{i+1}(t) >0 \qquad \forall t \in (0,\infty), \,\, i \in [d-1].
\end{align}
Therefore, plugging in \eqref{eq_w1}, \eqref{eq_w2} and \eqref{eq_w3} into Proposition \ref{prop_stochastic_derivative_comparison}, we have that
\begin{equation}  \label{eq_w4}
\mathrm{d} \gamma_i(\tau) - \mathrm{d}  \gamma_{i+1}(\tau) \geq \mathrm{d} \xi_i(\tau) - \mathrm{d}  \xi_{i+1}(\tau).
\end{equation}
Next, we consider two cases: when $\mathrm{d} \gamma_i(\tau) - \mathrm{d}  \gamma_{i+1}(\tau) > \mathrm{d} \xi_i(\tau) - \mathrm{d}  \xi_{i+1}(\tau)$, and when $\mathrm{d} \gamma_i(\tau) - \mathrm{d}  \gamma_{i+1}(\tau) = \mathrm{d} \xi_i(\tau) - \mathrm{d}  \xi_{i+1}(\tau)$.

\paragraph{Case 1,  $\mathrm{d}  \gamma_i(\tau) - \mathrm{d}  \gamma_{i+1}(\tau) > \mathrm{d} \xi_i(\tau) - \mathrm{d}  \xi_{i+1}(\tau)$:}
For any $w \in \mathcal{W}_d$, define the ``drift'' function
\begin{equation}
    \mu_i(w) := \beta\sum_{j \neq i} \frac{1}{w_i - w_j}.
\end{equation}
Then we have that
\begin{align}\label{eq_w5}
&(\mathrm{d} \gamma_i(t) - \mathrm{d} \gamma_{i+1} (t))- (\mathrm{d} \xi_i(t) - \mathrm{d} \xi_{i+1} (t))\\ &\stackrel{\textrm{Eq. \eqref{eq_DBM_eigenvalues}}}{=}  \left[\left(\mathrm{d}B_{i, i}(t) +  \beta \sum_{j \neq i} \frac{1}{\gamma_i(t) - \gamma_j(t)} \mathrm{d}t \right)
-  \left(\mathrm{d}B_{i+1, i+1}(t) +  \beta \sum_{j \neq i+1} \frac{1}{\gamma_{i+1}(t) - \gamma_j(t)} \mathrm{d}t \right) \right] \nonumber\\
&\qquad  -\left[\left(\mathrm{d}B_{i, i}(t) +  \beta \sum_{j \neq i} \frac{1}{\xi_i(t) - \xi_j(t)} \mathrm{d}t \right) -  \left(\mathrm{d}B_{i+1, i+1}(t) +  \beta \sum_{j \neq i+1} \frac{1}{\xi_{i+1}(t) - \xi_j(t)} \mathrm{d}t \right) \right] \nonumber\\
&= \mu_i(\gamma(t)) - \mu_{i+1}(\gamma(t)) - (\mu_i(\xi(t)) - \mu_{i+1}(\xi(t))).\nonumber
\end{align}
From \eqref{eq_w3}, all the gaps of $\gamma$ and $\xi$ are strictly greater than zero at time $\tau$.
Therefore, since $\gamma$ and $\xi$ are almost surely continuous on $[0,\infty)$,
we must have that, $\mu(\gamma(t))$ and $\mu(\xi(t))$ are also continuous on all of $t\in(0,\infty)$. 

Since
$\mathrm{d}  \gamma_i(\tau) - \mathrm{d}  \gamma_{i+1}(\tau) > \mathrm{d} \xi_i(\tau) - \mathrm{d}  \xi_{i+1}(\tau)$,
we have by \eqref{eq_w5} that
\begin{align}\label{eq_w6}
\mu_i(\gamma(\tau)) - \mu_{i+1}(\gamma(\tau)) - (\mu_i(\xi(\tau)) - \mu_{i+1}(\xi(\tau))) &= (\mathrm{d} \gamma_i(\tau) - \mathrm{d} \gamma_{i+1} (\tau))- (\mathrm{d} \xi_i(\tau) - \mathrm{d} \xi_{i+1} (\tau)) \nonumber \\
&>0.
\end{align}
Therefore, since $\mu(\gamma(t))$ and $\mu(\xi(t))$ is almost surely continuous on $(0,\infty)$, by \eqref{eq_w6}  we must have that there exits some open interval $\mathcal{I} \subset (0,\infty)$ containing $\tau$ such that
\begin{equation}\label{eq_w7}
 (\mathrm{d} \gamma_i(t) - \mathrm{d} \gamma_{i+1} (t)- (\mathrm{d} \xi_i(t) - \mathrm{d} \xi_{i+1} (t)) >0. \qquad \qquad \forall t \in \mathcal{I}.
\end{equation}
Consider any $t \in \mathcal{I}$ such that $t > \tau$.
Then
\begin{align}\label{eq_w8}
&(\gamma_i(t) -\gamma_{i+1} (t)) - (\xi_i(t) -  \xi_{i+1} (t)) \nonumber\\
    &\stackrel{\textrm{Eq. \eqref{eq_w2}}}{=} [(\gamma_i(t) -\gamma_{i+1} (t)) - (\xi_i(t) -  \xi_{i+1} (t))] -   [(\gamma_i(\tau) -  \gamma_{i+1} (\tau)) - (\xi_i(\tau) -  \xi_{i+1} (\tau))] \nonumber\\
    &=  \int_{\tau}^{t}  (\mathrm{d} \gamma_i(s) - \mathrm{d} \gamma_{i+1} (s)- (\mathrm{d} \xi_i(s) - \mathrm{d} \xi_{i+1} (s)) \mathrm{d}s \nonumber\\
    &\stackrel{\textrm{Eq. \eqref{eq_w7}}}{>} 0.
\end{align}
Therefore \eqref{eq_w8} implies that there exits some $\tau' \in \mathcal{I}$ where $\tau'>\tau$ such that 
\begin{equation*}
    \gamma_i(t) -\gamma_{i+1} (t) > \xi_i(t) -  \xi_{i+1} (t) \qquad \forall \tau<t <\tau'.
\end{equation*}
Therefore $\tau \neq \inf\{t\geq 0: \xi_i(t) - \xi_{i+1}(t) > \gamma_i(t) - \gamma_{i+1}(t)\}$ for any $i \in [d]$.
This contradicts the definition of $\tau$.
Therefore, by contradiction our assumption that $\tau < \infty$ is false.

\paragraph{Case 2, $\mathrm{d} \gamma_i(\tau) - \mathrm{d}  \gamma_{i+1}(\tau) = \mathrm{d} \xi_i(\tau) - \mathrm{d}  \xi_{i+1}(\tau)$:}

Consider the system of stochastic differential equation for the process $\gamma_i(t) - \gamma_{i+1} (t)$:
\begin{align}\label{eq_w9}
\mathrm{d} \gamma_i(t) - \mathrm{d} \gamma_{i+1} (t) 
&\stackrel{\textrm{Eq. \eqref{eq_DBM_eigenvalues}}}{=}  \left(\mathrm{d}B_{i, i}(t) +  \beta \sum_{j \neq i} \frac{1}{\gamma_i(t) - \gamma_j(t)} \mathrm{d}t \right) \nonumber\\
&\qquad \qquad -  \left(\mathrm{d}B_{i+1, i+1}(t) +  \beta \sum_{j \neq i+1} \frac{1}{\gamma_{i+1}(t) - \gamma_j(t)} \mathrm{d}t \right) \qquad \forall i \in [d]
\end{align}
and the system of stochastic differential equation for the process $\xi_i(t) - \xi_{i+1} (t)$:
\begin{align}\label{eq_w10}
 \mathrm{d} \xi_i(t) - \mathrm{d} \xi_{i+1} (t)  &\stackrel{\textrm{Eq. \eqref{eq_DBM_eigenvalues}}}{=}
\left(\mathrm{d}B_{i, i}(t) +  \beta \sum_{j \neq i} \frac{1}{\xi_i(t) - \xi_j(t)} \mathrm{d}t \right) \nonumber\\
&\qquad \qquad -  \left(\mathrm{d}B_{i+1, i+1}(t) +  \beta \sum_{j \neq i+1} \frac{1}{\xi_{i+1}(t) - \xi_j(t)} \mathrm{d}t \right) \qquad \forall i \in [d].
\end{align}
Then we have that
\begin{align}\label{eq_w11}
&0 = (\mathrm{d} \gamma_i(\tau) - \mathrm{d} \gamma_{i+1} (\tau))- (\mathrm{d} \xi_i(\tau) - \mathrm{d} \xi_{i+1} (\tau))\nonumber\\
&\stackrel{\textrm{Eq. \eqref{eq_DBM_eigenvalues}}}{=}  \left[\left(\beta \sum_{j \neq i} \frac{1}{\gamma_i(\tau) - \gamma_j(\tau)} \mathrm{d}t \right)
-  \left(\beta \sum_{j \neq i+1} \frac{1}{\gamma_{i+1}(\tau) - \gamma_j(\tau)} \mathrm{d}t \right) \right] \nonumber\\
&\qquad \qquad -\left[\left(\beta \sum_{j \neq i} \frac{1}{\xi_i(\tau) - \xi_j(\tau)} \mathrm{d}t \right) -  \left(\beta \sum_{j \neq i+1} \frac{1}{\xi_{i+1}(\tau) - \xi_j(\tau)} \mathrm{d}t \right) \right].
\end{align}
But we also have from \eqref{eq_w1} that $\gamma_i(\tau) - \gamma_{i+1}(\tau) = \xi_i(\tau) - \xi_{i+1}(\tau)$ and from \eqref{eq_w2} that $ \gamma_j(\tau) - \gamma_{j+1}(\tau) \geq \xi_j(\tau) - \xi_{j+1}(\tau)$ for all $j \in [d-1]$.
Thus, the only way for the r.h.s. of \eqref{eq_w11} to be equal to zero is if we have $\xi_i(\tau) -\xi_{i+1}(\tau) = \gamma_i(\tau) -\gamma_{i+1}(\tau)$ for all $i \in [d-1]$. 
Therefore, since the  $\xi_i(\tau) -\xi_{i+1}(\tau) = \gamma_i(\tau) -\gamma_{i+1}(\tau)$ for all $i \in [d-1]$.

Moreover, by Lemma \ref{lemma_strong}, for any initial conditions $\gamma(\tau)$ and $\xi(\tau)$, the processes $\gamma$ and $\xi$ have unique strong solutions  on $(0,\infty)$.
Therefore, since the stochastic differential equation \eqref{eq_DBM_eigenvalues} for $\gamma$ and $\xi$ are invariant to spatial translation, we must have that
\begin{equation}\label{eq_w12}
    \xi_i(t) -\xi_{i+1}(t) = \gamma_i(t) -\gamma_{i+1}(t) \qquad \forall t \geq \tau, i \in [d].
\end{equation}
By \eqref{eq_w12}, we have that  $\tau = \inf\{t\geq 0: \xi_i(t) - \xi_{i+1}(t) > \gamma_i(t) - \gamma_{i+1}(t) \textrm{ for some } i \in [d]  \} = \infty$.
This contradicts our assumption that $\tau < \infty$.
Therefore, by contradiction our assumption that $\tau < \infty$ is false.

\end{proof}

\subsection{Showing gaps are uniformly bounded below over time with high probability}
\begin{lemma}\label{lemma_bad_event}
Let $\gamma(t) = (\gamma_1(t), \ldots, \gamma_d(t))$ be a strong solution to \eqref{eq_DBM_eigenvalues}, with initial condition $\gamma(0) = (0,\cdots, 0)$.
Then for any $t_0 \geq \frac{1}{d^{40}}$ and any $T>0$ we have
\begin{equation} \label{eq_e1}
    \mathbb{P}\left(\inf_{t_0 \leq t \leq T, 1\leq i < d  }\gamma_i(t) - \gamma_{i+1}(t) \leq \frac{1}{d^{10}} \frac{\sqrt{t}}{\mathfrak{b}\sqrt{d}}\right) \leq \frac{T}{d^{600}},
\end{equation}
for any $d \geq N_0$ where $N_0$ is a universal constant.

\end{lemma}

\begin{proof}
By Weyl's inequality (Lemma \ref{lemma_weyl}), we have that for any $z \geq t_0$,
\begin{align}\label{eq_e2}
& \mathbb{P}\left( \gamma_i(t) - \gamma_{i+1}(t) \leq \frac{1}{d^{10}} \frac{\sqrt{t}}{\mathfrak{b}\sqrt{d}}   \qquad \textrm{ for some }  t \in \left[z, z + \frac{1}{d^{200}}\right], i \in[d-1] \right) \nonumber\\
 & \stackrel{\textrm{Lemma \ref{lemma_weyl}}}{\leq} \mathbb{P}\left( \gamma_i(z) - \gamma_{i+1}(z) \leq \frac{1}{d^{10}} \frac{\sqrt{t}}{\mathfrak{b}\sqrt{d}} +  2 \|B(t)\|_2   \, \, \textrm{ for some }  t \in \left[z, z + \frac{1}{d^{200}}\right], i \in[d-1] \right) \nonumber\\
  &\leq \mathbb{P}\left( \gamma_i(z) - \gamma_{i+1}(z) \leq \frac{1}{d^{10}} \frac{\sqrt{t}}{\mathfrak{b}\sqrt{d}} +  4 \frac{1}{d^{200}} \sqrt{d}   \qquad \textrm{ for some }  t \in \left[z, z + \frac{1}{d^{200}}\right], i \in[d-1]  \right)\nonumber\\
   & \qquad \qquad+ \mathbb{P}\left(\sup_{t \in [0, \frac{1}{d^{200}}]} \|B(t)\|_2  > 2 \frac{1}{d^{200}} \sqrt{d} \right)\nonumber\\
   &  \stackrel{\textrm{Lemma \ref{lemma_spectral_martingale_b}}}{\leq}  \mathbb{P}\left( \gamma_i(z) - \gamma_{i+1}(z) \leq \frac{1}{d^{10}} \frac{\sqrt{t}}{\mathfrak{b}\sqrt{d}} +  4 \frac{1}{d^{200}} \sqrt{d}     \quad \textrm{ for some }   t \in \left[z, z + \frac{1}{d^{200}}\right], i \in[d-1]   \right)\nonumber\\
   &\qquad \qquad \qquad \qquad+  \frac{1}{d^{1000}}\nonumber\\
   & \leq  \mathbb{P}\left( \gamma_i(z) - \gamma_{i+1}(z) \leq \frac{2}{d^{10}} \frac{\sqrt{z}}{\mathfrak{b}\sqrt{d}}       \qquad \textrm{ for some }     i \in[d-1]    \right)+  \frac{1}{d^{1000}}\nonumber\\
      & \leq  \sum_{i=1}^{d-1} \mathbb{P}\left( \gamma_i(z) - \gamma_{i+1}(z) \leq \frac{2}{d^{10}} \frac{\sqrt{z}}{\mathfrak{b}\sqrt{d}}   \right)+  \frac{1}{d^{1000}}\nonumber\\
   &  \stackrel{\textrm{Lemma \ref{lemma_GUE_gaps}}}{\leq}  \sum_{i=1}^{d-1}  \left(\frac{2}{d^{10}}\right)^{3} + \frac{1}{d^{1000}}\nonumber\\
   &\leq \frac{1}{d^{997}},
 \end{align}
 where the third inequality holds by Lemma \ref{lemma_spectral_martingale_b} whenever $d \geq N_0$ for some sufficiently large universal constant $N_0$, and the sixth inequality holds by Lemma \ref{lemma_GUE_gaps}.
Thus, we have,
\begin{align*}
&  \mathbb{P}\left(\inf_{t_0 \leq t \leq T, 1\leq i < d  }\gamma_i(t) - \gamma_{i+1}(t) \leq \frac{1}{d^{10}} \frac{\sqrt{t}}{\mathfrak{b}\sqrt{d}}\right)\\
& =\mathbb{P}\left(\gamma_i(t) - \gamma_{i+1}(t) \leq \frac{1}{d^{10}} \frac{\sqrt{t}}{\mathfrak{b}\sqrt{d}}      \qquad \textrm{ for some }   t\in [t_0, T], i \in [d-1] \right)\\
&\leq \mathbb{P}\left( \bigcup_{z\in [t_0, T] \cap \frac{1}{d^{200}} \mathbb{Z}} \inf_{1\leq i < d  }\gamma_i(t) - \gamma_{i+1}(t) \leq \frac{1}{d^{10}} \frac{\sqrt{t}}{\mathfrak{b}\sqrt{d}}   \qquad \forall  t \in [z, z + \frac{1}{d^{200}}] \right)\\
&\leq \sum_{z\in [t_0, T] \cap \frac{1}{d^{200}}\mathbb{Z}}  \mathbb{P}\left( \inf_{1\leq i < d  }\gamma_i(t) - \gamma_{i+1}(t) \leq \frac{1}{d^{10}} \frac{\sqrt{t}}{\mathfrak{b}\sqrt{d}}   \qquad \forall  t \in [z, z + \frac{1}{d^{200}}] \right)\\
&\leq \sum_{z\in [t_0, T] \cap \frac{1}{d^{200}}\mathbb{Z}}  \frac{1}{d^{998}}\\
  & \leq d^{200} T \times  \frac{1}{d^{998}}\\
   &\leq \frac{T}{d^{600}}.
\end{align*}
\end{proof}

\subsection{Gaps between not necessarily neighboring eigenvalues} \label{section_non_neighboring_gaps}

\begin{proposition}\label{prop_sum_nonindependent}
Suppose that $X_1, \ldots, X_r$  are (not necessarily independent) random variables satisfying $\mathbb{P}(X_i \leq s) \leq F(s)$ for all $i \in [r]$, where $F: \mathbb{R} \rightarrow \mathbb{R}$ is some nondecreasing function.
Then
\begin{equation} \label{eq_sum_nonindependent}
    \mathbb{P}\left(\sum_{i=1}^r X_i \leq \frac{1}{2}rs\right) \leq 2F(s). 
\end{equation}
\end{proposition}

\begin{proof}
Let $E$ be the ``bad'' event that $|\{i: X_i \leq s\}| \geq \frac{r}{2}$.
Choose $J$ uniformly at random from $\{1,\ldots, r\}$.
Then $\mathbb{P}(J \in \{i: X_i \leq s\} | E) \geq \frac{1}{2}$.
Therefore,
\begin{align*}
    \mathbb{P}(X_j \leq s) &=  \mathbb{P}(X_j \leq s| E) \times \mathbb{P}(E)\\
    & = \mathbb{P}(J \in \{i: X_i \leq s\} | E)  \times \mathbb{P}(E)\\
    &\geq \frac{1}{2} \mathbb{P}(E).
\end{align*}
Thus, $\mathbb{P}(E) \leq 2  \mathbb{P}(X_j \leq s)$
But $\{\sum_{i=1}^r X_i \leq \frac{1}{2}rs\} \subseteq E$.
Therefore,
\begin{equation*}
      \mathbb{P}\left(\sum_{i=1}^r X_i \leq \frac{1}{2}rs\right) \leq \mathbb{P}(E) \leq  2  \mathbb{P}(X_J \leq s) \leq 2F(s),
\end{equation*}
where the last inequality holds since $\mathbb{P}(X_i \leq s) \leq F(s)$  for all $i \in [r]$.
\end{proof}

\begin{corollary}[Gaps between not-necessarily neighboring  eigenvalues]\label{lemma_gaps_any_start}

\noindent Let $\gamma(t) = (\gamma_1(d), \ldots, \gamma_d(t))$ be a strong solution  of \eqref{eq_DBM_eigenvalues} starting from any initial $\gamma(0) \in \mathcal{W}_d$.
Then for every $i,j \in [d]$ and every $\alpha>0$,
\begin{equation*}
    \mathbb{P}\left(\left\{ \gamma_i(t) - \gamma_{j}(t) \leq (j-i) \times s \frac{\sqrt{t}}{\mathfrak{b}\sqrt{d}} \right\} \cap \hat{E}_\alpha^c\right) \leq  s^3 \qquad \forall s>0, t>0.
\end{equation*}
\end{corollary}

\begin{proof}
Let $\xi(t) = (\xi_1(d), \ldots, \xi_d(t))$ be a strong solution  of \eqref{eq_DBM_eigenvalues} starting from $\xi(0) = (0, \cdots, 0)$.
Further, let $\tilde{E}$ be the event that $\inf_{t_0 \leq t \leq T, 1\leq i < d  }\xi_i(t) - \xi_{i+1}(t) \leq \frac{1}{d^{10}} \frac{\sqrt{t}}{\mathfrak{b}\sqrt{d}}$.

Since $\xi(t) =(\xi_1(t),\ldots, \xi_d(t))$ have the same joint distribution as the eigenvalues of $\sqrt{t}(G + G^\ast)$ where $G$ is a $d \times d$ matrix with i.i.d. complex standard Gaussian entries, by Lemma \ref{lemma_GUE_gaps} we have that, 
\begin{equation}\label{eq_e3}
    \mathbb{P}\left(\left \{\xi_i(t) - \xi_{i+1}(t) \leq s \frac{\sqrt{t}}{\mathfrak{b}\sqrt{d}}\right\} \cap \tilde{E}^c\right) \leq 2s^{3} \qquad \forall s>0, \forall 1\leq i < d.
\end{equation}
Define $X_\ell := \xi_{i+\ell}(t) - \xi_{i+\ell+1}(t) $ for all $\ell \in \{1, \ldots, j-i\}$.
Then using the fact that by Lemma \ref{lemma_gap_comparison} all the eigenvalue gaps of $\gamma(t)$ are at least as large as the corresponding gaps of $\xi(t)$,  and plugging \eqref{eq_e3} into Proposition \ref{prop_sum_nonindependent}, we have that
\begin{align*}
   &\mathbb{P}\left(\left\{ \gamma_i(t) - \gamma_{j}(t) \leq (j-i) \times s \frac{\sqrt{t}}{2\mathfrak{b}\sqrt{d}} \right \} \cap\hat{E}_\alpha^c\right)\nonumber\\
   &\qquad \stackrel{\textrm{Lemmas  \ref{lemma_gap_comparison}}}{\leq} \mathbb{P}\left(\left\{ \xi_i(t) - \xi_{j}(t)\leq (j-i) \times s \frac{\sqrt{t}}{2\mathfrak{b}\sqrt{d}} \right \} \cap \tilde{E}^c \right)\\
   &\qquad =\mathbb{P}\left(\left\{ \sum_{\ell=1}^{j-i} X_j \leq \frac{1}{2}(j-i) \times s \frac{\sqrt{t}}{\mathfrak{b}\sqrt{d}} \right \} \cap \tilde{E}^c \right)\\
    &\qquad \stackrel{\textrm{Lemma  \ref{prop_sum_nonindependent}, \textrm{ Eq. } \eqref{eq_e3}}}{\leq}  4 s^{3}.
\end{align*}
Redefining $\mathfrak{b}$ to be $4$ times the original value of $\mathfrak{b}$ completes the proof.
\end{proof}

\section{Eigenvalue gaps of Gaussian Unitary Ensemble: Proof of Lemma \ref{lemma_GUE_gaps}}\label{section_GUE_proof}

\subsection{Eigenvalue ridgidity}
Denote by $\eta_1,\ldots, \eta_d$ the eigenvalues of the GUE random matrix-- that is the matrix $G+G^\ast$ where each entry of $G$ is an independent standard complex Gaussian.
The eigenvalue gaps of the GUE satisfy a rigidity property (\cite{erdHos2012rigidity}; restated here as Lemma \ref{lemma_rigidity}).
Roughly, for every $i \in [d]$ the $i$'th eigenvalue $\eta_i$ does not deviate by more than $\mathrm{polylog}(d)$ times the average gap size $\eta_i - \eta_{i+1}.
$
More formally, for every $i\in [d]$ we define the ``classical'' eigenvalue location $\omega_i$ to be the number such that
\begin{equation}\label{eq_a7}
      d\int_{\frac{\omega_i}{\sqrt{d}}}^{\infty}\rho(x) \mathrm{d}x = i - 1,
\end{equation}
where $\rho(x):= \frac{1}{2\pi}\sqrt{\max(4 - x^2, \,\, 0)}$ is the semi-circle law.
For convenience, we also define $\omega_{d+1} := -2 \sqrt{d}$ (that way, the locations of the  $\omega_{d+1} \leq \omega_d \leq \cdots \leq \omega_1$  are symmetric about $0$).

\begin{proposition}\label{prop_classical}
The classical eigenvalues $\omega_i$ satisfy
\begin{equation}\label{eq_a10}
2\sqrt{d}- 3d^{-\frac{1}{6}} (i-1)^{\frac{2}{3}}   \leq \omega_i \leq 2\sqrt{d}- d^{-\frac{1}{6}} (i-1)^{\frac{2}{3}} \qquad \forall \ 1\leq i \leq \frac{d}{2},
\end{equation}
\begin{equation}\label{eq_a11}
    d^{-\frac{1}{6}} (d-i+1)^{\frac{2}{3}}  -2\sqrt{d}\leq \omega_{i} \leq 3d^{-\frac{1}{6}} (d-i+1)^{\frac{2}{3}} -2\sqrt{d} \qquad \forall \ \frac{d}{2} \leq i \leq d.
\end{equation}
Moreover, their gaps satisfy
\begin{equation}\label{eq_a12}
d^{-\frac{1}{6}}\min(i, d-i+1)^{-\frac{1}{3}}\leq \omega_i - \omega_{i+1} \leq 2\pi d^{-\frac{1}{6}} \min(i, d-i+1)^{-\frac{1}{3}} \qquad \forall \ 1\leq i \leq d,
\end{equation}

\end{proposition}

\begin{proof}
For every $x\in[-2,0]$, we have
\begin{equation}\label{eq_a13}
    \frac{1}{2\pi}\sqrt{x+2} \leq \rho(x) \leq \frac{1}{2\pi} 2\sqrt{x+2}  
\end{equation}
Furthermore, since $\rho(x)$ is symmetric about $0$, for every $x\in[0,2]$, we have
\begin{equation}\label{eq_a13b}
    \frac{1}{2\pi}\sqrt{2-x} \leq \rho(x) \leq \frac{1}{2\pi} 2\sqrt{2-x}.
\end{equation}
Thus, \eqref{eq_a13} implies that,
\begin{equation*}
\int_{-2}^x \rho(s) \mathrm{d} s
\geq \int_{-2}^x    \frac{1}{2\pi} \sqrt{x+2} \mathrm{d} s = \frac{3}{2}(x+2)^{\frac{3}{2}}
\end{equation*}
and that
\begin{equation*}
\int_{-2}^x \rho(s) \mathrm{d} s
\leq 2\int_{-2}^x    \frac{1}{2\pi} \sqrt{x+2} \mathrm{d} s = 2(x+2)^{\frac{3}{2}}.
\end{equation*}
    Then for every $i \in [d]$ we have
      \begin{equation}\label{eq_a5}
\int_{-2}^{d^{-\frac{2}{3}} i^{\frac{2}{3}} -2} \rho(x) \mathrm{d} x \leq \frac{1}{2\pi} 2( d^{-\frac{2}{3}} i^{\frac{2}{3}}  -2)^{1.5} \leq \frac{i}{d},
\end{equation}
and
      \begin{equation} \label{eq_a6}
\int_{-2}^{3d^{-\frac{2}{3}} i^{\frac{2}{3}} -2} \rho(x) \mathrm{d} x \geq \frac{1}{2\pi} \frac{3}{2}(3d^{-\frac{2}{3}} i^{\frac{2}{3}})^{1.5} \geq \frac{i}{d}.
\end{equation}
Since $\rho(x)$ is nonnegative, $\int_{-2}^x \rho(s) \mathrm{d} s$ is nondecreasing in $x$.
Therefore, from \eqref{eq_a5} and \eqref{eq_a6}, we have by the definition of $\omega_i$ (Equation \eqref{eq_a7}) that

\begin{equation}\label{eq_a8}
d^{-\frac{2}{3}} (d-i+1)^{\frac{2}{3}}  -2 \leq \frac{\omega_{i}}{\sqrt{d}} \leq 3d^{-\frac{2}{3}} (d-i+1)^{\frac{2}{3}} -2 \qquad \forall 1\leq i \leq d+1,
\end{equation}
which proves \eqref{eq_a12}.
Moreover, since the density $\rho(x)$ is symmetric about $0$, \eqref{eq_a8} implies that
\begin{equation} \label{eq_a8b}
2- 3d^{-\frac{2}{3}} (i-1)^{\frac{2}{3}}   \leq \frac{\omega_i}{\sqrt{d}} \leq 2- d^{-\frac{2}{3}} (i-1)^{\frac{2}{3}}\qquad \qquad \forall 1\leq i \leq d+1,
\end{equation}
which proves \eqref{eq_a11}.
Since $\rho(x)$ is nonincreasing  on [0,2] we  also have that for all $2\leq i \leq \frac{d}{2} +1$,
\begin{equation} \label{eq_a14}
    \frac{\omega_i}{\sqrt{d}} - \frac{\omega_{i+1}}{\sqrt{d}} \leq \frac{1}{d\times \rho(\omega_{i})} \stackrel{\textrm{Eq. \eqref{eq_a13b}}}{\leq} \frac{1}{d     \frac{1}{2\pi} \sqrt{2-\omega_{i}}} \stackrel{\textrm{Eq. \eqref{eq_a8b}}}{\leq}  \frac{2 \pi}{d \sqrt{   d^{-\frac{2}{3}} (i-1)^{\frac{2}{3}}}} \leq 4\pi d^{-\frac{2}{3}}i^{-\frac{1}{3}}.
\end{equation}
and that, for all $1\leq i \leq \frac{d}{2} +1$,
\begin{equation} \label{eq_a15}
    \frac{\omega_i}{\sqrt{d}} - \frac{\omega_{i+1}}{\sqrt{d}} \geq \frac{1}{d\times \rho(\omega_{i+1})} \stackrel{\textrm{Eq. \eqref{eq_a13b}}}{\geq} \frac{1}{2d     \frac{1}{2\pi} \sqrt{2-\omega_{i+1}}} \stackrel{\textrm{Eq. \eqref{eq_a8b}}}{\geq}  \frac{1}{2d     \frac{1}{2\pi} \sqrt{    3d^{-\frac{2}{3}} i^{\frac{2}{3}}}} \geq \frac{\pi}{\sqrt{3}} d^{-\frac{2}{3}}i^{-\frac{1}{3}}.
\end{equation}
Therefore,
\begin{equation}\label{eq_a9}
\frac{\pi}{\sqrt{3}} d^{-\frac{1}{6}}i^{-\frac{1}{3}} \stackrel{\textrm{Eq. \eqref{eq_a15}}}{\leq} \omega_i - \omega_{i+1} \stackrel{\textrm{Eq. \eqref{eq_a14}}}{\leq} 4 \pi d^{-\frac{1}{6}}i^{-\frac{1}{3}} \qquad \forall 2\leq i \leq \frac{d}{2} +1.
\end{equation}
Moreover, plugging in $i=2$ to \eqref{eq_a10} plugging in the fact that $\omega_1 = 2\sqrt{d}$, we have that
\begin{equation}\label{eq_a10b}
d^{-\frac{1}{6}}   \leq \omega_1 -  \omega_2 \leq 3d^{-\frac{1}{6}}.
\end{equation}
Therefore \eqref{eq_a9} and \eqref{eq_a10b} together imply that,
\begin{equation}\label{eq_a9b}
\frac{\pi}{\sqrt{3}} d^{-\frac{1}{6}}i^{-\frac{1}{3}} \stackrel{\textrm{Eq. \eqref{eq_a15}}}{\leq} \omega_i - \omega_{i+1} \stackrel{\textrm{Eq. \eqref{eq_a14}}}{\leq} 4 \pi d^{-\frac{1}{6}}i^{-\frac{1}{3}} \qquad \forall 1\leq i \leq \frac{d}{2} +1.
\end{equation}
Finally,  since the density $\rho(x)$ is symmetric about $0$, \eqref{eq_a9b} implies that
\begin{equation*}
 d^{-\frac{1}{6}}\min(i, d-i+1)^{-\frac{1}{3}}\leq \omega_i - \omega_{i+1} \leq 2\pi d^{-\frac{1}{6}} \min(i, d-i+1)^{-\frac{1}{3}} \qquad \forall 1\leq i \leq d,
\end{equation*}
which proves \eqref{eq_a12}.
\end{proof}

\begin{lemma}[Eigenvalue rigidity of GUE (Theorem 2.2 of \cite{erdHos2012rigidity})]\label{lemma_rigidity}
There exist universal constants $C\geq 1$ and $c_1, c_2, N_0 >0$ such that for every $L \in \left[c_1, \frac{\log(10 d)}{10 (\log\log d)^2}\right]$ and every $d \geq N_0$,
\begin{equation*}
    \mathbb{P}\left(\exists j \in [d] : |\eta_j - \omega_j| \geq (\log d)^{ L\log \log d} \min(j, d-j+1)^{-\frac{1}{3}} d^{-\frac{1}{6}} )\right) \leq C \exp[ -(\log d)^{c_2 L \log \log d} ].
\end{equation*}
\end{lemma}

\subsection{Bounding the eigenvalue gaps of the GUE matrix}

In this section, we prove high-probability bounds for the eigenvalue gaps of the GUE random matrix (Lemma \ref{lemma_GUE_gaps}).

\paragraph{Step 1.}  Define the rigidity event $E$ and prove that it holds with high probability (Use Lemma \ref{lemma_rigidity}).
Set $L:= \max( \frac{2}{c_2} \log \log(C), c_1, 1)$;  thus, $L$ is a universal constant.
Define the event $E$ as follows:
\begin{equation*}
E:= \left \{\exists j \in [d] : |\eta_j - \omega_j| \geq (\log d)^{ L\log \log d} \min(j, d-j+1)^{-\frac{1}{3}} d^{-\frac{1}{6}} ) \right \}^c.
\end{equation*}
Then (replacing the universal constant $N_0$ with a universal constant such that $\max( \frac{2}{c_2} \log \log(C), c_1) \leq \frac{\log(10 d)}{10 (\log\log d)^2}$ and $N_0 \geq e^{4}$), we have by Lemma \ref{lemma_rigidity} that
\begin{equation} \label{eq_rigidity_1}
    \mathbb{P}(E^c) \leq  C \exp[ -(\log d)^{c_2 L \log \log d}]\leq  \exp[ -(\log d)^{2 \log \log d}] \leq  \exp[ -(\log d)^{2}] \leq  d^{ -\log d} \leq \frac{1}{d^{1000}},
    \end{equation}
    for all $d \geq N_0$, where $N_0$ is a  universal constant.
Define $\mathfrak{b} := 10^6(\log d)^{ L\log \log d}$.
Further, define $\omega_j=\eta_{j} = + \infty$ for all $j < d$ and  $\omega_j=\eta_{j} = - \infty$ for all $j>d$. 

\paragraph{Step 2.} Consider any $\eta \in \mathcal{W}_d$ such that the event $E$ holds.
Define $j_{\mathrm{min}}:= \max(i-\mathfrak{b}^2, 1)$ and $j_{\mathrm{max}} := \min(i+\mathfrak{b}^2, d)$.
Define the following quantities: 
\begin{itemize}

\item $a_{\mathrm{min}}:= \omega_{j_{\mathrm{min}}} - \frac{1}{30} \mathfrak{b}^2 d^{-\frac{1}{6}} \min(i, d-i)^{-\frac{1}{3}}$
\item $a_{\mathrm{max}}:=\omega_{j_{\mathrm{min}}} + \frac{1}{30} \mathfrak{b}^2 d^{-\frac{1}{6}} \min(i, d-i)^{-\frac{1}{3}}$

\item $b_{\mathrm{min}}:= \omega_{j_{\mathrm{max}}} - \frac{1}{30} \mathfrak{b}^2 d^{-\frac{1}{6}} \min(i, d-i)^{-\frac{1}{3}}$

\item $b_{\mathrm{max}}:= \omega_{j_{\mathrm{max}}} + \frac{1}{30} \mathfrak{b}^2 d^{-\frac{1}{6}} \min(i, d-i)^{-\frac{1}{3}}$.
\end{itemize}

\begin{proposition}
Suppose that the event $E$ occurs.
Then for all $\mathfrak{b}^2 \leq i \leq d-\mathfrak{b}^2$ we have
\begin{equation}\label{eq_c2}
    \eta_{i- \mathfrak{b}^2} - \eta_{i+ \mathfrak{b}^2} \geq    \frac{29}{30} \mathfrak{b}^2 d^{-\frac{1}{6}} \min(i, d-i)^{-\frac{1}{3}} \geq  \frac{29}{30} \mathfrak{b}^2 \frac{1}{\sqrt{d}}.
\end{equation}
and
$ \eta_{j_{\mathrm{max}}}  \in[a_{\mathrm{min}}, a_{\mathrm{max}}]$ and  $ \eta_{j_{\mathrm{min}}}  \in[b_{\mathrm{min}}, b_{\mathrm{max}}]$.

\end{proposition}
\begin{proof}
Without loss of generality, we may assume that $i > \mathfrak{b}^2$ (since otherwise we have $    \eta_{i- \mathfrak{b}^2} - \eta_{i+ \mathfrak{b}^2} \geq \infty$), and that $i \leq \frac{1}{2}$ (since the GUE matrix $G$ and $-G$ have the same distribution and hence the eigenvalue distribution of the GUE is symmetric about $0$).

If $E$ occurs, then by the definition of the event $E$ we have
\begin{align*}
    \eta_{i- \mathfrak{b}^2} - \eta_{i+ \mathfrak{b}^2} &\geq     \omega_{i- \mathfrak{b}^2} - \omega_{i+ \mathfrak{b}^2} - 2\mathfrak{b} (i-\mathfrak{b}^2)^{-\frac{1}{3}} d^{-\frac{1}{6}}\\
    & \stackrel{\textrm{Prop. } \ref{prop_classical}}{\geq} \mathfrak{b}^2 \times  d^{-\frac{1}{6}} i^{-\frac{1}{3}} - 2\mathfrak{b} (i-\mathfrak{b}^2)^{-\frac{1}{3}} d^{-\frac{1}{6}}\\
        & \geq \mathfrak{b}^2 \times  d^{-\frac{1}{6}} i^{-\frac{1}{3}} - 2\mathfrak{b} \left(\frac{i}{2\mathfrak{b}^2}\right)^{-\frac{1}{3}} d^{-\frac{1}{6}}\\
         & \geq \mathfrak{b}^2 \times  d^{-\frac{1}{6}} i^{-\frac{1}{3}} -  2^{\frac{4}{3}}\mathfrak{b}^{\frac{5}{3}} i^{-\frac{1}{3}} d^{-\frac{1}{6}}\\
         &\geq \frac{29}{30} \mathfrak{b}^2 d^{-\frac{1}{6}} i^{-\frac{1}{3}},
\end{align*}
where the third inequality holds since $\frac{i}{2\mathfrak{b}^2} \leq i-\mathfrak{b}^2$ because $i \geq \mathfrak{b}^2 +1 > 4$, and the last inequality holds since $\mathfrak{b} \geq 10^6$.
This proves \eqref{eq_c2}.

Moreover, by the definition of the event $E$, we also have that 
\begin{align}\label{eq_c1}
    |\eta_{i- \mathfrak{b}^2} -  \omega_{i- \mathfrak{b}^2}| &\leq   \mathfrak{b} (i-\mathfrak{b}^2)^{-\frac{1}{3}} d^{-\frac{1}{6}}\nonumber\\
    &\leq \mathfrak{b} \left(\frac{i}{2\mathfrak{b}^2}\right)^{-\frac{1}{3}} d^{-\frac{1}{6}}\nonumber\\
    &\leq 2^{\frac{1}{3}}\mathfrak{b}^{\frac{5}{3}} i^{-\frac{1}{3}} d^{-\frac{1}{6}}\nonumber\\
    &\leq \frac{1}{30} \mathfrak{b}^2 d^{-\frac{1}{6}} i^{-\frac{1}{3}}
\end{align}
where the second inequality holds since $\frac{i}{2\mathfrak{b}^2} \leq i-\mathfrak{b}^2$ because $i \geq \mathfrak{b}^2 +1 > 4$, and the last inequality holds since $\mathfrak{b} \geq 10^6$.
Thus, \eqref{eq_c1} implies that $\eta_{j_{\mathrm{max}}}  \in[a_{\mathrm{min}}, a_{\mathrm{max}}]$. 

Again, by the definition of the event $E$, we also have that 
\begin{align}\label{eq_c1b}
    |\eta_{i+ \mathfrak{b}^2} -  \omega_{i+ \mathfrak{b}^2}| &\leq   \mathfrak{b} (i+\mathfrak{b}^2)^{-\frac{1}{3}} d^{-\frac{1}{6}}\nonumber\\
    &\leq \mathfrak{b} i^{-\frac{1}{3}} d^{-\frac{1}{6}}\nonumber\\
    &\leq \frac{1}{30} \mathfrak{b}^2 d^{-\frac{1}{6}} i^{-\frac{1}{3}}
\end{align}
where the last inequality holds since $\mathfrak{b} \geq 10^6$.
Thus, \eqref{eq_c1b} implies that $\eta_{j_{\mathrm{min}}}  \in[b_{\mathrm{min}}, b_{\mathrm{max}}]$.
\end{proof}
Then, whenever the event $E$ occurs, we have that $ \eta_{j_{\mathrm{max}}}  \in[a_{\mathrm{min}}, a_{\mathrm{max}}]$ and  $ \eta_{j_{\mathrm{min}}}  \in[b_{\mathrm{min}}, b_{\mathrm{max}}]$.
Consider any $a,b$ such that $a_{\mathrm{min}} \leq a \leq  a_{\mathrm{max}}$  and $b_{\mathrm{min}} \leq b \leq  b_{\mathrm{max}}$.
 Define the sets 
\begin{itemize}
\item $S_0(a,b):= \cap\{ \eta \in \mathbb{R}^d : \eta_{j_{\mathrm{max}}} = a, \eta_{j_{\mathrm{min}}} = b\}$,

\item  $S_3(a,b; y) :=  \{\eta \in \mathbb{R}^d : \eta_i-\eta_{i+1} = y\} \cap S_0(a,b)$ for any $y \leq s\frac{1}{8\mathfrak{b}^4\sqrt{d} }$, and

\item $S_4(a,b) :=  \{\eta \in \mathbb{R}^d : \eta_i-\eta_{i+1} \geq s\} \cap S_0(a,b)$.

\end{itemize}

\paragraph{Step 3.} Define an injective map from the ``bad'' set $S_3$ to the good set $S_4$, which, roughly speaking, shows that the good set has a much bigger volume and a much larger probability density than the bad set.
  For any $y \leq s\frac{1}{\mathfrak{b}\sqrt{d}}$, we want to define an injective map $g:   S_2(a,b;y) \rightarrow S_4(a,b)$, such that its Jacobian $J_g(\eta)$ satisfies $\mathrm{det}(J_g(\eta)) \geq \frac{1}{2}$ for all $\eta \in \mathbb{R}^d$, and 
\begin{equation}\label{eq_b3}
    \frac{f(\eta)}{f(g(\eta))} \leq (\mathfrak{b} \sqrt{d} )^2 \times y^2,
    \end{equation}
  for any $\eta \in  S_3(a,b;y)$.
  Towards this end, we consider the map $g:  \mathcal{W}_d \rightarrow \mathcal{W}_d$ such that 

  \begin{itemize}
  
  \item 
\begin{equation}\label{eq_g3}
      g(\eta)[j] = \eta_j \qquad \forall j \notin [j_{\mathrm{max}}, j_{\mathrm{min}}]
      \end{equation}
  
  \item 
\begin{equation}\label{eq_g4}
      g(\eta)[j_{\mathrm{max}}] = \eta_{j_{\mathrm{max}}} = a,
      \end{equation}

\item   
\begin{equation}\label{eq_g2}
    g(\eta)[j] = g(\eta)[j+1] + (1-\alpha) \times (\eta_j - \eta_{j+1}) \qquad \forall j \in[j_{\mathrm{max}}, j_{\mathrm{min}}]\backslash \{i\}
\end{equation}
  
  \item  
  \begin{equation}\label{eq_g1}
  g(\eta)[i] = g(\eta)[i+1] + \left(\frac{2}{s}(\eta_{i}- \eta_{i+1})+ 2\frac{1}{8\mathfrak{b}^4\sqrt{d}}\right) \times \frac{b-a -(\eta_{i}- \eta_{i+1})}{b-a},
  \end{equation}
\end{itemize}
where $ \alpha := \frac{ \frac{2}{s}(\eta_{i}- \eta_{i+1}) +2\frac{1}{8\mathfrak{b}^4\sqrt{d}} }{b-a}$.

\begin{proposition} \label{prop_map}
Suppose that  $\mathfrak{b}^2 \leq i \leq d-\mathfrak{b}^2$.
Then the following properties hold for $g$:
\begin{itemize}

\item $g$ is injective.

\item $g(\eta)[j_{\mathrm{min}}] = \eta_{j_{\mathrm{min}}} = b,$

\item  $g(\eta)[i] -  g(\eta)[i+1]  \geq   \frac{1}{8\mathfrak{b}^4\sqrt{d}}$,  and hence  

\begin{equation}\label{eq_b5}
\frac{\eta_{i} - \eta_{i+1}}{g(\eta)[i] -  g(\eta)[i+1]}  \leq  8\mathfrak{b}^4 \sqrt{d}  \times (\eta_{i} - \eta_{i+1}) =  8 \mathfrak{b}^4 \sqrt{d}  \times y
\end{equation}
  for any $\eta \in  S_3(a,b;y)$ and any $y \leq s\frac{1}{8\mathfrak{b}^4\sqrt{d} }$.

\item  \begin{equation}\label{eq_b4}
g(\eta)[j] -  g(\eta)[j+1] \geq (1- \alpha)(\eta_{j} - \eta_{j+1}) \qquad \forall j \in [d].
\end{equation}
\end{itemize}

\end{proposition}

\begin{proof}

\noindent
{\em Injectivity:}
To prove that $g$ is injective, we note that, given any vector $z \in \mathbb{R}^d$ we can solve for the {\em unique} $\eta \in \mathcal{W}_d$ such that $g(\eta) = z$ whenever such a value of $\eta$ exits.

First, we solve for $\eta_{i}- \eta_{i+1}$ by solving the quadratic equation \eqref{eq_g1} for $\eta_{i}- \eta_{i+1}$, and noting that $\eta_{i}- \eta_{i+1}\geq 0$ (and hence we have a unique solution for $\eta_{i}- \eta_{i+1}$, whenever this solution exists).
This gives us the value of $\eta_{i}- \eta_{i+1}$ in terms of $g(\eta)[i] - g(\eta)[i+1]$.

Next, we use plug in the value of $\eta_{i}- \eta_{i+1}$ to compute $ \alpha = \frac{ \frac{2}{s}(\eta_{i}- \eta_{i+1}) +2\frac{1}{\mathfrak{b}\sqrt{d}} }{b-a}$, and  for every $j \in[j_{\mathrm{max}}, j_{\mathrm{min}}]$, plug in this value of $\alpha$ to \eqref{eq_g2}, to solve for $\eta_j-\eta_{j+1}$ in terms of $g(\eta)[j] - g(\eta)[j+1]$.

Finally, $\eta_{j_{\mathrm{max}}} = a$ we can compute each  $\eta_j = a+ \sum_{\ell=j}^{j_{\mathrm{max}} -1} \eta_\ell - \eta_{\ell+1}$ for each $j \in[j_{\mathrm{max}}+1, j_{\mathrm{min}}]$.
Thus, given any vector $z \in \mathbb{R}^d$ we can solve for the unique $\eta \in \mathbb{R}^d$ such that $g(\eta) = z$ whenever such a value of $\eta$ exits.
Therefore, $g$ is injective.

\medskip
\noindent
{\em Showing that $g(\eta)[j_{\mathrm{min}}] = b$:}
\begin{align*}
    g(\eta)[j_{\mathrm{min}}] &= a+ \sum_{\ell=j_{\mathrm{min}}}^{j_{\mathrm{max} -1}}  g(\eta)[\ell] -  g(\eta)[\ell+1]\\
&= a+  \left(\frac{2}{s}(\eta_{i}- \eta_{i+1})+ 2\frac{1}{8\mathfrak{b}^4\sqrt{d}}\right) \times \frac{b-a -(\eta_i- \eta_{i+1})}{b-a} \\
& \ \ \ + \sum_{\ell\in [j_{\mathrm{min}}, j_{\mathrm{max}} ] \backslash \{i\}}    (1-\alpha) \times (\eta_j - \eta_{j+1})\\
&=b.
\end{align*}

\medskip
\noindent
{\em Showing \eqref{eq_b5}:}
Since  $y \leq s\frac{1}{8\mathfrak{b}^4\sqrt{d} }$ and $\eta \in  S_3(a,b;y)$,
we have that $\eta_i-\eta_{i+1} = y \leq s\frac{1}{8\mathfrak{b}^4\sqrt{d} }$.
Thus by \eqref{eq_g1},
  \begin{align*}
  g(\eta)[i] - g(\eta)[i+1] &= \left(\frac{2}{s}(\eta_{i}- \eta_{i+1})+ 2\frac{1}{8\mathfrak{b}^4\sqrt{d}}\right) \times \frac{b-a -(\eta_{i}- \eta_{i+1})}{b-a}\\
  &\geq 2\frac{1}{8\mathfrak{b}^4\sqrt{d}} \times \frac{1}{2}\\
  &= \frac{1}{8\mathfrak{b}^4\sqrt{d}}.
  \end{align*}
Hence,  
\begin{equation}
\frac{\eta_{i} - \eta_{i+1}}{g(\eta)[i] -  g(\eta)[i+1]}  \leq  8\mathfrak{b}^4 \sqrt{d}  \times (\eta_{i} - \eta_{i+1}) =   8\mathfrak{b}^4 \sqrt{d}  \times y,
\end{equation}
which proves \eqref{eq_b5}.

\medskip
\noindent
{\em Showing \eqref{eq_b4}:}
By \eqref{eq_g1}, we have
  \begin{align*}
  g(\eta)[i] - g(\eta)[i+1] &= \left(\frac{2}{s}(\eta_{i}- \eta_{i+1})+ 2\frac{1}{\mathfrak{b}\sqrt{d}}\right) \times \frac{b-a -(\eta_{i}- \eta_{i+1})}{b-a}\\
  &\geq \left(\frac{2}{s}(\eta_{i}- \eta_{i+1})+ 2\frac{1}{\mathfrak{b}\sqrt{d}}\right) \times \frac{1}{2}\\
    &\geq \frac{1}{s}(\eta_{i}- \eta_{i+1})\\
        &\geq (\eta_{i}- \eta_{i+1})\\
        &\geq (1-\alpha)\times (\eta_{i}- \eta_{i+1}),
  \end{align*}
  where the second-to-last inequality holds since $s\leq 1$.
  Thus, \eqref{eq_b4} holds for $j=i$.
  Moreover, \eqref{eq_b4} holds for all  $j \in[j_{\mathrm{max}}, j_{\mathrm{min}}]\backslash \{i\}$ by \eqref{eq_g2} and  \eqref{eq_b4} holds for  all $j \notin [j_{\mathrm{max}}, j_{\mathrm{min}}]$ by \eqref{eq_g3}.
  Therefore \eqref{eq_b4} holds for all $j \in [d]$.
\end{proof}

\paragraph{Step 4.} Bounding the Jacobian determinant of the map $g$.

\begin{proposition}[Jacobian determinant of $g$]\label{prop_Jacobian}
If  $y \leq s\frac{1}{8\mathfrak{b}^4\sqrt{d} }$ and $\eta \in  S_3(a,b;y)$, we have that
$$\mathrm{det}(J_g(\eta)) \geq \frac{1}{2s}.$$
\end{proposition}
\begin{proof}
Consider the map $h: \mathcal{W}_{d} \rightarrow \mathbb{R}^{d}$, where   $h(\eta)[j] = \eta_j - \eta_{j+1}$ for $j \in[j_{\mathrm{min}}, j_{\mathrm{max}} -1]$ and $h(\eta)[j] = \eta_{j}$ for $j \in [1,d] \backslash [j_{\mathrm{min}}, j_{\mathrm{max}} -1]$.
Then for every $\Delta \in \mathcal{W}_{j_{\mathrm{max}} - j_{\mathrm{min}}}$, by \eqref{eq_g3}-\eqref{eq_g1}, we have that
\begin{equation} \label{eq_derivative1}
    \frac{\partial g\circ h^{-1}(\Delta)[\ell]}{\partial \Delta_j} = 0 \qquad \forall \ell \neq j,  j \neq i
\end{equation}

\begin{equation*}
    \frac{\partial g\circ h^{-1}(\Delta)[\ell]}{\partial \Delta_j} = (1-\alpha) \qquad \forall \ell = j \neq i, \quad  j \in  [j_{\mathrm{min}}, j_{\mathrm{max}}]
\end{equation*}

\begin{align*}
    \frac{\partial g\circ h^{-1}(\Delta)[i]}{\partial \Delta_i} &= \frac{2}{s} -2\frac{1}{\mathfrak{b}\sqrt{d} (b-a)} -\frac{1}{s (b-a)}(\eta_{i}- \eta_{i+1})\\
    &\geq \frac{1}{2s},
\end{align*}
\begin{equation*}
    \frac{\partial g\circ h^{-1}(\Delta)[j]}{\partial \Delta_{j}} = 1 \qquad \forall j \notin  [j_{\mathrm{min}}, j_{\mathrm{max}}].
\end{equation*}
where the inequality holds since $s \leq \frac{1}{\mathfrak{b}}$,  $\mathfrak{b} <1$, and $(\eta_{i}- \eta_{i+1}) \leq b-a$.
Moreover, since $\eta_{i}- \eta_{i+1} \leq s\frac{1}{8\mathfrak{b}^4\sqrt{d} }$ because $y \leq s\frac{1}{8\mathfrak{b}^4\sqrt{d} }$ and $\eta \in  S_3(a,b;y)$, we also have that $\alpha= \frac{ \frac{2}{s}(\eta_{i}- \eta_{i+1}) +2\frac{1}{8\mathfrak{b}^4\sqrt{d}} }{b-a} \leq  \frac{1}{(b-a)\mathfrak{b}^4\sqrt{d}}$.
Thus, $J_{g \circ h^{-1}}(\Delta)$  has diagonal entries $1-\alpha \geq 1- \frac{1}{(b-a)\mathfrak{b}^2\sqrt{d}}$ for $j \in  [j_{\mathrm{min}}, j_{\mathrm{max}}-1] \backslash \{i\}$,  and $i$'th entry $\geq \frac{1}{2s}$, and all other diagonal entries equal to $=1$.
Moreover, if one exchanges the $i$'th row and column of $J_{g \circ h^{-1}}(\Delta)$ with its first row and column,  by \eqref{eq_derivative1} the resulting matrix is a  $d\times d$ upper triangular matrix with the same determinant as  $J_{g \circ h^{-1}}(\Delta)$.
Thus, by Sylvester's formula, the determinant of $J_{g \circ h^{-1}}(\Delta)$ is equal to the product of its diagonal entries, and hence
\begin{align}\label{eq_g5}
\mathrm{det}(J_{g \circ h^{-1}}(\Delta)) &\geq \frac{1}{2s} \left(1- \frac{1}{(b-a)\mathfrak{b}^4\sqrt{d}}\right)^{j_{\mathrm{max}} - j_{\mathrm{min}}- 2} \times 1 \nonumber\\
&\geq \left(1- \frac{1}{(b-a)\mathfrak{b}^4\sqrt{d}}\right)^{2 \mathfrak{b}^2}\nonumber\\
&\geq \left(1- \frac{1}{(b-a)\mathfrak{b}^4\sqrt{d}}\right)^{2 \mathfrak{b}^2}\nonumber\\
&\geq \left(1- \frac{1}{\mathfrak{b}^2}\right)^{2 \mathfrak{b}^2} \nonumber\\
&\geq \frac{1}{16 s},
\end{align}
where the fourth inequality holds by Proposition \eqref{prop_classical} we have that $b-a \geq \frac{\mathfrak{b}^2}{\sqrt{d}}$.

Hence,

\begin{align*}
    \mathrm{det}(J_{g}(\eta)) &=   \mathrm{det}(J_{g \circ h^{-1} \circ h}(\eta)) = \mathrm{det}(J_{g \circ h^{-1}}(h(\eta))  \times J_{h}(\eta))\\
    &= \mathrm{det}(J_{g \circ h^{-1}}(h(\eta))  \times \mathrm{det}(J_{h}(\eta))\\
    &=\mathrm{det}(J_{g \circ h^{-1}}(h(\eta))  \times 1\\
    &\stackrel{\textrm{Eq. \eqref{eq_g5}}}{\geq} \frac{1}{16s},
\end{align*}
where the third equality holds since $J_{h}(\eta)$ is the bidiagonal matrix with diagonal entries $1$ and entries $-1$ above the diagonal, and this matrix has determinant $1$.
\end{proof}

\paragraph{Step 5.} {This step is a mean-field approximation for far-away eigenvalues.}

\begin{lemma}[Mean field approximation for far-away eigenvalues]\label{lemma_mean_field}
\begin{equation}
\prod_{j \in [j_{\mathrm{min}} , j_{\mathrm{max}}] , \, \, \, \ell \notin [j_{\mathrm{min}} - 2\mathfrak{b} , j_{\mathrm{max}} + 2\mathfrak{b}]} \frac{| \eta_{j} - \eta_\ell|^2}{| g(\eta)[j] - g(\eta)[\ell]|^2}  \leq 2,
\end{equation}

\end{lemma}

\begin{proof}
Consider any $j \in [j_{\mathrm{min}}, j_{\mathrm{max}}]$ and any $r \geq 2\mathfrak{b}$.
Then, if $\eta \in E$,  by the definition of the event $E$ we have that 
\begin{equation}\label{eq_b10}
    |\eta_j - \eta_{j+r} | = r \frac{1}{2\sqrt{d}} + \rho
\end{equation}
for some $\rho \geq 0$.
Moreover, from \eqref{eq_g2} and \eqref{eq_g1} that 
\begin{align}\label{eq_b9}
   |(g(\eta)[j] - g(\eta)[j+r]) - (\eta_j - \eta_{j+r})| &\leq \alpha (b-a)    \nonumber\\
    &=  \frac{2}{s}(\eta_{i}- \eta_{i+1}) +2\frac{1}{8\mathfrak{b}^4\sqrt{d}}    \nonumber\\
    &\leq \frac{1}{2\mathfrak{b}^4\sqrt{d}}
\end{align}
where the last inequality holds since $y \leq s\frac{1}{8\mathfrak{b}^4\sqrt{d}}$.
Thus, by \eqref{eq_b10} and \eqref{eq_b9} we have that for some $\zeta \in \mathbb{R}$ where $|\zeta| \leq \frac{1}{2\mathfrak{b}^4\sqrt{d}}$, we have 
\begin{align} \label{eq_b7}
\frac{| \eta_{j} - \eta_{j+r}|^2}{| g(\eta)[j] - g(\eta)[j+r]|^2}   \nonumber
& \leq \frac{|  r \frac{1}{2\sqrt{d}} + \rho  |^2}{| r \frac{1}{2\sqrt{d}} + \rho + \zeta|^2}  \nonumber\\
&\leq \left(1+ \frac{1}{r}   \times \frac{1}{\mathfrak{b}^4}\right)^2  \nonumber\\
&= \left(1+ \frac{1}{r}   \times \frac{1}{\mathfrak{b}^4}\right)^2,
\end{align}
where the second inequality holds since  $|\zeta| \leq \frac{1}{4\mathfrak{b}^4\sqrt{d}}$ and $\rho\geq 0$.
Therefore, we have 
\begin{align}  
&\prod_{j \in [j_{\mathrm{min}} , j_{\mathrm{max}}] , \, \, \, \ell \notin [j_{\mathrm{min}} - 2\mathfrak{b} , j_{\mathrm{max}} + 2\mathfrak{b}]} \frac{| \eta_{j} - \eta_\ell|^2}{| g(\eta)[j] - g(\eta)[\ell]|^2} \nonumber\\
& \stackrel{\textrm{Eq.  \eqref{eq_b7}}}{\leq} \prod_{\ell \in [j_{\mathrm{min}} , j_{\mathrm{max}}]} \prod_{r = 2 \mathfrak{b}}^d \left(1+ \frac{1}{r} \times \frac{1}{\mathfrak{b}^4}\right)^2  \nonumber\\
&\leq (e^{\frac{1}{\mathfrak{b}^3}})^{ j_{\mathrm{max}} -  j_{\mathrm{min}}} \nonumber\\
&=(e^{\frac{1}{\mathfrak{b}^3}})^{2 \mathfrak{b}^2} \nonumber\\
& = e^{\frac{1}{\mathfrak{b}}} \nonumber\\
& \geq 2\nonumber,
\end{align}
where the second inequality holds since  $\prod_{r=1}^{d}   (1+\frac{1}{r}) = d+1$  and hence that $\prod_{r=2}^{d}   (1+ \frac{1}{\kappa r}) \geq (d+1)^{\frac{1}{\kappa}}$ for every $\kappa \geq 1$.
Plugging in $\kappa = \mathfrak{b}^4$, we have %
$\prod_{r=2}^{n}   (1+ \frac{1}{\kappa r}) \geq (d+1)^{-\frac{1}{\mathfrak{b}^4}} \geq \left(d^{-\frac{1}{\log(d)^{\log \log d}}}\right)^{\frac{1}{\mathfrak{b}^3}} = e^{-\frac{1}{\mathfrak{b}^3}}$.
\end{proof}

\paragraph{Step 6.} Bounding the density ratio to show that $    \frac{f(\eta)}{f(g(\eta))} \leq \tilde{O}((\sqrt{d} \log d )^2 \times y^2)$.

\begin{lemma}\label{lemma_density_ratio}
For any $y \leq s\frac{1}{8\mathfrak{b}^4\sqrt{d} }$ and any $\eta\in  S_3(a,b;y)$, we have that

\begin{align*}
  \frac{f(\eta)}{f(g(\eta))} \leq 50 (8 \mathfrak{b}^4 \sqrt{d})^2 \times y^2,
\end{align*}
\end{lemma}

\begin{proof}
Since $\eta_{i}- \eta_{i+1} \leq s\frac{1}{8\mathfrak{b}^4\sqrt{d} }$ because $y \leq s\frac{1}{8\mathfrak{b}^4\sqrt{d} }$ and $\eta \in  S_3(a,b;y)$  , we also have that $\alpha= \frac{ \frac{2}{s}(\eta_{i}- \eta_{i+1}) +2\frac{1}{8\mathfrak{b}^4\sqrt{d}} }{b-a} \leq  \frac{1}{(b-a)\mathfrak{b}^4\sqrt{d}}$.
Therefore,
\begin{equation}\label{eq_b11}
    1-\alpha \geq 1- \frac{1}{(b-a)\mathfrak{b}^4\sqrt{d}} \geq 1- \frac{1}{\mathfrak{b}^2},
\end{equation}
where the last inequality holds by Proposition \eqref{prop_classical} we have that $b-a \geq \frac{\mathfrak{b}^2}{\sqrt{d}}$.

\begin{align*}
  \frac{f(\eta)}{f(g(\eta))} &=  \prod_{\ell<j,  :\, \, \, \ell,j \in [j_{\mathrm{min}} - 2 \mathfrak{b} , j_{\mathrm{max}} + 2 \mathfrak{b}]} \frac{| \eta_{\ell} - \eta_j|^2}{| g(\eta)[\ell] - g(\eta)[j]|^2} e^{-\frac{1}{2} \sum_{\ell=j_{\mathrm{min}}}^{j_{\mathrm{max}}}  (\eta_{\ell}^2 - g(\eta)[\ell]^2 )}\\
&   \times \prod_{j \in [j_{\mathrm{min}} , j_{\mathrm{max}}] , \, \, \, \ell \notin [j_{\mathrm{min}}  - 2 \mathfrak{b}  , j_{\mathrm{max}}  + 2 \mathfrak{b} ]} \frac{| \eta_{j} - \eta_\ell|^2}{| g(\eta)[j] - g(\eta)[\ell]|^2}\\
  &\stackrel{\textrm{Eq. \eqref{eq_b4},  \eqref{eq_b5}, Lemma \ref{lemma_mean_field}}}{\leq} ( 8 \mathfrak{b}^4 \sqrt{d}  \times y)^2 \times (1-\alpha)^{j_{\mathrm{max}} - j_{\mathrm{min}} +2 \mathfrak{b}} \times e^{\frac{1}{2}(j_{\mathrm{max}} - j_{\mathrm{min}} +2 \mathfrak{b})\alpha(b-a)\times 2\sqrt{d}} \times 2 \\
  &\leq ( 8 \mathfrak{b}^4 \sqrt{d}  \times y)^2 \times (1-\alpha)^{3\mathfrak{b}^2} \times e^{\frac{1}{2}(j_{\mathrm{max}} - j_{\mathrm{min}} +2 \mathfrak{b})\alpha(b-a)\times 2\sqrt{d}} \times 2 \\
   &\leq 2 e^3 (8 \mathfrak{b}^4 \sqrt{d})^2 \times y^2\\
 &\leq 50 (8 \mathfrak{b}^4 \sqrt{d})^2 \times y^2,
\end{align*}
where the first inequality holds since $|\eta_\ell| \leq 2 \sqrt{d}$ whenever $\eta \in E$, 
and the third inequality holds since  $1-\alpha \geq 1- \frac{1}{\mathfrak{b}^2}$ by \eqref{eq_b11} and since $\alpha (b-a)\leq \frac{1}{2\mathfrak{b}^4\sqrt{d}}$ by \eqref{eq_b9}.
\end{proof}

\paragraph{Step 7.} Dealing with the eigenvalues near the edge of the spectrum.

In this step, we extend the results of the previous steps to the eigenvalues which are near the edge of the spectrum.

  Suppose that $i \leq 2\mathfrak{b}^2$.
  In place of the map $g$, we instead consider the map $\phi:  \mathcal{W}_d \rightarrow \mathcal{W}_d$ such that 

  \begin{itemize}

  \item 
\begin{equation}\label{eq_phi1}
      \phi(\eta)[j] = \eta_j \qquad \forall j > i
      \end{equation}

 \item \begin{equation}\label{eq_phi2}
      \phi(\eta)[i] =  \eta_{i+1} + \frac{2}{s}(\eta_{i}- \eta_{i+1})+ 2\frac{1}{8\mathfrak{b}^4\sqrt{d}}.
      \end{equation} 
      
       \item \begin{equation}\label{eq_phi3}
      \phi(\eta)[j] =  \phi(\eta)[j+1] + (\eta_{j}- \eta_{j+1}) \qquad \forall j \leq i.
      \end{equation} 
  
\end{itemize}

\begin{proposition} \label{prop_map_phi}
Suppose that  $ i \leq 2\mathfrak{b}^2$.
Then the following properties hold for $\phi$:
\begin{itemize}

\item $\phi$ is injective.

\item  $\phi(\eta)[i] -  \phi(\eta)[i+1]  \geq   \frac{1}{8\mathfrak{b}^4\sqrt{d}}$,  and hence  

\begin{equation}\label{eq_b5e}
\frac{\eta_{i} - \eta_{i+1}}{\phi(\eta)[i] -  \phi(\eta)[i+1]}  \leq  8\mathfrak{b}^4 \sqrt{d}  \times (\eta_{i} - \eta_{i+1}) =  8 \mathfrak{b}^4 \sqrt{d}  \times y
\end{equation}
  for any $\eta \in  S_3(a,b;y)$ and any $y \leq s\frac{1}{8\mathfrak{b}^4\sqrt{d} }$.

\item  \begin{equation}\label{eq_b4e}
\phi(\eta)[j] -  \phi(\eta)[j+1] \geq \eta_{j} - \eta_{j+1} \qquad \forall j \in [d].
\end{equation}
\end{itemize}

\end{proposition}

\begin{proof}

\medskip
\noindent
{\em Injectivity:}
To prove that $g$ is injective, we note that, given any vector $z \in \mathbb{R}^d$ we can find the {\em unique} $\eta \in \mathcal{W}_d$ such that $g(\eta) = z$ whenever such a value of $\eta$ exits.
We can do this by solving the system of linear equations given by  \eqref{eq_phi1}-\eqref{eq_phi3}:
First, we note that by \eqref{eq_phi1}, we can solve for $\eta_j$ for all $j >i$.
Then we can plug in the value we found for  $\eta_{i+1}$ into \eqref{eq_phi2} to solve for $\eta_i$.
Finally, we can use \eqref{eq_phi3} to solve for $\eta_j$ for all $j<i$ recursively, starting with $\eta_{i-1}$.

\medskip
\noindent
{\em Showing \eqref{eq_b5e}:}
Since  $y \leq s\frac{1}{8\mathfrak{b}^4\sqrt{d} }$ and $\eta \in  S_3(a,b;y)$,
we have that $\eta_i-\eta_{i+1} = y \leq s\frac{1}{8\mathfrak{b}^4\sqrt{d} }$.
Thus, by \eqref{eq_phi2},
  \begin{align*}
  g(\eta)[i] - g(\eta)[i+1] &= \frac{2}{s}(\eta_{i}- \eta_{i+1})+ 2\frac{1}{8\mathfrak{b}^4\sqrt{d}}\\
  &\geq 2\frac{1}{8\mathfrak{b}^4\sqrt{d}}\\
  &\geq \frac{1}{8\mathfrak{b}^4\sqrt{d}}.
  \end{align*}
and, hence, 
\begin{equation*}
\frac{\eta_{i} - \eta_{i+1}}{g(\eta)[i] -  g(\eta)[i+1]}  \leq  8\mathfrak{b}^4 \sqrt{d}  \times (\eta_{i} - \eta_{i+1}) =   8\mathfrak{b}^4 \sqrt{d}  \times y,
\end{equation*}
which proves \eqref{eq_b5e}.

\medskip
\noindent
{\em Showing \eqref{eq_b4e}:}
By \eqref{eq_phi2}, we have
  \begin{align*}
  g(\eta)[i] - g(\eta)[i+1] &= \frac{2}{s}(\eta_{i}- \eta_{i+1})+ 2\frac{1}{8\mathfrak{b}^4\sqrt{d}}\\
    &\geq \frac{1}{s}(\eta_{i}- \eta_{i+1})\\
        &\geq \eta_{i}- \eta_{i+1}.
  \end{align*}
  where the last inequality holds since $s\leq 1$.
  Thus, \eqref{eq_b4e} holds for $j=i$.
  Moreover, \eqref{eq_b4e} holds for all  $j \neq i$ by \eqref{eq_phi2}.
  Therefore \eqref{eq_b4e} holds for all $j \in [d]$.
\end{proof}

\begin{proposition}[Jacobian determinant of $\phi$]\label{prop_Jacobian_phi}
If  $y \leq s\frac{1}{8\mathfrak{b}^4\sqrt{d} }$ and $\eta \in  S_3(a,b;y)$, we have that
$$\mathrm{det}(J_{\phi}(\eta)) = \frac{2}{s}.$$
\end{proposition}
\begin{proof}
Consider the map $h: \mathcal{W}_{d} \rightarrow \mathbb{R}^{d}$, where   $h(\eta)[j] = \eta_j - \eta_{j+1}$ for $j \in[j_{\mathrm{min}}, j_{\mathrm{max}} -1]$ and $h(\eta)[j] = \eta_{j}$ for $j \in [1,d] \backslash [j_{\mathrm{min}}, j_{\mathrm{max}} -1]$.
Then, for every $\Delta \in \mathcal{W}_{j_{\mathrm{max}} - j_{\mathrm{min}}}$, by \eqref{eq_g3}-\eqref{eq_g1}, we have that
\begin{equation} \label{eq_derivative1.2}
    \frac{\partial \phi \circ h^{-1}(\Delta)[\ell]}{\partial \Delta_j} = 0 \qquad \forall \ell \neq j
\end{equation}

\begin{equation*}
    \frac{\partial \phi\circ h^{-1}(\Delta)[\ell]}{\partial \Delta_j} = 1 \qquad \forall \ell = j \neq i
\end{equation*}

\begin{align*}
    \frac{\partial \phi\circ h^{-1}(\Delta)[i]}{\partial \Delta_i} = \frac{2}{s},
\end{align*}
Thus, $J_{\phi \circ h^{-1}}(\Delta)$  has diagonal entries $1$ for $j \neq i$ and $i$'th entry $= \frac{2}{s}$.
Moreover, if one exchanges the $i$'th row and column of $J_{\phi \circ h^{-1}}(\Delta)$ with its first row and column,  by \eqref{eq_derivative1.2} the resulting matrix is a  $d\times d$ upper triangular matrix with the same determinant as  $J_{\phi \circ h^{-1}}(\Delta)$.
Thus, by Sylvester's formula, the determinant of $J_{g \circ h^{-1}}(\Delta)$ is equal to the product of its diagonal entries, and hence
\begin{align}\label{eq_g5b}
\mathrm{det}(J_{\phi \circ h^{-1}}(\Delta)) &= \frac{2}{s} \times 1 = \frac{2}{s}.
\end{align}
Hence,

\begin{align*}
    \mathrm{det}(J_{\phi}(\eta)) &=   \mathrm{det}(J_{\phi \circ h^{-1} \circ h}(\eta)) = \mathrm{det}(J_{\phi \circ h^{-1}}(h(\eta))  \times J_{h}(\eta))\\
    &= \mathrm{det}(J_{\phi \circ h^{-1}}(h(\eta))  \times \mathrm{det}(J_{h}(\eta))\\
    &=\mathrm{det}(J_{\phi \circ h^{-1}}(h(\eta))  \times 1\\
    &\stackrel{\textrm{Eq. \eqref{eq_g5b}}}{=} \frac{2}{s},
\end{align*}
where the third equality holds since $J_{h}(\eta)$ is the bidiagonal matrix with diagonal entries $1$ and entries $-1$ above the diagonal, and this matrix has determinant $1$.
\end{proof}

\begin{lemma}[Mean field approximation, edge case]\label{lemma_mean_field_phi}
\begin{equation}
\prod_{j \in [j_{\mathrm{min}} , j_{\mathrm{max}}] , \, \, \, \ell \notin [j_{\mathrm{min}} - 2\mathfrak{b} , j_{\mathrm{max}} + 2\mathfrak{b}]} \frac{| \eta_{j} - \eta_\ell|^2}{| \phi(\eta)[j] - \phi(\eta)[\ell]|^2}  \leq 2,
\end{equation}
\end{lemma}

\begin{proof}
Consider any $j \in [j_{\mathrm{min}}, j_{\mathrm{max}}]$ and any $r \geq 2\mathfrak{b}$.
Then, if $\eta \in E$,  by the definition of the event $E$ we have that 
\begin{equation}\label{eq_b10e}
    |\eta_j - \eta_{j+r} | = r \frac{1}{2\sqrt{d}} + \rho
\end{equation}
for some $\rho \geq 0$.
Further, from \eqref{eq_phi1}-\eqref{eq_phi3} we have that 
\begin{align}\label{eq_b9e}
   |(\phi(\eta)[j] - \phi(\eta)[j+r]) - (\eta_j - \eta_{j+r})| &\leq \frac{2}{s}(\eta_{i}- \eta_{i+1}) +2\frac{1}{8\mathfrak{b}^4\sqrt{d}}    \nonumber\\
    &\leq \frac{1}{2\mathfrak{b}^4\sqrt{d}}
\end{align}
where the last inequality holds since $y \leq s\frac{1}{8\mathfrak{b}^4\sqrt{d}}$.
Thus, by \eqref{eq_b10e} and \eqref{eq_b9e} we have that for some $\zeta \in \mathbb{R}$ where $|\zeta| \leq \frac{1}{2\mathfrak{b}^4\sqrt{d}}$, we have 
\begin{align} \label{eq_b7.2}
\frac{| \eta_{j} - \eta_{j+r}|^2}{| \phi(\eta)[j] - \phi(\eta)[j+r]|^2}   \nonumber
& \leq \frac{|  r \frac{1}{2\sqrt{d}} + \rho  |^2}{| r \frac{1}{2\sqrt{d}} + \rho + \zeta|^2}  \nonumber\\
&\leq \left(1+ \frac{1}{r}   \times \frac{1}{\mathfrak{b}^4}\right)^2  \nonumber\\
&= \left(1+ \frac{1}{r}   \times \frac{1}{\mathfrak{b}^4}\right)^2,
\end{align}
where the second inequality holds since  $|\zeta| \leq \frac{1}{4\mathfrak{b}^4\sqrt{d}}$ and $\rho\geq 0$.
Therefore, we have 
\begin{align} 
&\prod_{j \in [j_{\mathrm{min}} , j_{\mathrm{max}}] , \, \, \, \ell \notin [j_{\mathrm{min}} - 2\mathfrak{b} , j_{\mathrm{max}} + 2\mathfrak{b}]} \frac{| \eta_{j} - \eta_\ell|^2}{| \phi(\eta)[j] - \phi(\eta)[\ell]|^2} \nonumber\\
& \stackrel{\textrm{Eq.  \eqref{eq_b7.2}}}{\leq} \prod_{\ell \in [j_{\mathrm{min}} , j_{\mathrm{max}}]} \prod_{r = 2 \mathfrak{b}}^d \left(1+ \frac{1}{r} \times \frac{1}{\mathfrak{b}^4}\right)^2  \nonumber\\
&\leq (e^{\frac{1}{\mathfrak{b}^3}})^{ j_{\mathrm{max}} -  j_{\mathrm{min}}} \nonumber\\
&=(e^{\frac{1}{\mathfrak{b}^3}})^{2 \mathfrak{b}^2} \nonumber\\
& = e^{\frac{1}{\mathfrak{b}}} \nonumber\\
& \geq 2\nonumber,
\end{align}
where the second inequality holds since  $\prod_{r=1}^{d}   (1+\frac{1}{r}) = d+1$  and hence that $\prod_{r=2}^{d}   (1+ \frac{1}{\kappa r}) \geq (d+1)^{\frac{1}{\kappa}}$ for every $\kappa \geq 1$.
Plugging in $\kappa = \mathfrak{b}^4$, we have %
$\prod_{r=2}^{n}   (1+ \frac{1}{\kappa r}) \geq (d+1)^{-\frac{1}{\mathfrak{b}^4}} \geq \left(d^{-\frac{1}{\log(d)^{\log \log d}}}\right)^{\frac{1}{\mathfrak{b}^3}} = e^{-\frac{1}{\mathfrak{b}^3}}$.
\end{proof}

\begin{lemma}\label{lemma_density_ratio_edge}
For any $y \leq s\frac{1}{8\mathfrak{b}^4\sqrt{d} }$ and any $\eta\in  S_3(a,b;y)$, we have that

\begin{align*}
  \frac{f(\eta)}{f(\phi(\eta))} \leq 4 (8 \mathfrak{b}^4 \sqrt{d})^2 \times y^2,
\end{align*}
\end{lemma}

\begin{proof}

\begin{align*}
  \frac{f(\eta)}{f(\phi(\eta))} &=  \prod_{\ell<j,  :\, \, \, \ell,j \in [j_{\mathrm{min}} - 2 \mathfrak{b} , j_{\mathrm{max}} + 2 \mathfrak{b}]} \frac{| \eta_{\ell} - \eta_j|^2}{| \phi(\eta)[\ell] - \phi(\eta)[j]|^2} e^{-\frac{1}{2} \sum_{\ell=i}^{1}  (\eta_{\ell}^2 - \phi(\eta)[\ell]^2 )}\\
&   \times \prod_{j \in [j_{\mathrm{min}} , j_{\mathrm{max}}] , \, \, \, \ell \notin [j_{\mathrm{min}}  - 2 \mathfrak{b}  , j_{\mathrm{max}}  + 2 \mathfrak{b} ]} \frac{| \eta_{j} - \eta_\ell|^2}{| \phi(\eta)[j] - \phi(\eta)[\ell]|^2}\\
  &\stackrel{\textrm{Eq. \eqref{eq_b4e},  \eqref{eq_b5e}, Lemma \ref{lemma_mean_field_phi}}}{\leq} ( 8 \mathfrak{b}^4 \sqrt{d}  \times y)^2 \times 1 \times e^{\mathfrak{b}^2 \times \frac{1}{8\mathfrak{b}^4 \sqrt{d}} \times 2\sqrt{d}} \times 2 \\
  &\leq ( 8 \mathfrak{b}^4 \sqrt{d}  \times y)^2 \times e^{\frac{1}{4\mathfrak{b}^2}}  \times 2 \\
   &\leq 4 (8 \mathfrak{b}^4 \sqrt{d})^2 \times y^2,
 \end{align*}
where the first inequality holds since $|\eta_\ell| \leq 2 \sqrt{d}$ whenever $\eta \in E$.
\end{proof}

\paragraph{Step 8.} Completing the proof.

\begin{proof}[Proof of Lemma \ref{lemma_GUE_gaps}]

\medskip
\noindent
{\em Bulk case ($\mathfrak{b}^2 \leq i \leq d-\mathfrak{b}^2$):}
By Proposition \ref{prop_map} we have that $g$ is invertible.
Therefore, since $f$ is a probability density,
\begin{equation*}
\int_{a_{\min}}^{a_{\max}} \int_{b_{\min}}^{b_{\max}} \int_0^{s\frac{1}{8\mathfrak{b}^4\sqrt{d}}} \int_{S_3(a,b;y) \cap E} f(g(\eta)) \mathrm{det}(J_g(\eta)) \mathrm{d} \eta \mathrm{d}y \mathrm{d} a \mathrm{d} b \leq  \int_{\mathcal{W}_d} f(\eta) \mathrm{d} \eta = 1.
\end{equation*}
Therefore,

\begin{align*}
1 \geq \int_{a_{\min}}^{a_{\max}} \int_{b_{\min}}^{b_{\max}} \int_0^{s\frac{1}{8\mathfrak{b}^4\sqrt{d}}} \int_{S_3(a,b;y) \cap E} \frac{f(g(\eta))}{f(\eta)}  \mathrm{det}(J_g(\eta))\times f(\eta) \mathrm{d} \eta \mathrm{d}y \mathrm{d} a \mathrm{d} b\\
\geq \int_{a_{\min}}^{a_{\max}} \int_{b_{\min}}^{b_{\max}} \int_0^{s\frac{1}{8\mathfrak{b}^4\sqrt{d}}} \int_{S_3(a,b;y) \cap E} \frac{1}{50 (8 \mathfrak{b}^4 \sqrt{d})^2 \times y^2} \times \frac{1}{2s} \times f(\eta) \mathrm{d} \eta \mathrm{d}y \mathrm{d} a \mathrm{d} b\\
\geq \int_{a_{\min}}^{a_{\max}} \int_{b_{\min}}^{b_{\max}} \int_0^{s\frac{1}{8\mathfrak{b}^4\sqrt{d}}} \int_{S_3(a,b;y) \cap E} \frac{1}{50   s^2} \times \frac{1}{2s} \times f(\eta) \mathrm{d} \eta \mathrm{d}y \mathrm{d} a \mathrm{d} b\\
= \frac{1}{100   s^3}  \int_{a_{\min}}^{a_{\max}} \int_{b_{\min}}^{b_{\max}} \int_0^{s\frac{1}{8\mathfrak{b}^4\sqrt{d}}} \int_{S_3(a,b;y) \cap E}  f(\eta) \mathrm{d} \eta \mathrm{d}y \mathrm{d} a \mathrm{d} b,
\end{align*}
where the second equality holds by Lemma \ref{lemma_density_ratio} and Proposition \ref{prop_Jacobian}.
Therefore,
\begin{equation}\label{eq_f1}
 \int_{a_{\min}}^{a_{\max}} \int_{b_{\min}}^{b_{\max}} \int_0^{s\frac{1}{8\mathfrak{b}^4\sqrt{d}}} \int_{S_3(a,b;y) \cap E}  f(\eta) \mathrm{d} \eta \mathrm{d}y \mathrm{d} a \mathrm{d} b \leq 100 s^3.
\end{equation}
Hence,
\begin{align*}
    &\mathbb{P}\left(\eta_i - \eta_{i+1} \leq s\frac{1}{8\mathfrak{b}^4\sqrt{d} }\right) 
    \leq \mathbb{P}\left(\left\{\eta_i - \eta_{i+1} \leq  s\frac{1}{8\mathfrak{b}^4\sqrt{d} }\right\} \cap E\right)+ \mathbb{P}(E^c)\\
        &= \int_{\left \{\eta \in \mathcal{W}_d \, \, \, : \, \, \,\eta_i - \eta_{i+1} \leq  s\frac{1}{8\mathfrak{b}^4\sqrt{d} } \right\} \cap E} f(\eta) \mathrm{d} \eta + \mathbb{P}(E^c)\\
&=  \int_{a_{\min}}^{a_{\max}} \int_{b_{\min}}^{b_{\max}} \int_{\left \{\eta \in \mathcal{W}_d \, \, \, : \, \, \,\eta_i - \eta_{i+1} \leq  s\frac{1}{8\mathfrak{b}^4\sqrt{d} } \right\} \cap E \cap \{\eta \in \mathcal{W}_d \, \, \, : \, \, \, \eta_{j_{\mathrm{min}}} = a, \, \, \eta_{j_{\mathrm{max}}} = b\}} f(\eta) \mathrm{d} \eta \mathrm{d} a \mathrm{d} b + \mathbb{P}(E^c)\\
    &=\int_{a_{\min}}^{a_{\max}} \int_{b_{\min}}^{b_{\max}} \int_0^{s\frac{1}{8\mathfrak{b}^4\sqrt{d}}} \int_{S_3(a,b;y) \cap E} f(\eta) \mathrm{d} \eta \mathrm{d}y \mathrm{d} a \mathrm{d} b + \mathbb{P}(E^c)\\
    &\stackrel{\textrm{Eq. \eqref{eq_f1}}}{\leq}  100 s^3 + \mathbb{P}(E^c)\\
        &\stackrel{\textrm{Eq. \eqref{eq_rigidity_1}}}{\leq} 100 s^3 + \frac{1}{d^{1000}}.
\end{align*}
Redefining the universal constant $L$ (and hence redefining $\mathfrak{b}$), we get that
\begin{equation*}
    \mathbb{P}\left(\eta_i - \eta_{i+1} \leq s \frac{1}{\mathfrak{b}\sqrt{d}}\right) \leq s^{3} + \frac{1}{d^{1000}} \qquad \forall s>0
\end{equation*}
which proves  Lemma \ref{lemma_GUE_gaps} for any $\mathfrak{b}^2 \leq i \leq d-\mathfrak{b}^2$.

\medskip
\noindent
{\em Edge case ($\min(i, d-i) \leq \mathfrak{b}^2$):}
Since the joint density of the eigenvalues  \eqref{eq_joint_density} is symmetric about $0$, without loss of generality we may assume that $i\leq \mathfrak{b}^2$.

The proof of Lemma \ref{lemma_GUE_gaps} for the edge case $i\leq \mathfrak{b}^2$ follows exactly the same steps as for the bulk case ($\mathfrak{b}^2 \leq i \leq d-\mathfrak{b}^2$), if we replace the map $g$ with the map $\phi$, Proposition \ref{prop_map} with Proposition \ref{prop_map_phi}, Lemma \ref{lemma_density_ratio} with \ref{lemma_density_ratio_edge}, and Proposition \ref{prop_Jacobian} with Proposition  \ref{prop_Jacobian_phi}.
\end{proof}

\section{Conclusion and future work}\label{sec:conclusion}

We present and analyze a complex variant of the Gaussian mechanism for rank-$k$ covariance matrix approximation under $(\epsilon, \delta)$-differential privacy.
Our analysis leverages the fact that the eigenvalues of  complex matrix Brownian motion repel more than in the real case, and uses Dyson's stochastic differential equations governing the evolution of its eigenvalues to show that, after any time $t>0$,
the eigenvalues of the matrix $M$ perturbed by complex Gaussian noise have large gaps of size $\tilde{\Omega}\left(s\frac{\sqrt{t}}{\sqrt{d}}\right)$ with high probability $1- O(s^3)$.
We believe the decay rate $1- O(s^3)$ in our eigenvalue gap bound is tight, as its derivative corresponds has the same exponent $\beta = 2$ which appears in the joint eigenvalue density formula \eqref{eq_joint_density} for the complex GUE random matrix.

We suspect our techniques can also be useful for other matrix approximation problems.
              For instance, one may consider the more general problem where one is given a covariance matrix $M$ and a function $f: \mathbb{R}^{d\times d} \rightarrow \mathbb{R}^{d\times d}$ (e.g.,  the matrix exponential), and the goal is to find a rank-$k$ symmetric matrix $Y$ which minimizes $\|Y -f(M)\|_F$ under $(\epsilon, \delta)$-differential privacy.
The main  problem left open by our work  is if the complex Gaussian mechanism is necessary to achieve our utility bounds, or can our analysis be extended to the real Gaussian mechanism.

\paragraph{Acknowledgments.}
OM was supported in part by an NSF CCF-2104528 award and a Google Research Scholar award. 
NV was supported in part by an NSF CCF-2112665 award.

\bibliography{DP}
\bibliographystyle{plain}

\appendix

\section{Extensions of prior results to complex matrices}

\subsection{Proof of Lemma \ref{Lemma_integral}} \label{Appendix_proof_of_Lemma_integral}

\begin{proof}[Proof of Lemma \ref{Lemma_integral}; modification of the proof of Lemma 4.5 in \cite{DBM_Neurips}, for complex matrices]
By the definition of $Z_\eta(t)$ we have that
\begin{align*}
Z_\eta\left(T\right) -  Z_\eta(t_0) &= \int_{t_0}^{T} \mathrm{d}Z_\eta(t)\\
&=  \frac{1}{2}  \int_{t_0}^{T}\sum_{i=1}^{d} \sum_{j \neq i} |\lambda_i(t) - \lambda_j(t)| \frac{1}{\max(|\Delta_{ij}(t)|, \eta_{ij})}(u_i(t) u_j^\ast(t)\mathrm{d}B_{ij}(t) + u_j(t) u_i^\ast(t)\mathrm{d}B_{ij}^\ast(t)) \\
    & -   \int_{t_0}^{T}\sum_{i=1}^{d} \sum_{j\neq i} (\lambda_i(t) - \lambda_j(t)) \frac{\mathrm{d}t}{\max(\Delta^2_{ij}(t), \eta_{ij}^2)} u_i(t) u_i^\ast(t).
\end{align*}
Therefore, we have that
\begin{align} \label{eq_int_1}
&\left\|Z_\eta(T) -  Z_\eta(t_0)\right \|_F^2\nonumber\\
&\leq  \frac{1}{2} \left\| \int_{t_0}^{T}\sum_{i=1}^{d} \sum_{j \neq i} |\lambda_i(t) - \lambda_j(t)| \frac{1}{\max(|\Delta_{ij}(t)|, \eta_{ij})}(u_i(t) u_j^\ast(t)\mathrm{d}B_{ij}(t) + u_j(t) u_i^\ast(t)\mathrm{d}B_{ij}^\ast(t))\right\|_F^2 \nonumber\\
    & +  \left\| \int_{t_0}^{T}\sum_{i=1}^{d} \sum_{j\neq i} (\lambda_i(t) - \lambda_j(t)) \frac{\mathrm{d}t}{\max(\Delta^2_{ij}(t), \eta_{ij}^2)} u_i(t) u_i^\ast(t) \right\|_F^2. 
\end{align}
The first term on the r.h.s. of \eqref{eq_int_1} (inside its Frobenius norm) is a ``diffusion'' term--that is, the integral has mean $0$ and Brownian motion differentials $\mathrm{d}B_{ij}(t)$ inside the integral.
The second term on the r.h.s. (inside its Frobenius norm) is a ``drift'' term-- that is, the integral has non-zero  mean and deterministic differentials $\mathrm{d}t$ inside the integral.
We bound the diffusion and drift terms separately.

\paragraph{Bounding the diffusion term:}

We first use It\^o's Lemma (Lemma \ref{lemma_ito_lemma_new}) to bound the diffusion term in \eqref{eq_int_1}.
The idea is to apply Ito's Lemma separately to the real and complex parts of the integrand.
Define $$X(t):=  \int_{t_0}^{t}\sum_{i=1}^{d} \sum_{j \neq i} |\lambda_i(t) - \lambda_j(t)| \frac{1}{\max(|\Delta_{ij}(s)|, \eta_{ij})}(u_i(s) u_j^\ast(s)\mathrm{d}B_{ij}(s) + u_j(s) u_i^\ast(s)\mathrm{d}B_{ij}^\ast(s))$$ for all $t \geq 0.$

Then
\begin{equation*}
\mathrm{d}X_{\ell r}(t) = \sum_{j=1}^d R_{(\ell r) (i j)}(t) \mathrm{d}B_{(ij)}(t) + Q_{(\ell r) (i j)}(t) \mathrm{d}B_{(ij)}^\ast(t) \qquad \qquad \forall t \geq 0,
\end{equation*}
where $$R_{(\ell r) (i j)}(t) := \left(\frac{ |\lambda_i(t) - \lambda_j(t)| }{\max(|\Delta_{ij}(t)|, \eta_{ij})}u_i(t) u_j^\ast(t) \right)[\ell, r]$$
and 
$$Q_{(\ell r) (i j)}(t) := \left(\frac{ |\lambda_i(t) - \lambda_j(t)| }{\max(|\Delta_{ij}(t)|, \eta_{ij})} u_j(t) u_i^\ast(t) \right)[\ell, r],$$ and where we denote by either $H_{\ell r}$ or $H[\ell, r]$  the $(\ell, r)$'th entry of any matrix $H$.
Thus,
\begin{align}\label{eq_t1}
\mathrm{d}X_{\ell r}(t) &= \sum_{j=1}^d R_{(\ell r) (i j)}(t) \mathrm{d}B_{(ij)}(t) + Q_{(\ell r) (i j)}(t) \mathrm{d}B_{(ij)}^\ast(t) \qquad \qquad \forall t \geq 0,\nonumber\\
&= \sum_{j=1}^d [\mathcal{R}(R_{(\ell r) (i j)}(t)) + \mathfrak{i}  \mathcal{I}(R_{(\ell r) (i j)}(t))] \times  [\mathcal{R}(\mathrm{d}B_{(ij)}(t)) + \mathfrak{i} \mathcal{I}(\mathrm{d}B_{(ij)}(t))]\nonumber\\
&+ \sum_{j=1}^d [\mathcal{R}(Q_{(\ell r) (i j)}(t)) + \mathfrak{i}  \mathcal{I}(Q_{(\ell r) (i j)}(t))]\times  [\mathcal{R}(\mathrm{d}B_{(ij)}^\ast(t)) + \mathfrak{i} \mathcal{I}(\mathrm{d}B_{(ij)}^\ast(t))]\nonumber\\
&= \sum_{j=1}^d \mathcal{R}(R_{(\ell r) (i j)}(t))\mathcal{R}(\mathrm{d}B_{(ij)}(t))  + \mathfrak{i}  \mathcal{I}(R_{(\ell r) (i j)}(t)) \mathcal{R}(\mathrm{d}B_{(ij)}(t))\nonumber\\
&\qquad \qquad+ \mathfrak{i}\mathcal{R}(R_{(\ell r) (i j)}(t)) \mathcal{I}(\mathrm{d}B_{(ij)}(t)) - \mathcal{I}(R_{(\ell r) (i j)}(t)) \mathcal{I}(\mathrm{d}B_{(ij)}(t))\nonumber\\
&+ \sum_{j=1}^d \mathcal{R}(Q_{(\ell r) (i j)}(t))\mathcal{R}(\mathrm{d}B_{(ij)}^\ast(t))  + \mathfrak{i}  \mathcal{I}(Q_{(\ell r) (i j)}(t)) \mathcal{R}(\mathrm{d}B_{(ij)}^\ast(t))\nonumber\\
&\qquad \qquad+ \mathfrak{i}W\mathcal{R}(Q_{(\ell r) (i j)}(t)) \mathcal{I}(\mathrm{d}B_{(ij)}^\ast(t)) - \mathcal{I}(Q_{(\ell r) (i j)}(t)) \mathcal{I}(\mathrm{d}B_{(ij)}^\ast(t)).
\end{align}

\noindent
Our goal is to bound $\mathbb{E}[\|X(T) - X(t_0)\|_F^2]$. 
Towards this end, let $f: \mathbb{R}^{d \times d}: \rightarrow \mathbb{R}$ be the function which takes as input a $d \times d$ matrix and outputs the square of its Frobenius norm: $f(Y):= \|Y\|_F^2 = \sum_{i=1}^d \sum_{j=1}^d Y_{ij}^2$ for every $Y \in \mathbb{R}^{d\times d}$.
Then
\begin{equation}\label{eq_int_5}
\frac{\partial^2 }{\partial Y_{ij} \partial Y_{\alpha \beta}}f(Y) =
\begin{cases} 
      2 & \textrm{if }  \, \, \, (i,j)= (\alpha, \beta) \\
      0 & \textrm{otherwise}.
   \end{cases}
   \end{equation}
   
\noindent
Then by \eqref{eq_t1} we have
\begin{align}\label{eq_t2}
    &\|X(T) - X(t_0)\|_F^2 =  \left\|\sum_{\ell, r} \int_{t_0}^T \mathrm{d}X_{\ell r}(t)  \right\|_F^2\\
    &\stackrel{\textrm{Eq.}  \eqref{eq_t1}}{\leq}  \left\| \int_{t_0}^T \sum_{\ell, r} \sum_{j=1}^d\mathcal{R}(R_{(\ell r) (i j)}(t))\mathcal{R}(\mathrm{d}B_{(ij)}(t))\right\|_F^2  +  \left\|\int_{t_0}^T \sum_{\ell, r} \sum_{j=1}^d \mathcal{I}(R_{(\ell r) (i j)}(t)) \mathcal{R}(\mathrm{d}B_{(ij)}(t))\right\|_F^2\nonumber\\
&\qquad + \left\|\int_{t_0}^T \sum_{\ell, r} \sum_{j=1}^d \mathcal{R}(R_{(\ell r) (i j)}(t)) \mathcal{I}(\mathrm{d}B_{(ij)}(t))\right\|_F^2 + \left\|\int_{t_0}^T \sum_{\ell, r} \sum_{j=1}^d \mathcal{I}(R_{(\ell r) (i j)}(t)) \mathcal{I}(\mathrm{d}B_{(ij)}(t)) \right\|_F^2\nonumber\\
&\qquad+ \left\|\int_{t_0}^T \sum_{\ell, r} \sum_{j=1}^d\mathcal{R}(Q_{(\ell r) (i j)}(t))\mathcal{R}(\mathrm{d}B_{(ij)}^\ast(t))
\right\|_F^2 + \left\|\int_{t_0}^T \sum_{\ell, r} \sum_{j=1}^d \mathcal{I}(Q_{(\ell r) (i j)}(t)) \mathcal{R}(\mathrm{d}B_{(ij)}^\ast(t))\right\|_F^2\nonumber\\
&\qquad + \left\|\int_{t_0}^T \sum_{\ell, r} \sum_{j=1}^d\mathcal{R}(Q_{(\ell r) (i j)}(t)) \mathcal{I}(\mathrm{d}B_{(ij)}^\ast(t))\right\|_F^2 +\left\| \int_{t_0}^T \sum_{\ell, r} \sum_{j=1}^d\mathcal{I}(Q_{(\ell r) (i j)}(t)) \mathcal{I}(\mathrm{d}B_{(ij)}^\ast(t))\right\|_F^2.
\end{align}
Since all of the terms on the r.h.s. of \eqref{eq_t2} are entirely real or imaginary for all $t\geq 0$, we can apply  It\^o's Lemma (Lemma \ref{lemma_ito_lemma_new}) individually to each of these terms.
The proof to bound each of these eight terms is identical (if we replace $\mathcal{R}$ with $\mathcal{I}$, $R$ with $Q$, and/or $\mathrm{d}B_{(ij)}(t)$ with $\mathrm{d}B_{(ij)}^\ast(t)$)), since $\mathcal{R}(\mathrm{d}B_{(ij)}(t))$, $\mathcal{R}(\mathrm{d}B_{(ij)}^\ast(t)),$  $\mathcal{I}(\mathrm{d}B_{(ij)}(t))$, $\mathcal{I}(\mathrm{d}B_{(ij)}^\ast(t))$ are equal in distribution.
Thus, without loss of generality, we only present the proof of how to bound the term $\left\| \int_{t_0}^T \sum_{j=1}^d\mathcal{R}(R_{(\ell r) (i j)}(t))\mathcal{R}(\mathrm{d}B_{(ij)}(t))\right\|_F^2$.

Towards this end, define
$$Y(t):=  \int_{t_0}^t \sum_{\ell, r} \sum_{j=1}^d\mathcal{R}(R_{(\ell r) (i j)}(t))\mathcal{R}(\mathrm{d}B_{(ij)}(t)) \qquad \forall t\geq 0$$
 Then we have,
\begin{align}\label{eq_int_2b}
       &\left\| \int_{t_0}^T \sum_{\ell, r} \sum_{j=1}^d\mathcal{R}(R_{(\ell r) (i j)}(t))\mathcal{R}(\mathrm{d}B_{(ij)}(t))\right\|_F^2  \nonumber\\
   &=\mathbb{E}[f(Y(T)) - f(Y(t_0))] \nonumber\\
     &\stackrel{\textrm{It\^o's Lemma (Lemma \ref{lemma_ito_lemma_new})}}{=} \mathbb{E}\left [\frac{1}{2} \int_{t_0}^t \sum_{\ell, r} \sum_{\alpha, \beta} \left(\frac{\partial}{ \partial Y_{\alpha \beta}} f(Y(t))\right) \mathcal{R}(R_{(\ell r) (\alpha \beta)}(t)) \mathcal{R}(\mathrm{d}B_{\ell r}(t)) \right] \nonumber\\
     &\qquad  \qquad +\mathbb{E}\left [\frac{1}{2} \int_{t_0}^t \sum_{\ell, r} \sum_{i,j} \sum_{\alpha, \beta} \left(\frac{\partial^2}{\partial Y_{ij} \partial Y_{\alpha \beta}} f(Y(t))\right) \mathcal{R}(R_{(\ell r) (i j)}(t))  \mathcal{R}(R_{(\ell r) (\alpha \beta)}(t)) \mathrm{d}t \right] \nonumber\\
        &= 0 +\mathbb{E}\left [\frac{1}{2} \int_{t_0}^t \sum_{\ell, r} \sum_{i,j} \sum_{\alpha, \beta} \left(\frac{\partial^2}{\partial Y_{ij} \partial Y_{\alpha \beta}} f(Y(t))\right) \mathcal{R}(R_{(\ell r) (i j)}(t)) \mathcal{R}(R_{(\ell r) (\alpha \beta)}(t)) \mathrm{d}t \right],
        \end{align}
        where the third equality is It\^o's Lemma (Lemma \ref{lemma_ito_lemma_new}), and the last equality holds since $$\mathbb{E}\left[\int_{t_0}^T  \left(\frac{\partial}{ \partial Y_{\alpha \beta}} f(Y(t))\right) \mathcal{R}(R_{(\ell r) (\alpha \beta)}(t)) \mathcal{R}(\mathrm{d}B_{\ell r}(t)) \right] = 0,$$
        for each $\ell, r, \alpha, \beta \in [d]$ because $\mathrm{d}B_{\ell r}(s)$ is independent of both $Y(t)$ and $R(t)$ for all $s \geq t$ and the Brownian motion increments $\mathrm{d}B_{\alpha \beta}(s)$ satisfy $\mathbb{E}[\int_t^{\tau} \mathrm{d}B_{\alpha \beta}(s)] = \mathbb{E}[B_{\alpha \beta}(\tau) - B_{\alpha \beta}(t)]= 0$ for any $\tau \geq t$.

        Thus, plugging \eqref{eq_int_5} into \eqref{eq_int_2b}, we have 
   \begin{align}\label{eq_int_2b2}
        &\left\| \int_{t_0}^T\sum_{\ell, r} \sum_{j=1}^d\mathcal{R}(R_{(\ell r) (i j)}(t))\mathcal{R}(\mathrm{d}B_{(ij)}(t))\right\|_F^2  \nonumber\\
   &\stackrel{\textrm{Eq.} \eqref{eq_int_5}, \eqref{eq_int_2b2}}{=} \mathbb{E}\left [\frac{1}{2} \int_{t_0}^t \sum_{\ell, r} \sum_{i,j}  2  [\mathcal{R}(R_{(\ell r) (i j)}(t))]^2 \mathrm{d}t \right]\nonumber\\
                   &= \mathbb{E}\left [ \int_{t_0}^t \sum_{\ell, r} \sum_{i,j} \left(\left( \frac{ |\lambda_i(t) - \lambda_j(t)| }{\max(|\Delta_{ij}(t)|, \eta_{ij})^2} \mathcal{R}\left(u_i(t) u_j^\ast(t)\right)\right)[\ell, r]\right)^2 \mathrm{d}t \right]\nonumber\\
                                      &= \mathbb{E}\left [ \int_{t_0}^t \sum_{i,j} \sum_{\ell, r}  \left(\left( \frac{ |\lambda_i(t) - \lambda_j(t)| }{\max(|\Delta_{ij}(t)|, \eta_{ij})^2}\mathcal{R}\left(u_i(t) u_j^\ast(t)\right)\right)[\ell, r]\right)^2 \mathrm{d}t \right]\nonumber\\
                        &= \mathbb{E}\left [ \int_{t_0}^t   \sum_{i,j} \left \|\frac{ |\lambda_i(t) - \lambda_j(t)| }{\max(|\Delta_{ij}(t)|, \eta_{ij})} \mathcal{R}\left(u_i(t) u_j^\ast(t)\right)\right\|_F^2\mathrm{d}t \right]\nonumber\\
                              &= \mathbb{E}\left [ \int_{t_0}^t   \sum_{i,j} \left \|\mathcal{R}\left(\frac{ |\lambda_i(t) - \lambda_j(t)| }{\max(|\Delta_{ij}(t)|, \eta_{ij})} u_i(t) u_j^\ast(t)\right)\right\|_F^2\mathrm{d}t \right]\nonumber\\
                                                       &\leq \mathbb{E}\left [ \int_{t_0}^t   \sum_{i,j} \left \|\frac{ |\lambda_i(t) - \lambda_j(t)| }{\max(|\Delta_{ij}(t)|, \eta_{ij})} u_i(t) u_j^\ast(t)\right\|_F^2\mathrm{d}t \right]\nonumber\\
                &= 2\int_{t_0}^{T}   \mathbb{E}\left[ \sum_{i=1}^{d}  \sum_{j \neq i}  \frac{(\lambda_i(t) - \lambda_j(t))^2}{\max(\Delta^2_{ij}(t), \eta_{ij}^2)} \left\|u_i(t) u_j^\ast(t)
                \right\|_F^2  \mathrm{d}t\right] \nonumber\\
                             &\leq 4\int_{t_0}^{T}   \mathbb{E}\left[ \sum_{i=1}^{d}  \sum_{j \neq i}  \frac{(\lambda_i(t) - \lambda_j(t))^2}{\max(\Delta^2_{ij}(t), \eta_{ij}^2)} \mathrm{d}t\right],
          \end{align}
where the sixth equality holds because $\langle u_i(t) u_j^\ast(t) ,  u_\ell(t) u_h^\ast(t) \rangle = 0$ for all $(i,j) \neq (\ell,h)$, and the last equality holds because  $\|u_i(t) u_j^\ast(t) + u_j(t) u_i^\ast(t)\|_F^2 = 2$ for all $t$ with probability $1$.

Thus, plugging \eqref{eq_int_2b} into \eqref{eq_t2} (and recalling that, from the discussion after \eqref{eq_t2}, the bound we derive in \eqref{eq_int_2b} holds without loss of generality for all eight terms in \eqref{eq_t2}), we have that
\begin{align}\label{eq_int_2}
    &\left\|\int_{t_0}^{t}\sum_{i=1}^{d} \sum_{j \neq i} |\lambda_i(t) - \lambda_j(t)| \frac{1}{\max(|\Delta_{ij}(s)|, \eta_{ij})}(u_i(s) u_j^\ast(s)\mathrm{d}B_{ij}(s) + u_j(s) u_i^\ast(s)\mathrm{d}B_{ij}^\ast(s))\right\|_F^2\nonumber\\
    &=\|X(T) - X(t_0)\|_F^2 \nonumber\\
    &\leq 32\int_{t_0}^{T}   \mathbb{E}\left[ \sum_{i=1}^{d}  \sum_{j \neq i}  \frac{(\lambda_i(t) - \lambda_j(t))^2}{\max(\Delta^2_{ij}(t), \eta_{ij}^2)} \mathrm{d}t\right]
    \end{align}

\paragraph{Bounding the drift term:}

To bound the drift term in \eqref{eq_int_1}, we use the Cauchy-Shwarz inequality:
\begin{align}\label{eq_int_3}
     &\left\|\int_{t_0}^{T}\sum_{i=1}^{d} \sum_{j\neq i} (\lambda_i(t) - \lambda_j(t)) \frac{\mathrm{d}t}{\max(\Delta^2_{ij}(t), \eta_{ij}^2)} u_i(t) u_i^\ast(t) \right\|_F^2\nonumber \\
     & =      \left\|\int_{t_0}^{T}\sum_{i=1}^{d} \sum_{j\neq i}  \frac{\lambda_i(t) - \lambda_j(t)}{\max(\Delta^2_{ij}(t), \eta_{ij}^2)} u_i(t) u_i^\ast(t) \times 1 \mathrm{d}t \right\|_F^2 \nonumber\\ 
     &     \stackrel{\textrm{Cauchy-Schwarz inequality}}{\leq}     \int_{t_0}^{T}\left\|\sum_{i=1}^{d} \sum_{j\neq i}  \frac{\lambda_i(t) - \lambda_j(t)}{\max(\Delta^2_{ij}(t), \eta_{ij}^2)} u_i(t) u_i^\ast(t)\right\|_F^2 \mathrm{d}t\times \int_{t_0}^{T} 1^2 \mathrm{d}t \nonumber\\
     &=   T \int_{t_0}^{T}\left\|\sum_{i=1}^{d} \sum_{j\neq i} \frac{\lambda_i(t) - \lambda_j(t)}{\max(\Delta^2_{ij}(t), \eta_{ij}^2)} u_i(t) u_i^\ast(t) \right\|_F^2 \mathrm{d}t \nonumber \\
          &=   T \int_{t_0}^{T}\sum_{i=1}^{d} \left\|\sum_{j\neq i} \frac{\lambda_i(t) - \lambda_j(t)}{\max(\Delta^2_{ij}(t), \eta_{ij}^2)} u_i(t) u_i^\ast(t) \right\|_F^2 \mathrm{d}t \nonumber\\
                    &=   T \int_{t_0}^{T}\sum_{i=1}^{d} \left\|\left(\sum_{j\neq i} \frac{\lambda_i(t) - \lambda_j(t)}{\max(\Delta^2_{ij}(t), \eta_{ij}^2)}\right) u_i(t) u_i^\ast(t) \right\|_F^2 \mathrm{d}t \nonumber\\
                                        &=   T \int_{t_0}^{T}\sum_{i=1}^{d}\left(\sum_{j\neq i} \frac{\lambda_i(t) - \lambda_j(t)}{\max(\Delta^2_{ij}(t), \eta_{ij}^2)}\right)^2 \left\| u_i(t) u_i^\ast(t) \right\|_F^2 \mathrm{d}t \nonumber\\
                                                                                &=   T \int_{t_0}^{T}\sum_{i=1}^{d}\left(\sum_{j\neq i} \frac{\lambda_i(t) - \lambda_j(t)}{\max(\Delta^2_{ij}(t), \eta_{ij}^2)}\right)^2 \times 1 \mathrm{d}t,
\end{align}
where the first inequality is by the Cauchy-Schwarz inequality for integrals (applied to each entry of the matrix-valued integral).
The third equality holds since $\langle u_i(t) u_i^\ast(t) , u_j(t) u_j^\ast(t) \rangle = 0$ for all $i \neq j$.
The last equality holds since $\| u_i(t) u_i^\ast(t) \|_F^2=1$ with probability $1$.
Therefore, taking the expectation on both sides of \eqref{eq_int_1}, and plugging \eqref{eq_int_2} and  \eqref{eq_int_3} into  \eqref{eq_int_1}, we have
\begin{align} \label{eq_int_4}
\mathbb{E}\left[\left\|Z_\eta\left(T\right) -  Z_\eta(t_0)\right \|_F^2\right] &\leq  32\int_{t_0}^{T}   \mathbb{E}\left[ \sum_{i=1}^{d}  \sum_{j \neq i}  \frac{(\lambda_i(t) - \lambda_j(t))^2}{\max(\Delta^2_{ij}(t), \eta_{ij}^2)} \right]\mathrm{d}t \nonumber \\
    & +    T \int_{t_0}^{T}\mathbb{E}\left[\sum_{i=1}^{d}\left(\sum_{j\neq i} \frac{\lambda_i(t) - \lambda_j(t)}{\max(\Delta^2_{ij}(t), \eta_{ij}^2)}\right)^2\right] \mathrm{d}t .
\end{align}
\end{proof}

\subsection{Proof of Lemma \ref{Lemma_projection_differntial}} \label{sec_proof_Lemma_projection_differntial}

\begin{proof}[Proof of Lemma \ref{Lemma_projection_differntial}]
To compute the stochastic Ito derivative $\mathrm{d}(u_i(t) u_i^\ast(t))$ we apply the Dyson Brownian motion equations \eqref{eq_DBM_eigenvectors}.
For any $t \in [0,T]$, we have
\begin{align}\label{eq_eq_derivative1}
  &\mathrm{d}(u_i(t) u_i^\ast(t))
  =(u_i(t) + \mathrm{d}u_i(t))(u_i(t) + \mathrm{d}u_i(t))^\ast - u_i(t)u_i^\ast(t) \nonumber\\
  &= \left(u_i(t)+ \sum_{j \neq i} \frac{\mathrm{d}B_{ij}(t)}{\gamma_i(t) - \gamma_j(t)}u_j(t) - \sum_{j \neq i} \frac{\mathrm{d}t}{(\gamma_i(t)- \gamma_j(t))^2}u_i(t) \right) \nonumber\\
  &\qquad\times \left(u_i(t) + \sum_{j \neq i} \frac{\mathrm{d}B_{ij}(t)}{\gamma_i(t) - \gamma_j(t)}u_j(t) - \sum_{j \neq i} \frac{\mathrm{d}t}{(\gamma_i(t)- \gamma_j(t))^2}u_i(t) \right)^\ast - u_i(t)u_i^\ast(t) \nonumber\\
  &= u_i(t) u_i^\ast(t) + \sum_{j \neq i} \frac{1}{\gamma_i(t) - \gamma_j(t)}(u_i(t) u_j^\ast(t)\mathrm{d}B_{ij}(t) + u_j(t) u_i^\ast(t)\mathrm{d}B_{ij}^\ast(t))\nonumber\\
  &\qquad- \sum_{j\neq i} \frac{\mathrm{d}t}{(\gamma_i(t)- \gamma_j(t))^2} u_i(t) u_i^\ast(t)\nonumber\\
  & \qquad + \sum_{j \neq i}  \sum_{\ell \neq i} \frac{\mathrm{d}B_{ij}(t)\mathrm{d}B_{i\ell}^\ast(t)}{(\gamma_i(t) - \gamma_j(t)) (\gamma_i(t) - \gamma_\ell(t))} u_j(t)u_{\ell}^\ast(t) \nonumber\\
  & \qquad   \qquad  \qquad -\varphi_1(t)\varphi_2^\ast(t)  -\varphi_2(t)\varphi_1^\ast(t) + -\varphi_2(t)\varphi_2^\ast(t) - u_i(t)u_i^\ast(t),
  \end{align}
  where we define $\varphi_1(t):= \sum_{j \neq i} \frac{\mathrm{d}B_{ij}(t)}{\gamma_i(t) - \gamma_j(t)}u_j(t)$ and $\varphi_2(t):=\sum_{j \neq i} \frac{\mathrm{d}t}{(\gamma_i(t)- \gamma_j(t))^2}u_i(t)$.
The terms $\varphi_1(t) \varphi_2^\ast(t)$ and $\varphi_2(t) \varphi_1^\ast(t)$ have differentials $O(\mathrm{d}B_{ij} \mathrm{d}t)$, and $\varphi_2(t) \varphi_2^\ast(t)$ has differentials $O(\mathrm{d}t^2)$; thus, all three terms vanish in the stochastic derivative by Ito's Lemma \ref{lemma_ito_lemma_new} (applied separately to the real and imaginary parts of these terms).
Therefore, \eqref{eq_eq_derivative1} implies that the stochastic derivative $\mathrm{d}(u_i(t) u_i^\ast(t))$ satisfies
  \begin{align}\label{eq_eq_derivative2}
  &\mathrm{d}(u_i(t) u_i^\ast(t)) = \sum_{j \neq i} \frac{1}{\gamma_i(t) - \gamma_j(t)}(u_i(t) u_j^\ast(t) \mathrm{d}B_{ij}(t) + u_j(t) u_i^\ast(t) \mathrm{d}B_{ij}^\ast(t))\nonumber \\
  &\qquad - \sum_{j\neq i} \frac{\mathrm{d}t}{(\gamma_i(t)- \gamma_j(t))^2} u_i(t) u_i^\ast(t)\nonumber\\
  & \qquad + \sum_{j \neq i}  \sum_{\ell \neq i} \frac{\mathrm{d}B_{ij}(t)\mathrm{d}B_{i\ell}^\ast(t)}{(\gamma_i(t) - \gamma_j(t)) (\gamma_i(t) - \gamma_\ell(t))} u_j(t)u_{\ell}^\ast(t)\nonumber\\
    &= \sum_{j \neq i} \frac{1}{\gamma_i(t) - \gamma_j(t)}(u_i(t) u_j^\ast(t)\mathrm{d}B_{ij}(t) + u_j(t) u_i^\ast(t)\mathrm{d}B_{ij}^\ast(t)) - \sum_{j\neq i} \frac{\mathrm{d}t}{(\gamma_i(t)- \gamma_j(t))^2} u_i(t) u_i^\ast(t)\nonumber\\
  & \qquad + \sum_{j \neq i} \frac{\mathrm{d}B_{ij}(t) \mathrm{d}B_{ij}^\ast(t)}{(\gamma_i(t) - \gamma_j(t))^2} u_j(t)u_j^\ast(t)\nonumber\\
      &=  \sum_{j \neq i} \frac{1}{\gamma_i(t) - \gamma_j(t)}(u_i(t) u_j^\ast(t)\mathrm{d}B_{ij}(t) + u_j(t) u_i^\ast(t)\mathrm{d}B_{ij}^\ast(t)) - \sum_{j\neq i} \frac{\mathrm{d}t}{(\gamma_i(t)- \gamma_j(t))^2} u_i(t) u_i^\ast(t)\nonumber\\
  & \qquad + \sum_{j \neq i} \frac{\mathrm{d}t}{(\gamma_i(t) - \gamma_j(t))^2} u_j(t)u_j^\ast(t),
\end{align}
where the second-to-last equality holds since all terms $\mathrm{d}B_{ij}(t)\mathrm{d}B_{i\ell}^\ast(t)$ with $j \neq \ell$ in the sum $\sum_{j \neq i}  \sum_{\ell \neq i} \frac{\mathrm{d}B_{ij}(t)\mathrm{d}B_{i\ell}^\ast(t)}{(\gamma_i(t) - \gamma_j(t)) (\gamma_i(t) - \gamma_\ell(t))} u_j(t)u_{\ell}^\ast(t)$  vanish by Ito's Lemma \ref{lemma_ito_lemma_new} since they have mean 0 and are $O(\mathrm{d}B_{ij}(t) \mathrm{d}B_{i\ell}^\ast(t))$; we are therefore left only with the terms $j = \ell$ in the sum which have differential terms $\mathrm{d}B_{ij}(t)\mathrm{d}B_{ij}^\ast(t)$ which have mean $\mathrm{d}t$ plus higher-order terms which vanish by Ito's Lemma \ref{lemma_ito_lemma_new}. 
Therefore  \eqref{eq_eq_derivative2} implies that
\begin{align*}
    \mathrm{d}(u_i(t) u_i^\ast(t)) 
     &= \sum_{j \neq i} \frac{1}{\gamma_i(t) - \gamma_j(t)}(u_i(t) u_j^\ast(t)\mathrm{d}B_{ij}(t) + u_j(t) u_i^\ast(t)\mathrm{d}B_{ij}^\ast(t))\nonumber \\
    &\qquad - \sum_{j\neq i} \frac{\mathrm{d}t}{(\gamma_i(t)- \gamma_j(t))^2} (u_i(t) u_i^\ast(t) - u_j(t)u_j^\ast(t)).
\end{align*}
\end{proof}

\section{Eigengap-free utility bounds in a weaker Frobenius norm metric} \label{appendix_gap_free_bounds_in_weaker_metric}

The following steps can be used to extend our main result in Theorem \ref{thm_rank_k_covariance_approximation_new} to obtain eigengap-free utility bounds on the weaker Frobenius metric $\|\hat{M}_k - M\|_F^2 - \|M_k - M\|_F^2$:

\begin{enumerate}
\item {\bf Applying Ito's lemma to the weaker Frobenius norm metric:}

When bounding the stronger utility metric $\left\|\hat{M}_k -  M_k\right \|_F^2$ in the proof of Theorem \ref{thm_rank_k_covariance_approximation_new} we apply Ito's lemma to the function $f(Y) = \|Y\|_F^2$.  
 If we only wish to bound the weaker utility metric $\|\hat{M}_k - M\|_F^2 - \|M_k - M\|_F^2$, we can instead apply Ito's Lemma to the function $g(Y) := \|Y - M\|_F^2$.
 Then we have
\begin{eqnarray*} \|\hat{M}_k - M\|_F^2 - \|M_k - M\|_F^2  &=& g(\Psi(T)) - g(\Psi(0)) = \int_0^T g(\Psi(t)) \mathrm{d}t\\
&=& \int_0^T  \frac{1}{2} (\mathrm{d}\Psi(t))^\top  \nabla^2 g(Y) \mathrm{d}\Psi(t)   + (\nabla g(Y))^\top \mathrm{d}\Psi(t) \, \, \mathrm{d}t,
\end{eqnarray*}
where the last equality is by Ito's lemma,
and where
\begin{equation} \label{eq_W2}
\nabla g(Y) [ij] = \frac{\partial}{\partial Y_{ij}} g(Y) = 2Y_{ij} -2M_{ij}, \qquad \textrm{ and }
\end{equation}
 \begin{equation*}
 \nabla^2 g(Y)[ij, \alpha \beta] = \frac{\partial}{\partial Y_{ij} \partial Y_{\alpha \beta}} g(Y) = \begin{cases} 2 & \textrm{ for } (i,j)=(\alpha, \beta)\\ 0 & \textrm{otherwise}. \end{cases}.
 \end{equation*}
 
\item {\bf Canceling the eigengap terms:} The extra term $-2M_{ij}$ in the first derivative \eqref{eq_W2} leads to cancellations of the terms in the utility bound which depend on the eigenvalue gap.
 More specifically, writing $M = \sum_{i=1}^d \gamma_i(0) u_i(0) u_i^\ast(0)$,  we have
 \begin{eqnarray} \label{eq_W1}
 &\mathbb{E}&\left[\int_0^T g(\Psi(t)) \mathrm{d}t\right] = \mathbb{E}\left[\int_0^T  \frac{1}{2} (\mathrm{d}\Psi(t))^\top  \nabla^2 g(Y) \mathrm{d}\Psi(t)   + (\nabla g(Y))^\top \mathrm{d}\Psi(t) \, \, \mathrm{d}t  \right ] \nonumber\\
 &=& \mathbb{E}\bigg[ \int_{t_0}^{T}   \sum_{i=1}^{d}  \sum_{j \neq i}  \frac{(\lambda_i(t) - \lambda_j(t))^2}{(\gamma_i(t) - \gamma_j(t))^2}  \mathrm{d}t \nonumber\\
&&\qquad - 2\sum_{i=1}^{d} \sum_{j\neq i} \frac{\lambda_i(t) - \lambda_j(t)}{(\gamma_i(t) - \gamma_j(t))^2} \times \sum_{\ell=1}^k \gamma_\ell(t)  \left \langle u_i(t) u_i^\ast(t), \, \, u_{\ell}(t) u_{\ell}^\ast(t) \right \rangle \mathrm{d}t \nonumber\\ 
 && \qquad  +  2\sum_{i=1}^{d} \sum_{j\neq i} \frac{\lambda_i(t) - \lambda_j(t)}{(\gamma_i(t) - \gamma_j(t))^2} \times \sum_{\ell=1}^d \gamma_\ell(0)  \left \langle u_i(0) u_i^\ast(0), \, \, u_{\ell}(t) u_{\ell}^\ast(t) \right \rangle \mathrm{d}t \bigg] \nonumber\\
&& +  \mathbb{E}\left[\int_0^T  (\nabla g(Y))^\top \sum_{i=1}^d  (\mathrm{d}\lambda_i(t)) (u_i(t) u_i^\ast(t)) \, \, \mathrm{d}t  \right ],
  \end{eqnarray}
  where the second equality holds since $\mathbb{E}[M_{ij}\mathrm{d}B_{ij}] = 0$ because $M_{ij}$ is a constant.

  The term $\mathbb{E}\left[\int_0^T  (\nabla g(Y) )^\top \sum_{i=1}^d  (\mathrm{d}\lambda_i(t)) (u_i(t) u_i^\ast(t)) \, \, \mathrm{d}t  \right ] = \tilde{O}(kd)$ in \eqref{eq_W1} can be bounded using the same steps as \eqref{eq_a1}.  Thus, we have

 \begin{eqnarray*}
 &\mathbb{E}&\left[\int_0^T g(\Psi(t)) \mathrm{d}t\right] \\
 &=& \mathbb{E}\bigg[ \int_{t_0}^{T} \bigg(   \sum_{i=1}^{d}  \sum_{j \neq i}  \frac{(\lambda_i(t) - \lambda_j(t))^2}{(\gamma_i(t) - \gamma_j(t))^2} \\
 &-& 2\sum_{i=1}^{d} \sum_{j\neq i} \frac{\lambda_i(t) - \lambda_j(t)}{(\gamma_i(t) - \gamma_j(t))^2} \times \sum_{\ell=1}^d \gamma_\ell(t)  \left \langle u_{\ell}(t) u_{\ell}^\ast(t) , \, \,   u_i(t) u_i^\ast(t) \right \rangle\\ 
 && \qquad +  2\sum_{i=1}^{d} \sum_{j\neq i} \frac{\lambda_i(t) - \lambda_j(t)}{(\gamma_i(t) - \gamma_j(t))^2} \times \sum_{\ell=1}^d \gamma_\ell(0)  \left \langle u_{\ell}(0) u_{\ell}^\ast(0), \, \, u_i(t) u_i^\ast(t) \right \rangle \bigg ) \mathrm{d}t \bigg] + \tilde{O}(kd)\\
 &=& \mathbb{E}\bigg[ \int_{t_0}^{T}  \bigg( \sum_{i=1}^{d}  \sum_{j \neq i}  \frac{(\lambda_i(t) - \lambda_j(t))^2}{(\gamma_i(t) - \gamma_j(t))^2}     - 2\sum_{i=k+1}^{d} \sum_{j\neq i} \frac{(\lambda_j(t)-\lambda_i(t)) \times \gamma_i(t)  }{(\gamma_i(t) - \gamma_j(t))^2}\\
&& \qquad +   2\sum_{i=k+1}^{d} \sum_{j\neq i} \frac{(\lambda_j(t)-\lambda_i(t)) \times \gamma_i(t)  }{(\gamma_i(t) - \gamma_j(t))^2} \bigg) \mathrm{d}t\bigg ] + \mathbb{E}\left[\int_0^T \mathcal{H}(t)\mathrm{d}t\right] + \tilde{O}(kd)\\
 &=& \mathbb{E}\bigg[  \int_{t_0}^{T}   \sum_{i=1}^{k}  \sum_{j \neq i}  \frac{(\gamma_i(t) - \gamma_j(t))^2}{(\gamma_i(t) - \gamma_j(t))^2}  \mathrm{d}t \bigg]  + \mathbb{E}\left[\int_0^T \mathcal{H}(t)\mathrm{d}t\right] + \tilde{O}(kd)\\
 &=& \tilde{O}(kd),
 \end{eqnarray*}
 where the second equality is obtained by making small-$t$ approximations $\gamma_i(t) \approx \gamma_i(0)$ and $u_i(0) \approx u_i(t)$, and  $\mathcal{H}(t)$ are the higher-order terms which remain after making these approximations.

\item {\bf Bounding the higher-order terms:} 
More specifically, the higher-order terms are
 \begin{eqnarray} \label{eq_W3}
&&\mathcal{H}(t) =\sum_{i=1}^{d} \sum_{j\neq i} \frac{\lambda_i(t) - \lambda_j(t)}{(\gamma_i(t) - \gamma_j(t))^2} \times \sum_{\ell=1}^d   (\gamma_\ell(0) - \gamma_\ell(t) ) \left \langle  u_{\ell}(t) u_{\ell}^\ast(t), \, \, u_i(t) u_i^\ast(t) \right \rangle \mathrm{d}t\nonumber\\
&+& 2\sum_{i=1}^{d} \sum_{j\neq i} \frac{\lambda_i(t) - \lambda_j(t)}{(\gamma_i(t) - \gamma_j(t))^2} \times \sum_{\ell=1}^d   \gamma_\ell(0) \left \langle u_{\ell}(0) u_{\ell}^\ast(0) -  u_{\ell}(t) u_{\ell}^\ast(t), \, \, u_i(t) u_i^\ast(t) \right \rangle \mathrm{d}t\nonumber\\
&= &2\sum_{i=1}^{d} \sum_{j\neq i} \frac{\lambda_i(t) - \lambda_j(t)}{(\gamma_i(t) - \gamma_j(t))^2} \times  (\gamma_i (0) - \gamma_i(t) ) \mathrm{d}t \nonumber\\
& +&2\sum_{i=1}^{d} \sum_{j\neq i} \frac{\lambda_i(t) - \lambda_j(t)}{(\gamma_i(t) - \gamma_j(t))^2} \times \sum_{\ell=1}^d   \gamma_\ell(0) \left \langle u_{\ell}(0) u_{\ell}^\ast(0) -  u_{\ell}(t) u_{\ell}^\ast(t), \, \, u_i(t) u_i^\ast(t) \right \rangle \mathrm{d}t, \qquad \qquad
\end{eqnarray}
where the first equality holds since $ \langle  u_i(t) u_i^\ast(t), \, \, u_{\ell}(t) u_{\ell}^\ast(t)  \rangle = 0$ for $\ell \neq i$, and $ \langle  u_i(t) u_i^\ast(t),$ $\, \, u_i(t) u_i^\ast(t)  \rangle = 1$.

 To bound the first term on the r.h.s. of \eqref{eq_W3}, we use the fact that the two-time joint distribution of Dyson Brownian motion, $f(\gamma(0), \gamma(t))$, is symmetric in the sense that it depends only on the quantities $\{|\gamma_i(t) - \gamma_j(0)|\}_{1\leq i,j \leq d}$ (see e.g. \cite{tao2012topics}), which implies that $\mathbb{E}[\sum_{i=1}^{d} \sum_{j\neq i} \frac{\lambda_i(t) - \lambda_j(t)}{(\gamma_i(t) - \gamma_j(t))^2} \times  (\gamma_i (0) - \gamma_i(t) ) \mathrm{d}t]=0$.
The second term can be bounded in a similar manner.

After bounding these higher-order terms, one gets the eigenvalue gap-free bound 
 \begin{equation*}
 \mathbb{E}[\|\hat{M}_k - M\|_F - \|M_k - M\|_F] \leq \sqrt{\mathbb{E}[\|\hat{M}_k - M\|_F^2 - \|M_k - M\|_F^2]} \leq  \tilde{O}(\sqrt{kd}),
 \end{equation*}
 since $\|\hat{M}_k - M\|_F \geq \|M_k - M\|_F \geq 0$ and since $(a-b)^2 =a^2+b^2-2ab \leq a^2 - b^2$ for any $a\geq b \geq 0$.

\end{enumerate}

\end{document}